\documentclass[table,11pt,a4paper]{article}
\usepackage{amsthm,amsmath,amssymb,fullpage,authblk,xr}
\usepackage{multirow,mathtools,booktabs,enumerate,hyperref}
\usepackage{epsfig,amsfonts}
\usepackage{bbold,dsfont} 
\usepackage{xcolor, colortbl}
\definecolor{Gray}{gray}{0.82}
\definecolor{Greenish}{rgb}{0.8,1,.7}
\usepackage{hhline}

\usepackage{dblfloatfix}

\usepackage{setspace}
\newtheorem{theorem}{Theorem}[section] 

\newtheorem{lemma}[theorem]{Lemma}

\newtheorem{claim}[theorem]{Claim}
\newtheorem{corollary}[theorem]{Corollary}

\newtheorem{proposition}[theorem]{Proposition}

\newtheorem{definition}[theorem]{Definition}

\renewcommand{\deg}{d} 

\renewcommand{\Pr}[1]{\mathds{P}\left[\,#1\,\right]}
\newcommand{\Pru}[2]{\mathds{P}_{#1}\!\left[#2\right]}
\newcommand{\Ex}[1]{\mathds{E} \left[\,#1\,\right]}
\newcommand{\Exu}[2]{\mathds{E}_{#1} \left[\,#2\,\right]}

\newcommand{\comment}[1]{}
\newcommand{\BO}[1]{\mathcal{O}\!\left(#1\right)} 
 
\newcommand{\BT}[1]{\Theta\!\left(#1\right)}
\newcommand{\lo}[1]{o\!\left(#1\right)}  
\newcommand{\BOhm}[1]{\Omega\!\left(#1\right)}

\newcommand{\eps}{\varepsilon}

\allowdisplaybreaks

\renewcommand{\hat}{\widehat} 
\newcommand{\dmax}{d_{\mathsf{max}}}
\newcommand{\dist}{\operatorname{dist}}

\newcommand{\dmin}{d_{\mathsf{min}}}
\newenvironment{poc}{\begin{proof}[Proof of Claim.]}{\end{proof}}
\newenvironment{pocd}[1]{\begin{proof}[Proof of {C}laim #1.]}{\end{proof}}

\newcommand{\tsep}{t_{\mathsf{sep}}}
\newcommand{\tmix}{t_{\mathsf{mix}}}
\newcommand{\thit}{t_{\mathsf{hit}}}
\newcommand{\trel}{t_{\mathsf{rel}}}
\newcommand{\tcov}{t_{\mathsf{cov}}}

\newcommand{\geo}[1]{\operatorname{Geo}\!\left(#1\right)}

\newcommand{\bin}[2]{\operatorname{Bin}\!\left( #1,#2 \right)}

\newcommand{\E}{\mathds{E}}
\newcommand{\bigO}{\mathcal O}
\newcommand{\ind}[1]{\mbox{{\large 1}} _{#1}}

\newcommand{\mix}{\mathsf{mix}}

\newcommand{\cov}{\mathsf{cov}}

\usepackage{setspace,xcolor}

\title{\textbf{Multiple Random Walks on Graphs: \\ Mixing Few to Cover Many}\thanks{An extended abstract of this paper has appeared at ICALP 2021 \cite{ICALPus}}} 

\author[1]{Nicol\'{a}s Rivera}
\author[2]{Thomas Sauerwald}
\author[3]{John Sylvester}
\affil[1]{\small Universidad de Valpara\'iso, Vapara\'iso, Chile\\
	\texttt{nicolas.rivera@uv.cl}}

\affil[2]{\small University of Cambridge, Cambridge, United Kingdom\\
	\texttt{thomas.sauerwald@cl.cam.ac.uk}}
	 \affil[3]{\small University of Liverpool, Liverpool, United Kingdom\\
	\texttt{john.sylvester@liverpool.ac.uk}}

\date{}
\begin{document}
	\maketitle
	\begin{abstract}
			Random walks on graphs are an essential primitive for many randomised algorithms and stochastic processes. It is natural to ask how much can be gained by running $k$ multiple random walks independently and in parallel. 
	Although the cover time of multiple walks has been investigated for many natural networks, the problem of finding a general characterisation of multiple cover times for {\em worst-case} start vertices (posed by Alon, Avin, Kouck\'y, Kozma, Lotker, and Tuttle~in 2008) remains an open problem.
	
	First, we improve and tighten various bounds on the {\em stationary} cover time when $k$ random walks start from vertices sampled from the stationary distribution. For example, we prove an unconditional lower bound of $\Omega((n/k) \log n)$ on the stationary cover time, holding for any $n$-vertex graph $G$ and any $1 \leq k =o(n\log n )$. Secondly, we establish the {\em stationary} cover times of multiple walks on several fundamental networks up to constant factors. Thirdly, we present a framework characterising {\em worst-case} cover times in terms of {\em stationary} cover times and a novel, relaxed notion of mixing time for multiple walks called the {\em partial mixing time}. Roughly speaking, the partial mixing time only requires a specific portion of all random walks to be mixed. Using these new concepts, we can establish (or recover) the {\em worst-case} cover times for many networks including expanders, preferential attachment graphs, grids, binary trees and hypercubes. 
	\end{abstract}
	\section{Introduction}
	
	A random walk on a graph is a stochastic process that at each time step chooses a neighbour of the current vertex as its next state. The fact that a random walk visits every vertex of a connected, undirected graph in polynomial time was first used to solve the undirected $s-t$ connectivity problem in logarithmic space~\cite{AKLLR}. Since then random walks have become a fundamental primitive in the design of randomised algorithms which feature in approximation algorithms and sampling~\cite{lovasz1993random,SNPT13}, load balancing \cite{load0,load1}, searching \cite{search1,search0}, resource location \cite{gossip}, property testing \cite{Testing,Minors,Minors2}, graph parameter estimation \cite{esti,EdgesTwitter} and biological applications \cite{Clementi0GN21,Korman18,Korman20}. 
	
	The fact that random walks are local and memoryless (Markov property) ensures they require very little space and are relatively unaffected by changes in the environment, e.g., dynamically evolving graphs or graphs with edge failures. These properties make random walks a natural candidate for parallelisation, where running parallel walks has the potential of lower time overheads. One early instance of this idea are
	space-time trade-offs for the undirected $s-t$ connectivity problem \cite{BroderSpaceTime,Feige97}. Other applications involving multiple random walks are sublinear algorithms \cite{CS10}, local clustering \cite{ACL06,ST13} or epidemic processes on networks \cite{LLM12,PPPU11}.
	
	Given the potential applications of multiple random walks in algorithms, it is important to understand fundamental properties of multiple random walks. The \emph{speed-up}, first introduced in \cite{Multi}, is the ratio of the worst-case cover time by a single random walk to the cover time of $k$ parallel walks. Following \cite{Multi} and subsequent works~\cite{ER09,ElsSau,IKPS17,KKPS,PatelCB16,S10} our understanding of when and why a speed-up is present has improved. In particular, various results in \cite{Multi,ER09,ElsSau} establish that as long as the lengths of the walks are not smaller than the mixing, the speed-up is linear in $k$. However, there are still many challenging open problems, for example, understanding the effect of different start vertices or characterising the magnitude of speed-up in terms of graph properties, a problem already stated in \cite{Multi}: ``\textit{...which leads us to wonder whether there is some other property of a graph that characterises the speed-up achieved by multiple random walks more crisply than hitting and mixing times}.'' Addressing the previous questions, we introduce new quantities and couplings for multiple random walks, that allow us to improve the state-of-the-art by refining, strengthening or extending results from previous works.

	While there is an extensive body of research on the foundations of (single) random walks (and Markov chains), it seems surprisingly hard to transfer these results and develop a systematic theory of \emph{multiple} random walks. One of the reasons is that processes involving multiple random walks often lead to questions about \emph{short} random walks, e.g., shorter than the  mixing time.
	Such short walks may also arise in applications including generating random walk samples in massively parallel systems \cite{LMOS20,SNPT13}, or in applications where random walk steps are expensive or subject to delays (e.g., when crawling social networks like Twitter \cite{EdgesTwitter}). The challenge
	of analysing \emph{short} random walks (shorter than mixing or hitting time) has been mentioned not only in the area of multiple cover times~(e.g., \cite[Sec.~6]{ER09}), but also in the context of concentration inequalities for random walks~\cite[p.~863]{Lez89} and property testing~\cite{CS10}.

	\subsection{Our Contribution}

	Our first set of results provide several tight bounds on $t_{\mathsf{cov}}^{(k)}(\pi)$ in general (connected) graphs, where $t_{\mathsf{cov}}^{(k)}(\pi)$ is the expected time for each vertex to be visited by at least one of $k$ independent walks each started from a vertex independently sampled  from the stationary distribution $\pi$.
	
	The main findings of Section \ref{ReturnSec} include:
	\begin{itemize}
		\item Proving general bounds of $\mathcal{O} \Bigl( \bigl(\frac{|E|}{k\dmin}\bigr)^2\log^2 n \Bigr)$, $\mathcal{O} \Bigl(  \frac{\max_{v\in V}\Exu{\pi}{\tau_v}}{k} \log n \Bigr)$ and $\BO{\frac{|E|\log n}{k\dmin \sqrt{1-\lambda_2} }}$ on $\tcov^{(k)}(\pi)$, where $\dmin$ is the minimum degree, $\Exu{\pi}{\tau_v}$ is the hitting time of $v\in V$ from a stationary start vertex and $\lambda_2$ is the second largest eigenvalue of the transition matrix of the walk. All three bounds are tight for certain graphs. The first bound improves over \cite{BroderSpaceTime}, the second result is a Matthew's type bound for multiple random walks, and the third yields tight bounds for non-regular expanders such as preferential attachment graphs.
		
		\item We prove that for any graph $G$ and $1\leq k=  o\!\left( n\log n\right)$,  $\tcov^{(k)}(\pi)=\Omega((n/k) \log n)$. Weaker versions of this bound were obtained in \cite{ElsSau}, holding only for certain values of $k$ or under additional assumptions on the mixing time. Our result matches the fundamental $\Omega(n \log n)$ lower bound for single random walks ($k=1$)~\cite{aldouslower,FeigeLow}, and generalises it in the sense that total amount of work by all $k$ stationary walks together for covering is always $\Omega(n \log n).$   We establish the $\Omega((n/k) \log n)$ bound by reducing the multiple walk process to a single, reversible Markov chain, and applying a general lower bound on stationary cover times~\cite{aldouslower}. 
		\item A technical tool that provides a bound on the lower tail of the cover time by $k$ walks from stationary for graphs with a large and  (relatively) symmetric set of hard to hit vertices (Lemma \ref{lemma:generalLowerboundpi}). When applied to the $2$d torus and binary trees this yields a tight lower bound.
	\end{itemize}

	In Section \ref{Sec:partialmixing} we introduce a novel quantity for multiple walks we call \emph{partial mixing}. Intuitively, instead of mixing all (or at least half) of the $k$ walks, we only need to mix a specified number $\tilde{k}$ of them. We put this idea on a more formal footing and prove $\min$-$\max$ theorems which relate worst-case cover times $t_{\mathsf{cov}}^{(k)}$ to partial mixing times $t_{\mathsf{mix}}^{(\tilde{k},k)}$ and stationary cover times:
	\begin{itemize}
		\item For any graph $G$ and any $1 \leq k \leq n$, we prove that: 
		\[
		t_{\mathsf{cov}}^{(k)} \leq 12 \cdot \min_{1 \leq \tilde{k} < k} \max \left(  t_{\mathsf{mix}}^{(\tilde{k},k)} , \tcov^{(\tilde{k})}(\pi)  \right).
		\]
	\end{itemize}
	For now, we omit details such as the definition of the partial mixing time $t_{\mathsf{mix}}^{(\tilde{k},k)} $ as well as some $\min$-$\max$ characterisations that serve as lower bounds  (these can be found in Section~\ref{Sec:partialmixing}). Intuitively these characterisations suggest that for any number of walks $k$, there is an ``optimal'' choice of $\tilde{k}$ so that one first waits until $\tilde{k}$ out of the $k$ walks are mixed, and then considers only these $\tilde{k}$ stationary walks when covering the remainder of the graph. 
	
	This argument involving mixing only some walks extends and generalises prior results that involve mixing all (or at least a constant portion) of the $k$ walks~\cite{Multi,ER09,ElsSau}. Previous approaches only imply a linear speed-up as long as the lengths of the walks are not shorter than the mixing time of a \emph{single} random walk. In contrast, our characterisation may still yield tight bounds on the cover time for random walks that are much shorter than the mixing time. 
	
	In Section~\ref{fundamental} we demonstrate how our insights can be used on several well-known graph classes. As a first step, we determine their stationary cover times; this is based on our bounds from Section~\ref{ReturnSec}. Secondly, we derive lower and upper bounds on the partial mixing times. Finally, with the stationary cover times and partial mixing times at hand, we can apply the characterisations from Section~\ref{Sec:partialmixing} to infer lower and upper bounds on the worst-case stationary times. For some of those graphs the worst-case cover times were already known before, while for, e.g., binary trees and preferential attachment graphs, our bounds are new.
	
	\begin{itemize}
		\item For the graph families of binary trees, cycles, $d$-dim. tori ($d=2$ and $d\geq3 $), hypercubes, cliques, and (possibly non-regular) expanders we determine the cover time up to constants, for both worst-case and stationary start vertices (see Table \ref{tbl:results} for the quantitative results).
	\end{itemize}
	We believe that this new methodology constitutes some progress towards the open question of Alon, Avin, Kouck\'y, Kozma, Lotker, \& Tuttle~\cite{Multi} concerning a characterisation of worst-case cover times.

	\subsection{Novelty of Our Techniques}  
	
	While a lot of the proof techniques in previous work \cite{Multi,ER09,ElsSau,S10} are based on direct arguments such as mixing time (or relaxation time), our work introduces a number of new methods which, to the best of our knowledge, have not been used in the analysis of cover time of multiple walks before. In particular, one important novel concept is the introduction of the so-called \emph{partial mixing time}. The idea is that instead of waiting for all (or a constant portion of) $k$  walks to mix, we can just mix some $\tilde{k}\leq k$ walks to reap the benefits of coupling these $\tilde{k}$ walks to stationary walks. This then presents a delicate balancing act where one must find an optimal $\tilde{k}$ minimising the overall bound on the cover time, for example in expanders the optimal $\tilde{k}$ is linear in $k$ whereas in binary trees it is approximately $\sqrt{k}$, and for the cycle it is roughly $\log k$. This turning point reveals something about the structure of the graph and our results relating partial mixing to hitting time of sets helps one find this. Another tool we frequently use is a reduction to random walks with geometric resets, similar to a PageRank process, which allows us to relate multiple  walks from stationary to a single reversible Markov chain.

	\bgroup
	\def\arraystretch{1.7}
	\begin{table}
		\small
		\centering
		\rule{0pt}{4ex}    
		\begin{tabular}{|l|c|c|c|c|c|}
			\hline 
			Graph & $\text{Cover} $ & $\text{Hitting} $ & $\text{Mixing}$ &\multicolumn{2}{c|}{$\text{$k$-Cover Time, where }2\leq k\leq n$} \\\cline{5-6} 
			family& $\displaystyle t_{\mathsf{cov}}$ & $t_{\mathsf{hit}}$ & $t_{\mathsf{mix}}$ & worst-case $\tcov^{(k)}$& From $\pi^k$, $\tcov^{(k)}(\pi)$  \\ \hline \hline \rule{0pt}{4ex}   
			
						Cycles & $n^2$ & $n^2$ & $n^2$ & \cellcolor{Gray}$\displaystyle \frac{n^2}{\log k}$& \cellcolor{Greenish}$\displaystyle \left(\frac{n}{k}\right)^2\log^2 k$ \\	  
			\hhline{|-|-|-|-|-|-|} \rule{0pt}{2.5ex} 
			
			& \multirow{4}{*}{$\displaystyle n \log^2 n $} &\multirow{4}{*}{$\displaystyle n\log n$} & \multirow{4}{*}{$\displaystyle n$} &\cellcolor{Greenish}$\displaystyle  (n/k) \log^ 2 n$ & \cellcolor{Greenish}\\ 
			Binary	& & & & \textbf{if} \cellcolor{Greenish}$\displaystyle k \leq \log^2 n$. &\cellcolor{Greenish} \\  \cline{5-5}
			trees & & & & $ \cellcolor{Greenish}\displaystyle  (n/\sqrt{k}) \log n  $  &\cellcolor{Greenish}  \\ 
			& & & &\cellcolor{Greenish}\textbf{if} $\displaystyle k\geq \log^2 n $. & \multirow{-4}{*}{$\displaystyle \frac{n\log n}{k} \log\left(\frac{n\log n}{k}\right)$\cellcolor{Greenish}} \\ 		\hhline{|-|-|-|-|-|-|}\rule{0pt}{4ex}

			& \multirow{4}{*}{$n\log^2 n$}& \multirow{4}{*}{$n\log n$} &\multirow{4}{*}{$n$} & \cellcolor{Gray} $ (n/k) \log^ 2 n$ & \cellcolor{Greenish}  \\
			2-Dim.\ &&&& \textbf{if} $k \leq \log^2 n$.&\cellcolor{Greenish} \\  
			Tori &&&&       $ \cellcolor{Gray} \displaystyle  \frac{ n}{\log (k/\log^2 n)} $ & \cellcolor{Greenish}\\ 
			&&&&\textbf{if} $k\geq  \log^2 n  $. &\multirow{-4}{*}{$ \displaystyle \frac{n\log n}{k} \log\left(\frac{n\log n}{k}\right)$ \cellcolor{Greenish}}
			
			\\		\hhline{|-|-|-|-|-|-|}\rule{0pt}{2.5ex}
			
			$\!\!d$-Dim.\  & \multirow{4}{*}{$n \log n$} & \multirow{4}{*}{$n$} & \multirow{4}{*}{$n^{2/d}$} &\cellcolor{Gray} $ (n/k) \log n $& \cellcolor{Greenish}\\
			Tori & & &  & \cellcolor{Gray}\textbf{if} $ k \leq n^{1-2/d} \log n$. &\cellcolor{Greenish} \\\cline{5-5}
			$d\geq 3$ & & & &\cellcolor{Gray} $\displaystyle\frac{n^{2/d}  }{\log (k/(n^{1-2/d} \log n))} $&\cellcolor{Greenish} \\ 
			& & & & \cellcolor{Gray}\textbf{if} $k\geq n^{1-2/d} \log n$.   &\multirow{-4}{*}{$\displaystyle \frac{n}{k} \log n$\cellcolor{Greenish} } \\
			\hline\rule{0pt}{2.5ex}
			\multirow{4}{*}{$\!\!$Hypercubes} & \multirow{4}{*}{$n \log n $ }&\multirow{4}{*}{$ n$} & \multirow{4}{*}{$\log n \log\log n$ }& \cellcolor{Gray}  $ (n/k) \log n $  &\cellcolor{Greenish}\\ 
			& & & & \cellcolor{Gray} \textbf{if} $ k \leq n/\log \log n$. &\cellcolor{Greenish}\\ 
			& & & & \cellcolor{Gray} $ \log n \log \log n $ &\cellcolor{Greenish} \\ 
			& & & &\cellcolor{Gray} \textbf{if} $k\geq n / \log \log n $. &\multirow{-4}{*}{$ \displaystyle\frac{n}{k} \log n$ \cellcolor{Greenish}} \\ 
			\hhline{|-|-|-|-|-|-|}

			Expanders & $n \log n$ & $n$ & $\BO{\log n}$ &\cellcolor{Gray}$\displaystyle \phantom{\bigg|}\frac{n}{k} \log n$ & \cellcolor{Greenish}$\displaystyle \phantom{\bigg|}\frac{n}{k} \log n$   \\\hhline{------}
			PA, $m\geq 2$ & $n \log n$ & $n$ & $\BO{\log n}$ & \cellcolor{Greenish}$\displaystyle \phantom{\bigg|}\frac{n}{k} \log n$ & \cellcolor{Greenish}$\displaystyle \phantom{\bigg|}\frac{n}{k} \log n$   \\\hline \rule{0pt}{4ex}
			Barbells & $n^2$ & $n^2 $ & $n^2$ & $n^2/k$&   $\displaystyle \phantom{\bigg|} \frac{2^{-k}n^2}{k}+ \frac{n\log n}{k} $ \\	  
			\hline
		\end{tabular}\caption{All results above are  $\Theta( \cdot)$, that is bounded above and below by a multiplicative constant, apart from the mixing time of expanders which is only bounded from above. PA above is the preferential attachment process where each vertex has $m$ initial links, the results hold w.h.p., see \cite{CoopPA,MPS06}. Cells shaded in \colorbox{Greenish}{Green} are new results proved in this paper with the exception that for $k= \BO{\log n }$ upper bounds on the stationary cover time for binary trees, expanders and preferential attachment graphs can be deduced from general bounds for the worst-case cover time in \cite{Multi}. Cells shaded \colorbox{Gray}{Gray} in the second to last column are known results we re-prove in this paper using our partial mixing time results, for the $2$-dim grid we only re-prove upper bounds. References for the second to last column are given in Section \ref{fundamental}, except for the barbell, see \cite[Page 2]{ER09}. The barbell consists of two cliques on $n/2$ vertices connected by single edge; we include this in the table as an interesting example where the speed-up by stationary walks is exponential in $k$. All other results for single walks can be found in \cite{aldousfill}, for example.}
		\label{tbl:results}
	\end{table}
	\egroup

	\section{Notation \& Preliminaries}\label{sec:notation} 
	Throughout $G=(V,E)$ will be a finite undirected, connected graph with $n=|V|$ vertices and $m=|E|$ edges. For $v\in V(G)$ let $\deg(v)=|\{u\in V : uv\in E(G)\}|$ denote the degree of $v$ and $\dmin= \min_{v \in V}\deg(v)$ and $\dmax= \max_{v \in V}\deg(v)$ denote the minimum and maximum degrees in $G$ respectively. For any $ k \geq 1$, let $X_t=\bigl(X_t^{(1)},\dots, X_t^{(k)} \bigr)$ be the \textit{multiple random walk} process where each $X_{t}^{(i)}$ is an independent random walk on $G$. Let \[\Exu{u_1, \dots, u_k}{\cdot } = \Ex{\cdot  \mid X_0= (u_1, \dots, u_k)}\] denote the conditional expectation where, for each $1\leq i\leq k$, $X_{0}^{(i)} =u_i\in V$ is the start vertex of the $i$\textsuperscript{th} walk. Unless mentioned otherwise, walks will be \emph{lazy}, i.e., at each step the walk stays at its current location with probability $1/2$, and otherwise moves to a neighbour chosen uniformly at random. We let the random variable $\tau_{\mathsf{cov}}^{(k)}(G)=\inf\{t : \bigcup_{i=0}^t\{X_i^{(1)}, \dots , X_i^{(k)} \} = V  \}$ be the first time every vertex of the graph has been visited by some walk
	$X_t^{(i)}$. For $u_1, \dots , u_k\in V$ let  \[t_{\mathsf{cov}}^{(k)}((u_1,\dots, u_k),G) = \Exu{u_1, \dots, u_k}{\tau_{\mathsf{cov}}^{(k)}(G)},\qquad t_{\mathsf{cov}}^{(k)}(G)=\max_{u_1, \dots , u_k \in V}\; t_{\mathsf{cov}}^{(k)}((u_1,\dots, u_k),G)   \]denote the cover time of $k$ walks from $(u_1, \dots, u_k)$ and the cover time of $k$ walks from worst-case start positions respectively. For simplicity, we drop $G$ from the notation if the underlying graph is clear from the context. We shall use $\pi$ to denote the stationary distribution of a single random walk on a graph $G$. For $v\in V$ this is given by $\pi(v) = \frac{\deg(v)}{ 2m}$ which is the degree over twice the number of edges. Let $\pi_{\min} = \min_{v\in V} \pi(v) $ and  $\pi_{\max}=\max_{v\in V} \pi(v) $. We use $\pi^k$, which is a distribution on $V^k$ given by the product measure of $\pi$ with itself, to denote the stationary distribution of a multiple random walk. For a probability distribution $\mu$ on $V$ let $\Exu{\mu^k}{\cdot}$ denote the expectation with respect to $k$ walks where each start vertex is sampled independently from $\mu$ and \[ t_{\mathsf{cov}}^{(k)}(\mu,G) = \Exu{\mu^k}{\tau_{\mathsf{cov}}^{(k)}(G)}. \]
	In particular $t_{\mathsf{cov}}^{(k)}(\pi,G) $ denotes the expected cover time from $k$ independent stationary start vertices. For a set $S\subseteq V$ (if $S=\{v\}$ is a singleton set we use $\tau_v^{(k)}$, dropping brackets) we define 
	\[
	\tau_{S}^{(k)}=\inf\{t : \text{ there exists $1\leq i \leq k$ such that }X_t^{(i)} \in S\}
	\]
	as the first time at least one vertex in the set $S$ is visited by any of the $k$ independent random walks. 
	Let  \[  t_{\mathsf{hit}}^{(k)}(G)=\max_{u_1, \dots , u_k \in V} \max_{v \in V} \; \Exu{u_1,\dots, u_k}{\tau^{(k)}_v}  \]be the worst-case vertex to vertex hitting time. When talking about a single random walk we drop the $(1)$ index, i.e. $t_{\mathsf{cov}}^{(1)}(G) =  t_{\mathsf{cov}}(G) $; we also drop $G$ from the notation when the graph is clear. If we wish the  graph $G$ to be clear we shall also use the notation $\Pru{u,G}{\cdot}$ and  $\Exu{u,G}{\cdot} $. For $t\geq 0$ we let $N_v(t)$ denote the number of visits to $v\in V$ by a single random walk up-to time $t$. 
	
	For a single random walk $X_t$ with stationary distribution $\pi$ and $x\in V$, let $d_{\mathrm{TV}}(t)$ and $s_x(t)$ be the total variation and separation distances for $X_t$ given by
	\begin{equation*}
	d_{\mathrm{TV}}(t) = \max_{x\in V}|| P^{t}_{x,\cdot} - \pi||_{\mathrm{TV}}, \qquad \text{and}\qquad s_x(t)= \max_{y\in V}\biggl[1 - \frac{P^{t}_{x,y}}{\pi(y)} \biggr],  
	\end{equation*}where $P_{x,\cdot}^t$ is the $t$-step probability distribution of a random walk starting from $x$ and, for probability measures $\mu,\nu$, $||\mu - \nu||_{\mathrm{TV}} = \frac{1}{2}\sum_{x\in V}|\mu(x)-\nu(x)|$ is the total variation distance. 
	Let $ s(t)=\max_{x\in V}s_x(t)$, then for $0< \eps \leq 1$ the mixing and separation times \cite[(4.32)]{levin2009markov} are 
	\[\tmix(\eps) = \inf\{t : d_{\mathrm{TV}}(t) \leq \eps\} \qquad \text{and} \qquad \tsep(\eps) = \inf\{t : s(t) \leq \eps \}, \]and $\tmix:=\tmix(1/4)$ and $\tsep:= \tsep(1/e)$. A strong stationary time (SST) $\sigma$, see \cite[Ch.\ 6]{levin2009markov} or \cite{AldDia}, is a randomised stopping time for a Markov chain $Y_t$ on $V$ with stationary distribution $\pi$ if   
	\begin{equation*}\Pru{u}{Y_\sigma=v \mid \sigma = k }= \pi(v) \qquad\text{ for any $u,v\in V$ and $k\geq 0$.}  \end{equation*}

	Let $t_{\mathsf{rel}} = \frac{1}{1-\lambda_2}$ be the relaxation time of $G$, where $\lambda_2$ is the second largest eigenvalue of the transition matrix of the (lazy) random walk on $G$. A  sequence of graphs $(G_n)$ is a \emph{sequence of expanders} (or simply an \textit{expander}) if $\liminf_{n\to \infty} 1-\lambda_2(G_n) > 0$, which implies that $t_{\mathsf{rel}} = \Theta(1)$ 
	
	For random variables $Y,Z$ we say that $Y$ dominates $Z$ ($Y\succeq Z$) if $ \Pr{ Y \geq x } \geq   \Pr{ Z \geq x }$ for all real $x$. Finally, we shall use the following inequality \cite[Proposition B.3]{Motwani}: \begin{equation}\label{eq:cheatsheet} (1 +x/n)^n\geq e^x (1-x^2/n)\quad\text{ for }n\geq 1, |x|\leq n. \end{equation}

	\section{Multiple Stationary Cover Times}\label{ReturnSec}
	We shall state our general upper and lower bound results for multiple walks from stationary in Sections \ref{sec:upstate} \& \ref{sec:lowstate} before proving these results in Sections \ref{sec:upproof} \& \ref{sec:lowproof} respectively. 
	
	\subsection{Upper Bounds}\label{sec:upstate}
	
	Broder, Karlin, Raghavan, and Upfal \cite{BroderSpaceTime} showed that for any graph $G$ (with $m=|E|$) and $k\geq 1$, \[t_{\mathsf{cov}}^{(k)}(\pi) = \BO{\left(\frac{m}{k}\right)^2\log^3 n } .\]
	We first prove a general bound which improves this bound by a multiplicative factor of $\dmin ^2\log n $ which may be $\Theta(n^2\log n)$ for some graphs. \begin{theorem}\label{nonregbdd}For any graph $G$ and any $k\geq 1$, 
		\[t_{\mathsf{cov}}^{(k)}(\pi) = \mathcal{O} \biggl( \left(\frac{m}{k\dmin}\right)^2\log^2 n \biggr). \] 
	\end{theorem}
	This bound is tight for the cycle if $k=n^{\Theta(1)}$, see Theorem \ref{cycleworst}. Theorem \ref{nonregbdd} is proved by relating the probability a vertex $v$ is not hit up to a certain time $t$ to the expected number of returns to $v$ by a walk of length $t$ from $v$ and applying a bound by Oliveira and Peres \cite{oliveira2018random}.

	The next bound is analogous to Matthew's bound \cite[Theorem 2.26]{aldousfill} for the cover time of single random walks from worst-case, however it is proved by a different method. 
	\begin{theorem}\label{thm:statmatthews}For any graph $G$ and any $k\geq 1 $, we have 
		\[t_{\mathsf{cov}}^{(k)}(\pi) =\BO{\frac{\max_{v\in V}\Exu{\pi}{\tau_{v}}\log n}{k}} . \] 
	\end{theorem}
	
	This bound is tight for many graphs, see Table \ref{tbl:results}. Following this paper the stronger bound $\tcov^{(k)}(\pi)  = \BO{\tcov /k} $ was recently proved by Hermon and Sousi \cite{hermon2021covering}. A version of Theorem \ref{thm:statmatthews} for $\tcov^{(k)}$ was established in \cite{Multi} provided $k=\BO{\log n}$, the restriction on $k$ is necessary (for worst-case) as witnessed by the cycle. Theorem \ref{thm:statmatthews} also gives the following explicit bound.
	
		\begin{corollary}\label{relbdd}For any graph $G$ and any $k\geq 1 $, we have 
		\[t_{\mathsf{cov}}^{(k)}(\pi) =\BO{\frac{m}{k\dmin}\sqrt{\trel}\log n} . \] 
	\end{corollary}
	\begin{proof} Use $\max_{v\in V}\Exu{\pi}{\tau_{v}} \leq 20 m \sqrt{\trel+1}/\dmin  $ from \cite[Theorem 1]{oliveira2018random} in Theorem \ref{thm:statmatthews}.
	\end{proof}
	Notice that, for any $k\geq 1$, this bound is tight for any expander with $d_{\min}=\Omega(m/n)$, such as preferential attachment graphs with $m\geq 2$, see Theorems \ref{thm:hypercubecover} and \ref{pro:expander}. We now establish some bounds for classes of graphs determined by the return probabilities of random walks. 
	\begin{lemma}\label{thm:constreturn}
		Let $G$ be any graph satisfying $\pi_{\mathsf{min}} =\Omega(1/n)$,  $\tmix = \BO{n}$, and $\sum_{i=0}^{t} P_{vv}^i =\BO{1 + t\pi(v)}$ for any $t\leq \trel$ and $v\in V$. Then for any $1\leq k\leq n$, \[t_{\mathsf{cov}}^{(k)}(\pi) = \Theta\left(\frac{n}{k}\log n\right).\]
	\end{lemma}
	The bound above applies to a broad class of graphs including the hypercube and high dimensional grids. The following bound holds for graphs with sub-harmonic return times, this includes binary trees and $2$d-grid/torus. 
	\begin{lemma}\label{harmonicreturns} 
		Let $G$ be any graph satisfying $\pi_{\max} =\BO{\pi_{\min}}$,  $\sum_{i=0}^tP_{v,v}^i  = \BO{1+\log t}$ for any $t\leq \trel$ and $v\in V$, and $\tmix=\BO{n}$. Then for any $1\leq k\leq (n\log n)/3$,
		\[\tcov^{(k)}(\pi)  = \BO{\frac{n\log n }{k}\log\left( \frac{n\log n }{k}\right) }.\]  
	\end{lemma}
	\subsection{Lower Bounds}\label{sec:lowstate}

	Generally speaking, lower bounds for random walks seem to be somewhat more challenging to derive than upper bounds. In particular, the problem of obtaining a tight lower bound for the cover time of a simple random walk on an undirected graph was open for many years \cite{aldousfill}. This was finally resolved by Feige \cite{FeigeLow} who proved $\tcov \geq (1-\lo{1})n\log n $ (this bound was known up-to constants for stationary Markov chains by Aldous \cite{aldouslower}). We prove a generalisation of this bound, up to constants, that holds for $k$ random walks starting from stationary (thus also for worst-case).

	\begin{theorem}\label{generallower}There exists a constant $c>0$ such that for any graph $G$ and $1 \leq k \leq c \cdot n\log n$, \[\tcov^{(k)}(\pi)\geq  c\cdot \frac{n}{k}\cdot  \log n .\]
	\end{theorem} 
	We remark that in this section all results hold (and are proven) for \emph{non-lazy} random walks, which by stochastic domination implies that the same result also holds for \emph{lazy} random walks. Theorem~\ref{generallower} is tight, uniformly for all $1 \leq k \leq n$, for the hypercube, expanders and high-dimensional tori, see Theorem \ref{thm:hypercubecover}. We note that~\cite{ElsSau} proved this bound for any start vertices under the additional assumption that $k \geq n^{\epsilon}$, for some constant $\eps>0$. One can track the constants in the proof of Theorem \ref{generallower} and show that $c > 2 \cdot 10^{-11}$, we have not optimised this but note that $c\leq 1$ must hold in either condition of Theorem \ref{generallower} due to the complete graph.

	To prove this result we introduce the \textit{geometric reset graph}, which allows us to couple the multiple random walk to a single walk to which we can apply a lower bound by Aldous~\cite{aldouslower}. The geometric reset graph is a small modification to a graph $G$ which gives an edge-weighted graph $\widehat {G}(x)$ such that the simple random walk on $\widehat {G}(x)$ emulates a random walk on $G$ with $\geo{x}$ \textit{resets} to stationarity, where $\geo{x}$ is a geometric random variable with expectation $1/x$. This is achieved by adding a dominating vertex $z$ to the graph and weighting edges from $z$ so that after $z$ is visited the walk moves to a vertex $v\in V(G)$ proportional with probability $\pi(v)$.
	
	\begin{definition}[The Geometric Reset Graph $\widehat {G}(x)$]\label{georesetgraph}
		For any graph $G$ the undirected, edge-weighted graph $\widehat {G}(x)$, where $0 < x \leq 1$,  consists of all vertices $V(G)$ and one extra vertex $z$. All edges from $G$ are included with edge-weight $1-x$. Further, $z$ is connected to each vertex $u \in G$ by an edge with edge-weight $x \cdot \deg(u)$, where $\deg(u)$ is the degree of vertex $u$ in $G$. \end{definition}
	Given a graph with edge weights $\{w_e\}_{e\in E}$ the probability a \emph{non-lazy} random walk moves from $u$ to $v$ is given by $w_{uv}/ \sum_{w\in V} w_{uw}$. Thus the walk on $\widehat {G}(x)$ behaves as a random walk in $G$, apart from that in any step, it may move to the extra vertex $z$ with probability $\frac{x  \deg(u)}{x  \deg(u) + (1-x)\deg(u)} = x $. Once the walk is at $z$ it moves back to a vertex $u \in V\backslash\{z\}$ with probability\begin{equation*} P_{z,u} = \frac{x\cdot d(u) }{ \sum_{v \in V\backslash\{z\}}x\cdot d(v)} = \pi(u).  \end{equation*} Hence the stationary distribution $\widehat{\pi}$ of the random walk on $\widehat {G}(x)$ is proportional to $\pi$ on $V(G)$, and for the extra vertex $z$ we have  \[\widehat{\pi}(z) = \frac{\sum_{u\in V}x\deg(u)}{\sum_{u\in V}(1-x)\deg(u) +\sum_{u\in V}x\deg(u) } =\frac{x }{(1-x) + x} = x.\]
	
	Using the next lemma we can then obtain bounds on the multiple stationary cover time by simply bounding the cover time in the augmented graph $\widehat {G}(x)$ for some $x$. 
	\begin{lemma}\label{resetcouple}Let $G$ be any graph, $k\geq 1$ and $x = Ck/T$ where $C>30$ and $T\geq 5Ck$. Then, \[\Pru{\pi^k,G}{\tau_{\mathsf{cov}}^{(k)} >\frac{T}{10Ck}}> \Pru{\widehat {\pi},\widehat{G}(x)}{\tau_{\mathsf{cov}} >T} -\exp\left(-\frac{Ck}{50}\right).\]
	\end{lemma}

 The result in Lemma \ref{resetcouple} will also be used later in this work to prove a lower bound for the stationary cover time of the binary tree and $2$-dimensional grid when $k$ is small.   
	
	The next result we present utilises the second moment method to obtain a lower bound which works well for $k=n^{\Theta(1)}$ walks on symmetric (e.g.,~transitive) graphs. In particular, we apply this to get tight lower bounds for cycles, the $2$-dim.~torus and binary trees in Section~\ref{fundamental}.

		\begin{lemma}\label{lemma:generalLowerboundpi} For any graph $G$,  subset $S\subseteq V$ and $t\geq 1$, let $p = \max_{u\in S}\Pru{\pi}{\tau_u\leq t}$. Then, for any $k\geq 2$ satisfying $ 2p^2k<1$, we have 	
			\begin{equation*}
				\Pru{\pi^k}{ \tau_{\mathsf{cov}}^{(k)} \leq  t}\leq \frac{4kp^2e^{2kp}}{|S|\min_{v \in S} \pi(v)}.
			\end{equation*}
			
	\end{lemma}

	\subsection{Proofs of Upper Bounds}\label{sec:upproof}
	We begin by stating some basic facts.

	\begin{lemma}\label{covfromvisits}Let $N_v(t)$ be the number of visits to $v\in V$ by a $t$-step walk, $t\geq 1$, then
		\begin{enumerate}[(i)]
			\item\label{union} ${\displaystyle\Pru{\pi^k}{\tau_{\mathsf{cov}}^{(k)} \geq t  } \leq  \sum_{v\in V}\exp\left(-k\Pru{\pi}{N_v(t)\geq 1} \right), }$
			\item\label{returns} ${\displaystyle \frac{(t+1)\pi(v)}{\sum_{i=0}^tP_{v,v}^{i}}\leq \Pru{\pi}{N_v(t)\geq 1 } =\frac{\Exu{\pi}{N_v(t)}}{\Exu{\pi}{N_v(t)\mid N_v(t)\geq 1}}    \leq   \frac{2(t+1)\pi(v)}{ \sum_{i=0}^{t/2} P_{v,v}^i }}$,
			\item\label{trelenough} ${\displaystyle\sum_{i=0}^tP_{u,u}^{i} \leq \frac{e}{e-1}\left(\sum_{i=0}^{\lceil \trel \rceil-1}P_{u,u}^{i}  -\lceil \trel \rceil\pi(u)\right)+ (t+1)\pi(u) }$.
		\end{enumerate}
	\end{lemma}
	\begin{proof} For Item \eqref{union}, by independence of the walks and the union bound we have 
		\begin{align*}\Pru{\pi^k}{\tau_{\mathsf{cov}}^{(k)} \geq t  } &
		\leq \sum_{v\in V} \prod_{i=1}^k\left(1-\Pru{\pi}{N_v^{(i)}(t)\geq 1}\right)\leq \sum_{v\in V}\exp\left(-k\Pru{\pi}{N_v(t)\geq 1} \right).    \end{align*}	
		
		For Item \eqref{returns},  since $N_v(t)$ is a non-negative integer, we have 
		\[ \Exu{\pi}{N_v(t)} = \Exu{\pi}{N_v(t)\mid N_v(t)\geq 1}\cdot \Pru{\pi}{N_v(t)\geq 1 }.  \] For the two inequalities first observe that $\Exu{\pi}{N_v(t)}=(t+1)\pi(v)$. Now, conditional on a walk first hitting $v$ at time $s$ we have $N_v(t)= \sum_{i=0}^{t-s}P_{v,v}^{t}$. The first inequality in Item \eqref{returns} then follows since $\Exu{\pi}{N_v(t)\mid N_v(t)\geq 1 }\leq \sum_{i=0}^t P_{v,v}^{i}$.
		For the last inequality in Item \eqref{returns}, observe that, by reversibility, for any $t\geq 1$ \begin{align*}\Pru{\pi}{N_v\left(\left\lceil \tfrac{t-1}{2}\right\rceil\right)\geq 1,\; N_v(  t) \geq 1} &= \sum_{v_0v_1\cdots v_t \,: \,v\in \bigcup_{i=0}^{\lceil \frac{t-1}{2}\rceil}\{v_i\} } \pi(v_0)P_{v_0,v_1}\cdots P_{v_{t-1},v_t} \\ &= \sum_{v_0v_1\cdots v_t \,: \,v\in \bigcup_{i=0}^{\lceil \frac{t-1}{2}\rceil}\{v_i\} } \pi(v_t)P_{v_t,v_{t-1}}\cdots P_{v_{1},v_0}\\
			&= \Pru{\pi}{N_v(t)- N_v\left(\left\lfloor \tfrac{t+1}{2}\right\rfloor -1\right)\geq 1,\; N_v(  t) \geq 1}. 
		\end{align*}Since $\{N_v(  t) \geq 1 \}=\{N_v\left(\left\lceil \tfrac{t-1}{2}\right\rceil\right)\geq 1 \}\cup \{N_v(t)- N_v\left(\left\lfloor \tfrac{t+1}{2}\right\rfloor -1\right)\geq 1\} $ we have \[\Pru{\pi}{N_v\left(\left\lceil \tfrac{t-1}{2}\right\rceil\right) \geq 1\mid N_v( t) \geq 1} \geq 1/2.\] Thus $\Exu{\pi}{N_v(t) \mid N_v(t)\geq 1} \geq (1/2)\cdot \sum_{i=0}^{\lceil \frac{t-1}{2}\rceil}P_{v,v}^{i}$ as claimed.

		Finally, for Item \eqref{trelenough}, the proof of \cite[Lemma 1]{oliveira2018random} shows that \[\sum_{i=0}^t P_{u,u}^{i} - (t+1)\pi(u)\leq \frac{e}{e-1}\left(\sum_{i=0}^{\lceil \trel \rceil-1}P_{u,u}^{i}  -\lceil \trel \rceil\pi(u)\right),\] rearranging gives the result.
	\end{proof}

	Recall	$a\wedge b$ denotes $\min( a, b )$. We shall use the following result of Oliveira and Peres \cite[Theorem 2]{oliveira2018random} for lazy walks: for any $v\in V $ and $t\geq 0$ we have 
	\begin{equation}\label{OPBound}
	P_{v,v}^t - \pi(v) \leq \frac{10\deg(v)}{\dmin}\left(\frac{1}{\sqrt{t+1}} \wedge \frac{\sqrt{\trel+1}}{t+1} \right) .
	\end{equation}
	
	Note that we prove Theorem \ref{nonregbdd} for lazy walks however this also applies to non-lazy walks as the cover time of a lazy walk stochastically dominates that of a non-lazy walk.
	\begin{proof}[Proof of Theorem \ref{nonregbdd}] Recall that we aim to prove $t_{\mathsf{cov}}^{(k)}(\pi) = \mathcal{O} \big( \big(\frac{m}{k\dmin}\big)^2\log^2 n \big)$. To begin observe that if $k\geq 10 (m/\dmin )\log n$ then the probability any vertex $u$ is unoccupied at time $0$ is $(1-\pi(u))^k\leq e^{-\pi(u)k}\leq e^{-\dmin k/(2m)}$. For any graph $\tcov \leq 16mn/\dmin \leq 16n^3 $ by \cite[Theorem 2]{KLScov}, thus we have \[\Exu{\pi^k}{\tau_{\mathsf{cov}}^{(k)}}\leq 16n^3 \cdot n e^{-\dmin k/(2m)} = o\Big(\left( m/(k\dmin)\right)^2\log^2 n\Big).\] It follows that we can assume $k\leq 10 (m/\dmin )\log n$ for the remainder of the proof.

		   We will apply Lemma \ref{covfromvisits} to bound $\Pru{\pi^k}{\tau_{\mathsf{cov}}^{(k)} \geq t}$ for $t\geq 1$. To begin, by \eqref{OPBound} we have 
		\begin{align*}
		\sum_{i=0}^{t} P_{u,u}^{i} &\leq \frac{10\deg(u)}{\dmin}\sum_{i=0}^{t} \frac{1}{\sqrt{i+1}} + (t+1) \pi(u) \leq \frac{10 \deg(u)}{\dmin}\left(\int_{1}^{t} \frac{1}{\sqrt{x}}\, \mathrm{d}x +1\right) + (t+1)\pi(u).\end{align*}Now, since $\dmin\leq \deg(u)$, $\pi(u)\leq 1$ and $t\geq 1$, we have  \begin{align*} \sum_{i=0}^{t} P_{u,u}^{i} &\leq \frac{10\deg(u)}{\dmin}\left(\left[2\sqrt{t}  - 2\sqrt{1}  \right] + 1 \right)+ t\pi(u) +1  \leq \frac{20\deg(u)}{\dmin}\sqrt{ t} +t\pi(u).\end{align*}
		Thus, by Lemma \ref{covfromvisits} \eqref{returns} and dividing each term by a factor of $\pi(u)=\deg(u)/2m$ we have
		\[\Pru{\pi}{N_v(t)\geq 1 } \geq \frac{t}{(40m/\dmin)\sqrt{ t}  +t}.\] We now define $t^* =\left(\frac{300m\log n}{k\dmin}\right)^2$. Firstly if  $k\geq 10\log n $, then for any $t\geq t^*$ we have
		\[k\Pru{\pi}{N_v(t)\geq 1 }\geq \frac{k\left(\frac{300m\log n}{k\dmin}\right)^2}{(40m/\dmin)\cdot \frac{300m\log n}{k\dmin} + \left(\frac{300m\log n}{k\dmin}\right)^2} \geq \frac{300\log n }{40+300/10}>4\log n.  \] Now, since $\tcov\leq 16n^3$, Lemma \ref{covfromvisits}\eqref{union} gives $\Exu{\pi^k}{\tau_{\mathsf{cov}}^{(k)}}\leq t^* +\lo{n^{-3}}\cdot16n^3 = \mathcal{O}\big(\big(\frac{m}{k\dmin}\big)^2\log^2 n\big)   $. 
		
		Finally, if $k\leq 10\log n $ then $ t^*\geq (300m/\dmin)^2/100 =900(m/\dmin)^2$. However, $\tcov \leq 16mn/\dmin \leq 32m^2/\dmin^2$ by \cite[Theorem 2]{KLScov}. Thus, $\Exu{\pi^k}{\tcov^{(k)}}\leq t^* $ holds as claimed.     
	\end{proof}
	\begin{proof}[Proof of Theorem \ref{thm:statmatthews}]
	
	We consider first the case $k < 8\log_2 n$. To begin, for any pair $u,v$ it holds that $\Exu{u}{\tau_{v}} \leq 2\max_{w\in V}\Exu{\pi}{\tau_{w}}$ by \cite[Lemma 10.2]{levin2009markov}.	Thus by Markov's inequality
	\begin{align*}
	    \Pru{u}{\tau_{v}\geq 4\left\lceil\max_{w\in V}\Exu{\pi}{\tau_{w}}\right\rceil} \leq \frac{\Exu{u}{\tau_v}}{4\left\lceil\max_{w\in V}\Exu{\pi}{\tau_{w}}\right\rceil}\leq \frac{1}{2},
	\end{align*}
for any $v\in V$. Then, the Markov property yields 
	\begin{align*}
	    \Pru{u}{\tau_{v}\geq 20\left\lceil \frac{\log_2 n}{k}\right\rceil\cdot   \left\lceil\max_{w\in V}\Exu{\pi}{\tau_{w}}\right\rceil} \leq \left(\frac{1}{2}\right)^{5(\log_2 n)/k} = \frac{1}{n^{5/k}}.
	\end{align*}
	Thus, by independence of the $k$ walks
	\begin{align*}
	    \Pru{\pi^k}{\tau_{v}^{(k)}\geq 20\left\lceil \frac{\log_2 n}{k}\right\rceil\cdot  \left\lceil\max_{w\in V}\Exu{\pi}{\tau_{w}}\right\rceil} \leq \frac{1}{n^5},
	\end{align*}
	and finally by the union bound,
	\begin{align*}
	    \Pru{\pi^k}{\tcov^{(k)}\geq  20\left\lceil \frac{\log_2 n}{k}\right\rceil\cdot \left\lceil\max_{w\in V}\Exu{\pi}{\tau_{w}}\right\rceil} \leq \frac{1}{n^4}.
	\end{align*}
	Therefore, since the cover time of a single walk satisfies $\tcov =\BO{n^3}$, we have
	\begin{align*}
	\Exu{\pi^k}{\tcov^{(k)}}  \leq  20\left\lceil \frac{\log_2 n}{k}\right\rceil\cdot \left\lceil\max_{w\in V}\Exu{\pi}{\tau_{w}}\right\rceil+ \lo{1}.
	\end{align*}
	
We now cover the case  $ 8\log_2 n\leq k \leq 100\max_{v\in V}\Exu{\pi}{\tau_v} \cdot \log n$, where we apply Items \eqref{union} \& \eqref{returns} of Lemma \ref{covfromvisits} to bound $\Pru{\pi^k}{\tau_{\mathsf{cov}}^{(k)} \geq t}$. Observe that
	\begin{align*}
		    \sum_{i=0}^{t} P_{v,v}^i = \sum_{i=0}^{t} (P_{v,v}^i-\pi(v)) + (t+1)\pi(v). 
		\end{align*}
		Since the walk is lazy  $(P_{v,v}^i-\pi(v))$ is non-negative and non-increasing in $i$ \cite[Exercise 12.5]{levin2009markov}. Thus
		\begin{align*}
		    \sum_{i=0}^{t} P_{v,v}^i \leq  \sum_{i=0}^{\infty} (P_{v,v}^i-\pi(v)) + (t+1)\pi(v) = \pi(v)\Exu{\pi}{\tau_v}+ (t+1)\pi(v), 
		\end{align*}
		as $\pi(v)\Exu{\pi}{\tau_v}=\sum_{i=0}^{\infty} (P_{v,v}^i-\pi(v))$ by \cite[Lemma 2.11]{aldousfill}. Now, by Lemma \ref{covfromvisits} \eqref{returns}  we have \begin{align*}
		     \Pru{\pi}{N_v(t)\geq 1 } \geq \frac{(t+1)\pi(v)}{\pi(v)\Exu{\pi}{\tau_v}+ (t+1)\pi(v)} =\frac{(t+1)}{\Exu{\pi}{\tau_v}+ (t+1)}.
		 \end{align*}
	     Choosing $t = 8\left\lceil \frac{\log n}{k}\max_{v\in V}\Exu{\pi}{\tau_v} \right\rceil$ yields
		\begin{align*}
		     \Pru{\pi}{N_v(t)\geq 1 } \geq \frac{8\log n}{k+8\log n} \geq \frac{4\log n}{k},
		 \end{align*}
		 since $k\geq 8\log n$. Then an application of Lemma \ref{covfromvisits} \eqref{union} gives 
		 \begin{align*}
		     \Pru{\pi^k}{\tau_{\mathsf{cov}}^{(k)} \geq t  } \leq n\exp\left(-4\log n\right) = \frac{1}{n^3}.
		 \end{align*}
		 Again, as $\tcov = \bigO(n^3)$, we conclude $\Exu{\pi^k}{\tcov^{(k)}}  \leq 8\left\lceil \frac{\log n}{k}\max_{v\in V}\Exu{\pi}{\tau_v} \right\rceil +\bigO(1)$.
	 
 Finally assume that $k\geq 100\max_{v\in V}\Exu{\pi}{\tau_v} \cdot \log n$. Recall that $ \max_v\Exu{\pi}{\tau_v} \geq (1/2) \cdot  \max_{x,y}\Exu{x}{\tau_y}$ by \cite[Lemma 3.15]{aldousfill} and $\max_{x,y}\Exu{x}{\tau_y} \geq \max_{x,y} m R(x,y)$ by the commute time identity \cite[Proposition 10.6]{levin2009markov}, where $R(x,y)$ is the effective resistance between $x$ and $y$ (see \cite[Section 9.4]{levin2009markov}). Thus we have \begin{equation} \label{eq:pilower} \max_v\Exu{\pi}{\tau_v} \geq \frac{1}{2} \cdot m\cdot \max_v \frac{1}{\deg(v)} \geq\frac{1}{2}\cdot  \max_v \frac{1}{\pi(v)},\end{equation}  as $R(x,y)\geq \max\{1/\deg(x),1/\deg(y)\}$ by the definition of the effective resistance.  
   The probability that any vertex $u$ is unoccupied at time $0$ is $(1-\pi(u))^k\leq e^{-\pi(u)k}$. For any graph $\tcov  \leq 16n^3 $ by \cite[Theorem 2]{KLScov} and thus, one can check that by \eqref{eq:pilower}, we have \[\Exu{\pi^k}{\tau_{\mathsf{cov}}^{(k)}}\leq \max_{v\in V} 16n^3 \cdot n e^{-\pi(v)k}= o\Big(\frac{\max_{v\in V}\Exu{\pi}{\tau_v} \cdot \log n}{k}\Big),\]as claimed. 
\end{proof}

	\begin{proof}[Proof of Lemma \ref{thm:constreturn}]
		The lower bound will follow from the general lower bound we prove later in Theorem \ref{generallower}. 
		For the upper bound notice that, by Lemma \ref{covfromvisits} \eqref{trelenough} and our hypothesis that $\sum_{i=0}^{t} P_{vv}^i =\BO{1 + t\pi(v)}$ for any $t\leq \trel$, we have that for any  $T\geq 0$,  
		\begin{equation*}
		\sum_{t=0}^T P_{v,v}^t \leq  \frac{e}{e-1}\left(\sum_{i=0}^{\lceil \trel \rceil-1}P_{v,v}^{i}  -\lceil \trel \rceil\pi(v)\right)+ (T+1)\pi(v)\leq C(1 + T\pi(v)), 
		\end{equation*} for some constant $C<\infty$. Now, by Lemma~\ref{covfromvisits} \eqref{returns} and as $\pi_{\mathsf{min}} =\Omega(1/n)$ holds by hypothesis, for any $v$ and $T= \BO{n}$ there exists a constant $C'<\infty$ such that
		
		\begin{equation}\label{eqn:returnHyper}\Pru{\pi}{N_v(T)\geq 1} \geq \frac{T\pi(v)}{C(1 + T\pi(v))}\geq \frac{T}{C'n}.\end{equation} First consider $k=\omega(\log n)$, and let $T= \left\lceil 4C'(n/k)\log n\right\rceil =\BO{n}$. Lemma~\ref{covfromvisits} \eqref{union} and \eqref{eqn:returnHyper}  give
		\begin{align*}
		\Pru{\pi^k}{\tau_{\mathsf{cov}}^{(k)} \geq T} \leq n \exp\left(-k\Pru{\pi}{N_v(T)\geq 1} \right)=n\exp\left(-4\log n\right)=\frac{1}{n^3}  .
		\end{align*}
		Otherwise, if $k = \BO{\log n}$, then we consider consecutive periods of length $5t_{\mix}= \BO{n}$. For any $u,v\in V$, \eqref{eqn:returnHyper} gives	\begin{equation*}
		\Pru{u}{N_v(5t_{\mix})\geq 1}\geq \frac{1}{4}\Pru{\pi}{N_v(t_{\mix})\geq 1} \geq \frac{t_{\mathsf{mix}}}{4C'n}.\end{equation*}
		The probability that a vertex $v$ is not hit in a period, starting from any vector in $V^k$ is\begin{align*}
		(1-\min_{u\in V}\Pru{u}{N_v(5t_{\mix})\geq 1})^k \leq e^{-kt_{\mathsf{mix}}/(4C'n)}
		\end{align*}
		and thus after $20C'\lceil \frac{n}{k t_{\mathsf{mix}}}\cdot \log n\rceil $ periods the probability $v$ has not been hit is at most $e^{-5 \log n}= n^{-5}$. By the union bound, and since worst-case cover time of a single walk on any graph is $\BO{n^3}$, the cover time is $\BO{t_{\mix} \cdot   \frac{n}{k t_{\mathsf{mix}}}\cdot \log n }= \BO{\frac{n}{k}\cdot \log n}$.
	\end{proof}

	\begin{proof}[Proof of Lemma \ref{harmonicreturns}] From Lemma \ref{covfromvisits} \eqref{trelenough} we deduce that $\sum_{i=0}^{t} P_{v,v}^i  = \BO{1+ t/n+\log t}$ for any $t\leq n(\log n)^2$. Thus we can apply Lemma \ref{covfromvisits} \eqref{returns} which, for any $2\leq t \leq n(\log n)^2$, gives   \begin{equation}\label{eq:kprob}k\Pru{\pi}{\tau_v\leq t} \geq \frac{ckt}{n\log t + t }, \end{equation} for some fixed constant $c>0$. Let $t^*=C \frac{n\log n}{k}\log \left(\frac{n\log n}{k}\right)$ and $C\log n \leq k \leq (n\log n )/3 $ for some constant $1<C:=C(c)<\infty$ to be determined later then,
		\begin{equation}\label{eq:t*harm}n\log t^* = n \log \left(C \frac{n\log n}{k}\log \left(\frac{n\log n}{k}\right)\right)\geq n\log\left( \frac{n\log n}{k}\right)\geq   t^* .    \end{equation}In addition we have 
		\begin{equation}\label{eq:t*harmotherway}\log t^* =\log \left(C \frac{n\log n}{k}\log \left(\frac{n\log n}{k}\right)\right)\leq (2 + \log C)\log\left( \frac{n\log n}{k}\right) .    \end{equation}
		Thus inserting \eqref{eq:t*harm} into \eqref{eq:kprob} then applying \eqref{eq:t*harmotherway} yields the following for any $v\in V$,
		\[k\Pru{\pi}{\tau_v\leq t} \geq \frac{c}{2} \cdot \frac{kt^*}{n\log t^*} \geq \frac{c  }{2  }\cdot \frac{(Cn\log n)\log\left(\frac{n\log n}{k}\right) }{(2 + \log C)n\log\left(\frac{n\log n}{k}\right)}  = \frac{Cc \log n }{4 + 2\log C  }  .\] We can assume w.l.o.g. that $c<1$ and thus taking $C= 100/c^2$ yields $\frac{Cc }{4 + 2\log C  } > \frac{100/c}{14 + 4\log(1/c)} >5 $. So  by independence of the walks we have  \[\Pru{\pi^k}{\tau_v^{(k)} >t^*} \leq \left( 1-\Pru{\pi}{\tau_v\leq t^* }\right)^k\leq \exp\left(-k\Pru{\pi}{\tau_v\leq t^* }\right) \leq n^{-5}  .\] Thus, by the union bound $\Pru{\pi^k}{\tau_\mathsf{cov}^{(k)} >t^*}\leq n^{-4} $. So, since the worst-case expected cover time by a single walk on any graph is $\BO{n^3}$, we have  $\tcov^{(k)}(\mathcal{T}_n,\pi)= \BO{t^*} $, as claimed. 
		
		The case $k\leq (100\log n )/c^2$ remains.	For any $k = \BO{\log n}$, consider periods of length $5t_{\mix}=\BO{n}$. Then by Lemma \ref{covfromvisits} \eqref{returns}, for any pair of vertices $v,w$ and some $C'<\infty$, we have \[
		\Pru{w}{N_v(5t_{\mix})\geq 1}\geq \frac{1}{4}\Pru{\pi}{N_v(t_{\mix})\geq 1} \geq \frac{t_{\mathsf{mix}}}{C'n\log n}.\]  Thus the probability $v\in V$ is not hit in a period, starting from any configuration in $ V^k$ is \begin{align*}
		(1-\Pru{w}{N_v(5t_{\mix})\geq 1})^k \leq e^{-kt_{\mathsf{mix}}/(C'n\log n)}.
		\end{align*}
		Thus after $5C'\lceil (n\log^2 n)/(k t_{\mathsf{mix}})\rceil $ periods the probability $v$ has not been hit is at most $e^{-5 \log n}= n^{-5}$, so (similarly) the cover time is $\BO{t_{\mix} \cdot \frac{n\log^2 n}{k t_{\mathsf{mix}}}}= \BO{\frac{n}{k}\cdot \log^2 n}$ as claimed.
	\end{proof}

	\subsection{Proofs of Lower Bounds}\label{sec:lowproof}

 	We begin by proving the lower bound obtained by the reset coupling.

	\begin{proof}[Proof of Lemma \ref{resetcouple}]Let $x = Ck/T$, and recall that in each step that the random walk on $\hat{G}(x)$ is not at the vertex $z$, it moves to $z$ with probability $x$. For $i\geq 1$, we refer to the portion of the walk  between the $i$\textsuperscript{th} and $(i+1)$\textsuperscript{th} visits to $z$ as the $i$\textsuperscript{th} sub-walk. For $i\geq 0$ let $X_i\sim \geo{x}$ be i.i.d.\ geometric random variables with mean $1/x$, then it follows that the $i$\textsuperscript{th} sub-walk has length $X_i+1$ and takes $X_i-1$ steps inside $G$ (the steps leaving and entering $z$ are not in $G$) and it takes at most $X_0$ steps to first visit $z$.

If we let $X=\sum_{i=0}^{Ck/3}X_i$ then we see that $X + Ck/3$ stochastically dominates the time taken for the first $Ck/3$ sub-walks to occur within a walk on $\hat{G}(x)$. Observe that $\Ex{X} = T/3$ and so a Chernoff bound for sums of geometric r.v.'s \cite[Theorem 2.1]{Geotail} gives \[\Pr{X >2\Ex{X}  } < e^{-x\Ex{X}(2- 1 -\ln(2))}\leq e^{-\frac{Ck}{10}}.\] Note that $2\Ex{X}+ Ck/3 \leq T $. Thus, if $\mathcal{E}_1$ is the event that there are at least $Ck/3$ sub-walks in a walk of length $T$, then  $\Pr{\mathcal{E}_1^c}\leq e^{-\frac{Ck}{10}}$. 
	
	For each $i\geq 0$, we have $\Pr{X_i\geq 1/(2x)}= (1-x)^{1/(2x)}\geq  1/2 $ by Bernoulli's inequality. Let $\mathcal{E}_2$ be the event that $\{X_i\geq 1/(2x)\}$ holds for more than $k$ values $1\leq i\leq Ck/3$. Recall that $C>30$ and, as the $X_i$'s are independent, a Chernoff bound \cite[Theorem 4.5]{PandC} gives 
		\[\Pr{\mathcal{E}_2^c}\leq  \Pr{\bin{\frac{Ck}{3}}{\frac{1}{2}}\leq  k}\leq  \exp\left(- \frac{1}{2}\cdot \left(1-\frac{6}{C}\right)^2\cdot \frac{Ck}{6}   \right) \leq  e^{-\frac{Ck}{20}} .\]

		Observe that conditional on $\mathcal{E}_1\cap \mathcal{E}_2$ there are at least $k$ sub-walks which take at least $s$ many steps within $G$ where $s= \frac{1}{2x} -1  = \frac{T}{2Ck}-1\geq \frac{T}{10Ck}$, since $T\geq 5Ck$. Recall that at each visit to $z$ the walk on $\hat{G}(x)$ moves to a vertex $v\in V(G)$ proportional to $\pi(v)$. It follows that, conditional on $\mathcal{E}_1\cap \mathcal{E}_2$, we can couple $k$  stationary random walks of length $s$ on $G$ to a single random walk of length $T$ on $\hat{G}$ such that if the latter walk has not covered $\hat{G}$ then the former walks have not covered $G$. Thus 
		\[\Pru{\pi^k,G}{\tau_{\mathsf{cov}}^{(k)} >\frac{T}{10Ck}}\geq  \Pru{\widehat {\pi},\widehat{G}(x)}{\tau_{\mathsf{cov}} >T, \mathcal{E}_1\cap \mathcal{E}_2}\geq  \Pru{\widehat {\pi},\widehat{G}(x)}{\tau_{\mathsf{cov}} >T} - \Pr{ \mathcal{E}_1^c\cup \mathcal{E}_2^c} , \] and the result follows since $ \Pr{ \mathcal{E}_1^c\cup \mathcal{E}_2^c} \leq e^{-\frac{Ck}{10}} + e^{-\frac{Ck}{20}}<e^{-\frac{Ck}{50}}  $. 	\end{proof}

	We shall prove the general lower bound with this coupling, however,  we must first state a technical lemma. 
	\begin{lemma}\label{lowerbddreversible} For any reversible Markov chain and $0<b<1$ we have $ \tcov\leq 2\tcov(\pi) $ and $\Pru{\pi}{\tau_{\cov}>b\cdot \tcov(\pi)  } > \frac{1-b}{2}$. 
	\end{lemma}
	
	\begin{proof}[Proof of Lemma \ref{lowerbddreversible}]
		We first show $\tcov\leq 2\tcov(\pi)$. The Random Target Lemma \cite[Lemma 10.1]{levin2009markov} states that $\Exu{x}{\tau_\pi}=\sum_{v\in V}\Exu{x}{\tau_v}\pi(v)$ does not depend on $x\in V$. Thus, for any $u\in V$,  \[ \Exu{u}{\tau_\pi} = \Exu{\pi}{\tau_\pi}\leq \max\limits_{v\in V} \Exu{\pi}{\tau_v}. \] Then, since choosing a vertex $v$ independently from $\pi$ and waiting for the random walk to hit $v$ is a strong stationary time for the random walk, we have 
		\begin{align*} \Exu{u}{\tau_{\mathsf{cov}}} &\leq \Exu{u}{\tau_\pi} + \Exu{\pi}{\tau_{\mathsf{cov}}} \leq \max\limits_{v\in V}\Exu{\pi }{\tau_v} + \Exu{\pi}{\tau_{\mathsf{cov}}} \leq 2 \Exu{\pi}{\tau_{\mathsf{cov}}}. 
		\end{align*}
	Observe that, by the Markov Property, for any $0<b<1$ we have 
		\begin{align*}\tcov(\pi) &\leq b \tcov(\pi)\cdot \Pru{\pi}{\tau_{\cov}\leq b\tcov(\pi) }+(b \tcov(\pi) + \tcov)\cdot \Pru{\pi}{\tau_{\cov}> b\tcov(\pi) }\\ 
		&=  b \tcov(\pi) + \tcov\cdot \Pru{\pi}{\tau_{\cov}> b \tcov(\pi) }\\
		&\leq   b \tcov(\pi) + 2\tcov(\pi)\cdot \Pru{\pi}{\tau_{\cov}> b \tcov(\pi) },\end{align*} as $\tcov\leq 2\tcov(\pi)$. Rearranging gives $\Pru{\pi}{\tau_{\cov}> b\cdot \tcov(\pi) } \geq \frac{1-b}{2}$.	\end{proof}
	
	Before proving Theorem \ref{generallower} we must recall a result by Aldous \cite{aldouslower}.
	
	\begin{theorem}[{\cite[Theorem 1]{aldouslower}}] \label{aldouslower} Let $(X_t)_{t\geq 0}$ be a stationary Markov chain on a state space $I$ with
		irreducible transition matrix $P_{i,j}$ and stationary distribution $\pi$. Let $\tau_{\mathsf{cov}}$ be the cover time and define $t^{*}$ to be the solution of
		\begin{align*} 
		\sum_{i\in I} \exp\left( -t^{*} \cdot \pi(i) \right) = 1.  
		\end{align*} Let $0 < \theta < 1$ and suppose the following hypotheses are satisfied:
		\begin{enumerate}[(a)]
			\item $P_{i,i} = 0$ for all $i\in I$.
			\item $\sum_{j\in J}\exp\left( -t^{*} \cdot \pi(j) \right)\geq \theta $ where $J\subseteq  I$ is the set of states $j$ such that $\max_{i\in I }P_{i,j}\leq 1 - \theta $.  
			\item  The chain is reversible; that is
			$\pi(i)P_{i,j} = \pi(j) P_{j,i}$
			for all $i, j\in I$.
		\end{enumerate}
		Then $\Ex{\tau_{\mathsf{cov}}}\geq  c_0 t^*$, where $c_0 > \theta$ depends only on $\theta$.
	\end{theorem}
We now have all the pieces we need to prove Theorem \ref{generallower}.
	
	\begin{proof}[Proof of Theorem \ref{generallower}] We wish to apply Theorem \ref{aldouslower} to the geometric reset graph $\hat{G}(x)$ for some suitable $x:=x(n,k)$. The following claim is based on \cite[Proposition 2]{aldouslower}. 
		
		\begin{claim}\label{clm:new} For $n\geq 3$ and $0< x \leq 1/2$ the walk on $\hat{G}(x)$ satisfies Theorem \ref{aldouslower} with $\theta=1/2$. 			
		\end{claim}
		\begin{poc}Let $\pi$ and $\hat{\pi}$ denote the stationary distributions of $G$ and $\hat{G}$ respectively. Note that for any $x<1$ the walk on the weighted graph $\hat{G}(x)$ satisfies $P_{i,i}=0$ and is reversible. These are items (a) and (c) from Theorem \ref{aldouslower} and we now prove item (b) also holds for $\theta=1/2$. 
			
			Partition $V(\hat{G})= V(G)\cup\{z\}$ into sets $I_1$, $I_2$, $I_3$ as follows. Let $I_1$ be the set of
			leaves of $G$, that is vertices of degree $1$ in $G$ (and degree $2$ in $\hat{G}$ as all vertices are connected to $z$). Let $I_2$ be the set of vertices of $V(G)$ which are
			adjacent to at least one leaf of $G$. Let $I_3$ be the remaining vertices (in particular $z \in I_3$). Let $J = I_1 \cup I_3$ and fix $j \in J$. For the case $j=z$, by the definition of $\hat{G}$ and by hypothesis, we have $P_{v,z} =x\leq 1/2 $ for any $v\in V$. Otherwise, if $i$ is a neighbour of $j\neq z$ then $i$ cannot be a leaf, so
			$P_{i,j} = \frac{1-x}{\deg(i)(1-x)+x\deg(i)} = \frac{1-x}{\deg(i)} \leq 1/2$. Thus it suffices to verify \begin{equation}\label{eq:claimsum}\sum_{j\in J}\exp\left( -t^{*} \cdot \hat{\pi}(j) \right)\geq 1/2 .\end{equation} Consider an edge $ik\in E(\hat{G})$  where $i\in I_1$, then we must have $k\in  I_2$. Since $\hat{\pi}(i)<\hat{\pi}(k)$ we have $\exp\left( -t^{*} \cdot \pi(i) \right)\geq \exp\left( -t^{*} \cdot \pi(k) \right)$ and also $|I_1|\geq |I_2|$. Thus, summing over all leaves $i$ gives \begin{equation}\label{eq:domsum}\sum_{i\in I_1}\exp\left( -t^{*} \cdot \hat{\pi}(i) \right)\geq \sum_{k\in I_2}\exp\left( -t^{*} \cdot \hat{\pi}(k) \right).\end{equation} However by the definition of $t^*$, the sum over all $i\in I$ is equal to $1$, thus so \eqref{eq:domsum} implies that the
			sum over $I_2$ is at most $1/2$, and so \eqref{eq:claimsum} follows.\end{poc}
		
		Now since $t^*\geq n\log n $ by \cite[(1.2)]{aldouslower}, it follows from Claim \ref{clm:new} and Theorem \ref{aldouslower} that there exists some universal constant $0<c<\infty$ such that for any $0<x\leq 1/2$ we have  \[\Exu{\widehat {\pi},\widehat{G}(x)}{\tau_{\mathsf{cov}}} \geq cn\log n. \] Thus by Lemma \ref{lowerbddreversible} we have $\Pru{\widehat {\pi},\widehat{G}(x)}{\tau_{\mathsf{cov}} >\frac{c}{2}\cdot n\log n} \geq \frac{1}{4}$. We seek to apply Lemma \ref{resetcouple} with $C=100$ and $T= \frac{c}{2}\cdot n\log n$. Firstly, we see that the condition $T\geq 5kC$ forces the restriction $k\leq c'\cdot  n \log n$ where $c' = \frac{c}{2\cdot 5\cdot C}= \frac{c}{1000}$ for  $0<c<\infty$ universal. Secondly the condition $1/2\geq x = Ck/T$ forces the condition $k\leq c''n\log n $ where $c'' = \frac{c}{2\cdot C}\geq c' $.  
		
		Thus, provided $1\leq k\leq c'\cdot  n \log n$ holds, all other assumptions are satisfied and we have \[\Pru{\pi^k,G}{\tau_{\mathsf{cov}}^{(k)} >\frac{(c/2)\cdot n\log n}{ 10\cdot 100\cdot k}}>\frac{1}{4}  -\exp\left(-\frac{100k}{50}\right) \geq \frac{1}{4} - e^{-2}> \frac{1}{10} \]by Lemma \ref{resetcouple}.  The result follows by taking the constant in the statement to be $c/2000\leq c'$,  since $c>0$ is universal.   
	\end{proof}
 
\begin{proof}[Proof of Lemma \ref{lemma:generalLowerboundpi}]
	Let $X_v= \mathbf{1}(\tau_v >t)$ be the indicator that $v\in V$ has \textit{not} been visited up-to time $t$ and $X = \sum_{v\in S} X_v$. Let  $p_v = \Pru{\pi}{\tau_v\leq t}$, $p=\max_{v\in S}p_v$, and observe that $\Exu{\pi^k}{X_v} = (1-p_v)^k$. Thus, we have   \begin{equation}\label{ekp}\Exu{\pi^k}{X}=\sum_{v\in S}(1-p_v)^k\geq \sum_{v\in S} e^{-kp_v}(1-p_v^2 k) \geq  \sum_{v\in S} e^{-kp_v}(1-p^2 k)>0,\end{equation}
	where the first inequality is by \eqref{eq:cheatsheet} since  $0\leq p\leq 1  $ and the last is by the assumption $p^2k \leq 2p^2k  <1$. Let  $R(t)$ be the number of vertices in $S$ that are visited by a single $t$-step random walk. Observe that $\Exu{\pi}{R(t)} = \sum_{v \in S} p_v$. Let $r(v,w) = \Pru{\pi}{\tau_v\leq t,\tau_w\leq t}$, then for any $v,w\in V$,
	\begin{align*}
		\Exu{\pi^k}{X_v X_w} 
		&= \left(1-\Pru{\pi}{\{\tau_v\leq t\} \cup \{\tau_w\leq t\}}\right)^k\nonumber= (1-p_v-p_w +r(v,w))^k.
	\end{align*}Recall the identity $a^k-b^k= (a-b)\sum_{i=0}^{k-1}a^{k-1-i}b^i$ for $k\geq 2$. Thus, 
	\begin{align*}
		\E_{\pi^k}[X^2] &= \sum_{v\in S} \sum_{w\in S} (1-p_v-p_w+r(v,w))^k\notag 
		\\
		&=\sum_{v\in S} \sum_{w\in S}  \left( (1-p_v-p_w)^k + ((1-p_v-p_w+r(v,w))^k-(1-p_v-p_w)^k)\right)\notag  \\
		&= \sum_{v\in S} \sum_{w\in S}  \left( (1-p_v-p_w)^k  +r(v,w)\sum_{i=0}^{k-1}(1-p_v-p_w)^i(1-p_v-p_w+r(v,w))^{k-1-i}\right).\end{align*}Since $2p^2k <1$ and $k\geq 2$, we have $1- p_v -p_w \geq 1-2p >0$. Additionally we have $ r(v,w)\leq \min(p_v,p_w)$ and so $1- p_v -p_w+r(v,w)\leq 1- \max\{p_v,p_w\} \leq 1$. Thus \begin{align} 
		\E_{\pi^k}[X^2]	&\leq \sum_{v\in S} \sum_{w\in S}  \left(  (1-p_v-p_w)^k +r(v,w) \cdot k\right)\notag  \\
		&\leq \sum_{v\in S} \sum_{w\in S}  e^{-kp_v-kp_w} +\left(\sum_{v\in S} \sum_{w\in S}  r(v,w)\right) \cdot k\notag  \\
		&= \left(\sum_{v\in S}e^{-kp_v} \right)^2+ \E_{\pi}[R(t)^2]\cdot k.\label{eq:ekp2} 
	\end{align}
	Recall the Paley-Zygmund inequality: for any non-negative random variable $X$,  $\Pru{\pi^k}{X> 0} \geq  \Exu{\pi^k}{X}^2/\Exu{\pi^k}{X^2} $. Inserting \eqref{ekp} and \eqref{eq:ekp2} into the (inverted) fraction $\frac{\Exu{\pi^k}{X^2}}{\Exu{\pi^k}{X}^2}$ gives
	\begin{equation*}
		\frac{\Exu{\pi^k}{X^2}}{\Exu{\pi^k}{X}^2}  \leq \frac{1}{(1-p^2k)^2}+ \frac{\E_{\pi}{[R(t)^2]}k  }{|S|^2e^{-2pk}(1-p^2k)^2}  
	 =\frac{1}{(1-p^2k)^2}+ \frac{ke^{2pk}\E_{\pi}[R(t)^2] } {|S|^2(1-p^2k)^2}.
	\end{equation*}
	Finally, we claim that $\Exu{\pi}{R(t)^2} \leq2p^2 |S|(\min_{v \in S} \pi(v))^{-1}$, therefore
	\begin{equation*}
		\frac{\Exu{\pi^k}{X^2}}{\Exu{\pi^k}{X}^2} \leq \frac{1}{(1-p^2k)^2}+  \frac{2kp^2e^{2kp}|S|/(\min_{v \in S} \pi(v))}{|S|^2(1-p^2k)^2} \leq  \frac{1+  \frac{2kp^2e^{2kp}}{|S|\min_{v \in S} \pi(v)}}{(1-p^2k)^2}  
	\end{equation*}Now since $\Pru{\pi^k}{ \tau_{\mathsf{cov}}^{(k)} \leq  t} = 1 - \Pru{\pi^k}{ \tau_{\mathsf{cov}}^{(k)} > t} \leq 1 - 	 \Exu{\pi^k}{X}^2/\Exu{\pi^k}{X^2} $  we have 
\[\Pru{\pi^k}{ \tau_{\mathsf{cov}}^{(k)} \leq  t} \leq 1 - \frac{(1-p^2k)^2}{1+  \frac{2kp^2e^{2kp}}{|S|\min_{v \in S} \pi(v)}}\leq 1 - (1-2p^2k)\left(1 - \frac{2kp^2e^{2kp}}{|S|\min_{v \in S} \pi(v)} \right), \] since $\frac{1}{1+x}\geq 1-x $ for any $x\geq 0$.  Now, as $|S|\min_{v \in S} \pi(v)\leq 1$ we have 
\[\Pru{\pi^k}{ \tau_{\mathsf{cov}}^{(k)} \leq  t} \leq 2kp^2 + \frac{2kp^2e^{2kp}}{|S|\min_{v \in S} \pi(v)} \leq \frac{4kp^2e^{2kp}}{|S|\min_{v \in S} \pi(v)}, \] which concludes the main proof.		To prove the claim, first observe that \begin{equation}\label{eq:Rexpression}\Exu{\pi}{R(t)^2} = \Exu{\pi}{\sum_{v\in S}\sum_{w\in S} \mathbf{1}(\tau_v\leq t )\mathbf{1}(\tau_w\leq t )} = \sum_{v\in S}\sum_{w\in S}  \Pru{\pi}{\tau_v\leq t, \tau_w\leq t}.\end{equation} Now, by the Markov property, for any pair $(v,w)\in V^2$ we have
\begin{equation}\label{eq:pibound}
	\begin{aligned}
		\Pru{\pi}{\tau_v\leq t, \tau_w\leq t} &\leq \Pru{\pi}{\tau_v\leq t}\Pru{v}{\tau_w\leq t} + \Pru{\pi}{\tau_w\leq t}\Pru{w}{\tau_v\leq t}\\
		&\leq  p\Pru{v}{\tau_w\leq t} + p\Pru{w}{\tau_v\leq t}.
	\end{aligned}\end{equation}
Thus by inserting the bound from \eqref{eq:pibound}	into \eqref{eq:Rexpression} we have \begin{align*}
		\Exu{\pi}{R(t)^2} &\leq p\sum_{v\in S}\sum_{w\in S}\Pru{v}{\tau_w\leq t} + p\sum_{w\in S}\sum_{v\in S}\Pru{w}{\tau_v\leq t}  \\
		&\leq p\sum_{v\in S} \Exu{v}{R(t)}+p\sum_{w\in S}\Exu{w}{R(t)}\nonumber\\
	&\leq p\sum_{v\in S} \Exu{v}{R(t)}\cdot \pi(v)\cdot(\min_{v \in S} \pi(v))^{-1}+p\sum_{w\in S}\Exu{w}{R(t)}\cdot \pi(w)\cdot (\min_{v \in S} \pi(v))^{-1}\nonumber\\
		&\leq  2p\Exu{\pi}{R(t)}(\min_{v \in S} \pi(v))^{-1}\\ &\leq 2p^2 |S|(\min_{v \in S} \pi(v))^{-1},
	\end{align*}as claimed.
\end{proof}

	\section{Mixing Few Walks to Cover Many Vertices}\label{Sec:partialmixing}
	
	In this section we present several bounds on $\tcov^{(k)}$, the multiple cover time from worst-case start vertices, based on $\tcov^{(k)}(\pi)$, the multiple cover time from stationarity, and a new notion that we call \emph{partial} mixing time. The intuition behind this is that on many graphs such as cycles or binary trees, only a certain number, say $\tilde{k}$ out of $k$ walks will be able to reach vertices that are ``far away'' from their start vertex. That means covering the whole graph $G$ hinges on how quickly these $\tilde{k}$ ``mixed'' walks cover $G$. However, we also need to take into account the number of steps needed to ``mix'' those. Theorem~\ref{thm:characterupper} makes this intuition more precise and suggests that the best strategy for covering a graph might be when $\tilde{k}$ is chosen so that the time to mix $\tilde{k}$ out of $k$ walks and the stationary cover time of $\tilde{k}$ walks are approximately equal. As in the previous section we shall first state our results before proving them in the final two subsections. 
	\subsection{Two Notions of Mixing for Multiple Random Walks}\label{sec:elementry}
	We begin by introducing the notion of \emph{partial mixing time}. For any graph $G$, and any $1 \leq \tilde{k} < k$, we define the \emph{partial mixing time}:
	\begin{equation}\label{def:partialMixingSep} \begin{aligned}
	t_{\mathsf{mix}}^{(\tilde{k},k)}(G) &= \inf\left\{t \geq 1 : \textit{there exists an SST  }\tau\text{ such that }\min_{v\in V}\Pru{v}{\tau \leq t} \geq \tilde{k}/k  \right\}\\
	&= \inf\left\{t \geq 1 : s(t)\leq 1-\tilde k/k  \right\}  .  
	\end{aligned}\end{equation}
	where SST stands for strong stationary time, and $s(t)$ is the separation distance (see Section~\ref{sec:notation}). We note that the two definitions above are equivalent by the following result. 
		\begin{proposition}[{\cite[Proposition 3.2]{AldDia}}] If $\sigma$ is an SST then $\Pr{\sigma >t}\geq s(t)$ for any $t\geq 0$. Furthermore there exists an SST for which equality holds.
	\end{proposition}

Finally, we can rewrite the second definition based on $s(t)$ as follows:
\begin{equation}\label{def:partialMixingSepnew}
 t_{\mathsf{mix}}^{(\tilde{k},k)}(G) = \inf \left\{ t \geq 1 \colon \forall u, v \in V \colon P_{u,v}^t \geq \tilde{k}/k \cdot \pi(v) \right\}.
\end{equation}

This notion of mixing, based on the idea of separation distance and strong stationary times for single walks, will be useful for establishing an upper bound on the worst-case cover time. For lower bounds on the cover time by multiple walk we will now introduce another notion of mixing for multiple random walks in terms of hitting probabilities of large sets. Before doing so, we recall a fundamental connection for single random walks which links mixing times with hitting times of large sets. In particular, let
	\[
	t_{\mathsf{H}}(\alpha) = \max_{u \in V, S \subseteq V \colon \pi(S) \geq \alpha} \Exu{u}{ \tau_{S} }, \qquad \text{and}\qquad t_{\mathsf{H}}:= t_{\mathsf{H}}(1/4),
	\]
then the following theorem shows this \textit{large-set hitting time} is equivalent to the mixing time.  
	\begin{theorem}[{\cite{OliMix} and independently \cite{PeresSousi}}]\label{thm:mixhit} Let $\alpha<1/2$.   Then  there  exist  positive  constants $c(\alpha)$ and $C(\alpha)$ so  that  for  every reversible chain \[ c(\alpha) \cdot t_{\mathsf{H}}(\alpha) \leq \tmix(\alpha) \leq C(\alpha) \cdot t_{\mathsf{H}}(\alpha).    \]
	\end{theorem}
	Inspired by that fundamental result, we introduce the following quantity, which will be used to lower bound the cover time of multiple random walks,
	\begin{align}
	t_{\mathsf{large-hit}}^{(\tilde{k},k)}(G)  = \min\left\{t \geq 1: \min_{u \in V,  S \subseteq V \colon \pi (S) \geq 1/4} \Pru{u}{ \tau_{S} \leq t} \geq \frac{\tilde{k}}{k}  \right\}.
	\label{eq:largehit}
	\end{align}

Note that both notions of mixing times are only defined for $\tilde k<k$. However, by the union bound, there exists a $C<\infty$ such that if we run $k$ walks for $C\tmix \log k$ steps then all $k$ walks will be close to stationarity in terms of TV-distance. Our next lemma generalises this fact. 	
	\begin{lemma}\label{lem:mixupper}There exists a constant $C<\infty$ such that for any graph and $1\leq \tilde{k}<k$ we have	\begin{enumerate}[(i)]\item\label{itm:maxa}$\displaystyle{\tmix^{(\tilde{k},k)} \leq 2\cdot \tmix\cdot  \left\lceil\log_2 \left(\frac{2k}{k-\tilde{k}}\right)\right\rceil } $ ,
	\item\label{itm:maxb} $\displaystyle{t_{\mathsf{large-hit}}^{(\tilde{k},k)} \leq C\cdot\tmix \cdot \left\lceil\log \left(\frac{k}{k-\tilde{k}}\right)\right\rceil } $.\end{enumerate}	
	\end{lemma}
	
	The partial mixing time can be bounded from below quite simply by mixing time. 
	\begin{lemma}For any graph and $1\leq \tilde{k}<k$ we have 	 
		\[\tmix^{(\tilde{k},k)} \geq \tmix \!\left(1-\frac{\tilde{k}}{k}\right).\]    
	\end{lemma}

	\begin{proof} This follows since $d_{\mathrm{TV}}(t)\leq s(t)$ holds for any $t\geq 0$ by \cite[Lemma 6.3]{levin2009markov}. 
	\end{proof}
We would prefer a bound in terms of $\tmix:=\tmix(1/4)$ instead of $\tmix(1-\tilde{k}/k)$ as the former is easier to compute for most graphs. The following lemma establishes such a lower bound for both notions of mixing time at the cost of a $\tilde{k}/k$ factor.   
	\begin{lemma}\label{lem:mixlowerbound} There exists some constant $c>0$ such that for any graph and $1\leq \tilde{k}<k$ we have \begin{enumerate}[(i)]\item\label{itm:mix1}$\displaystyle{\tmix^{(\tilde{k},k)} \geq c\cdot \frac{\tilde{k}}{k}\cdot \tmix}$, \item\label{itm:mix2} $\displaystyle{t_{\mathsf{large-hit}}^{(\tilde{k},k)}\geq c\cdot \frac{\tilde{k}}{k}\cdot \tmix}$.\end{enumerate}
	\end{lemma}

	We leave as an open problem whether our two notions of mixing for multiple random walks are equivalent up to constants, but the next result gives partial progress in one direction. 
	
\begin{lemma}\label{lem:mixingbdd} For any graph and $1\leq \tilde{k}<k/4$ we have \[ t_{\mathsf{large-hit}}^{(\tilde{k},k)} \leq \tmix^{(4\tilde{k},k)} + 1 \leq 2\tmix^{(4\tilde{k},k)} .\]
\end{lemma}	
	\subsection{Upper and Lower Bounds on Cover Time by Partial Mixing}
	Armed with our new notions of mixing time for multiple random walks from Section \ref{sec:elementry}, we can now use them to prove upper and lower bounds on the worst-case cover time in terms of stationary cover times and partial mixing times. We begin with the upper bound.
	\begin{theorem}\label{thm:characterupper}
		For any graph $G$ and any $1 \leq k \leq n$, 
		\[
		\tcov^{(k)}\leq 12 \cdot \min_{1 \leq \tilde{k} < k} \max \left(  t_{\mathsf{mix}}^{(\tilde{k},k)} ,\; t_{\mathsf{cov}}^{(\tilde{k})}(\pi ) \right).
		\]
	\end{theorem}
	This theorem improves on various results in \cite{Multi} and \cite{ER09} which bound the worst-case cover time by mixing all $k$ walks, and it also generalises a previous result in \cite[Lemma 3.1]{ElsSau}, where most walks were mixed, i.e., $\tilde{k} = k/2$.

	We also prove a lower bound for cover times, however this involves the related definition of partial mixing time based on the hitting times of large sets.

\begin{theorem}\label{thm:characterlower}
	For any graph $G$ with $\pi_{\max}=\max_{u} \pi(u)$ and any $1 \leq k \leq n$,
	\[
	\tcov^{(k)}\geq \frac{1}{16} \cdot  \max_{1 \leq \tilde{k} < k} \min \left(  t_{\mathsf{large-hit}}^{(\tilde{k},k)} ,\; \frac{1}{ \tilde{k} \pi_{\max}} \right).
	\]
	Further, for any regular graph $G$ any $\gamma>0$ fixed, there is a constant $C=C(\gamma)>0$ such that
	\[
	\tcov^{(k)} \geq C \cdot \max_{n^{\gamma} \leq \tilde{k} < k}  \min \left(  t_{\mathsf{large-hit}}^{(\tilde{k},k)} ,\; \frac{n \log n}{\tilde{k}}  \right).
	\]
\end{theorem}
	
	As we will see later, both Theorem~\ref{thm:characterupper} and Theorem~\ref{thm:characterlower} yield asymptotically tight (or tight up to logarithmic factors) upper and lower bounds for many concrete networks. To explain why this is often the case, note that both bounds include one non-increasing function in $\tilde{k}$ and one non-decreasing in function in $\tilde{k}$. That means both bounds are optimised when the two functions are as close as possible. Then balancing the two functions in the upper bound  asks for $\tilde{k}$ such that $t_{\mathsf{mix}}^{(\tilde{k},k)}  \approx \tcov^{(\tilde{k})}(\pi)$. Similarly, balancing the two functions in the first lower bound demands $t_{\mathsf{large-hit}}^{(\tilde{k},k)}  \approx n/\tilde{k}$ (assuming $\pi_{\max}=O(1/n)$). Hence for any graph $G$ where $t_{\mathsf{mix}}^{(\tilde{k},k)}  \approx t_{\mathsf{large-hit}}^{(\tilde{k},k)} $, and also  $\tcov^{(\tilde{k})}(\pi) \approx n/\tilde{k}$, the upper and lower bounds will be close. This turns out to be the case for many networks, as we will demonstrate in Section~\ref{fundamental}.

	One exception where Theorem \ref{thm:characterlower} is far from tight is the cycle, we shall also prove a min-max theorem but with a different notion of partial cover time which is tight for the cycle. 
	
	For a set $S\subseteq V$, we let $\tau^{(k)}_{\mathsf{cov}}(S)$ be the first time that every vertex in $S$ has been visited by at least one of the $k$ walks, thus $\tau_{\mathsf{cov}}^{(k)}(V) = \tau^{(k)}_{\mathsf{cov}}$. Then we define the set cover time \[t_{\mathsf{large-cov}}^{(k)}= \min_{S : \pi(S)\geq 1/4}\min_{\mu}\Exu{\mu^k}{\tau_{\mathsf{cov}}^{(k)}(S)},\] where the first  minimum is over all sets $S\subseteq V$ satisfying $\pi(S)\geq 1/4$ and the second is over all probability distributions $\mu$ on the set  $\partial S=\{x\in S : \text{exists }y\in S^c, \;xy \in E \}$.  
	\begin{theorem}\label{thm:cyclelower}
		For any graph $G$ and any $1 \leq k \leq n$,
		\[
		\tcov^{(k)}\geq \frac{1}{2} \cdot  \max_{1 \leq \tilde{k} < k} \min \left(  t_{\mathsf{large-hit}}^{(\tilde{k},k)} , t_{\mathsf{large-cov}}^{(\tilde{k})} \right).
		\]
	\end{theorem}

	\subsection{Geometric Lower Bounds on the Large-Hit and Large-Cover Times}
	
	We will now derive two useful lower bounds on $ t_{\mathsf{large-hit}}^{(\tilde{k},k)}$; one based on the conductance of the graph, and a second one based on the distance to a large set the random walk needs to hit.
	
	For two sets $A,B \subseteq V$ the \textit{ergodic flow} $Q(A,B)$ is given by $Q(A,B) = \sum_{a\in A, b\in B} \pi(a) P_{a,b},$ where $P$ denotes the transition matrix of a (lazy) single random walk. We define the \textit{conductance} $\Phi(S)$ of a set $S\subseteq V$ with $\pi(S) \in (0,1/2]$ to be \[\Phi(S)= \frac{Q(S,S^c)}{\pi(S)} \quad\text{ and let }\quad \Phi(G) = \min_{S\subseteq V, 0 < \pi( S) \leq 1/2 }\Phi(S).\]
	\begin{lemma}\label{conductancebd}
		For any graph $G$ with conductance $\Phi(G)$, and any $1 \leq \tilde{k} \leq k$, we have
		\[
		t_{\mathsf{large-hit}}^{(\tilde{k},k)}  \geq \frac{\tilde{k}}{k} \cdot \frac{2}{\Phi(G)}.
		\]
	\end{lemma}

	We remark that a similar bound to that in Lemma \ref{conductancebd} was used implicitly in \cite[Proof of Theorem~1.1]{S10}, which proved 
	$t_{\mathsf{cov}}^{(k)} \geq \sqrt{ \frac{n}{k \cdot \Phi(G)}}$. 
	
	The next lemmas are needed to apply Theorems \ref{thm:characterlower} \& \ref{thm:cyclelower} to cycles/tori. 
	
			\begin{lemma}\label{lem:largehit}
		Let $G$ be a $d$-dimensional torus with constant $d\geq 2$ (or cycle, $d=1$). Then, for any $u \in V$, $S\subseteq V$ satisfying $|S| \geq n/4$, and $\tilde{k} \leq k/2$, we have
		\[
		t_{\mathsf{large-hit}}^{(\tilde{k},k)}  = \Omega \left( \frac{\dist(u,S)^2}{\log( k/\tilde{k}) } \right).
		\]
	\end{lemma}

	\begin{lemma}\label{lem:obliviouslowerbound}
		Let $S \subseteq V$ be a subset of vertices with $\pi(S) \geq 1/4$, $t \geq 2$ be an integer and $k \geq 100$ such that $\sum_{s=0}^{t} P_{u,u}^s \geq 32 \cdot t \cdot \pi(u) \cdot k$ for all $u \in S$.	Then, for any distribution $\mu$ on $S$,
		\[
		\Exu{\mu^{k/8}}{\tau_{\mathsf{cov}}^{(k/8)}(S)} \geq t/5.
		\]
	\end{lemma}

	\subsection{Proofs of Lemmas in Section \ref{sec:elementry}}
	\begin{proof}[Proof of Lemma \ref{lem:mixupper}]We start with Item \eqref{itm:maxa}. Let $\bar d_\mathrm{TV}(t) = \max_{x,y \in V} \|P_{x,\cdot}^t -P_{y,\cdot}^t\|_{\mathrm{TV}}$. Then \[ d_{\mathrm{TV}}(t)\leq \bar d_\mathrm{TV}(t) \leq 2 d_{\mathrm{TV}}(t), \qquad\bar d_\mathrm{TV}(\ell t)\leq \bar d_\mathrm{TV}(t)^\ell \qquad \text{and}\qquad s(2t)\leq 1-(1-\bar d_\mathrm{TV}(t))^2 ,   \] hold for any integers $t,\ell\geq 0$ by Lemmas  4.11, 4.12 and 19.3, respectively, of \cite{levin2009markov}. 
		
		Thus if we take $\ell= \left \lceil \log_{2}\left(2k/(k-\tilde{k})\right)\right \rceil$ and $t = \tmix(1/4)$ then we have  \[ s(2\ell t) \leq 1-(1-\bar d_\mathrm{TV}(\ell t))^2 \leq  2\bar d_\mathrm{TV}(\ell t)\leq 2\bar d_\mathrm{TV}( t)^\ell \leq 2\cdot( 2 d_{\mathrm{TV}}( t))^\ell  \leq 2\cdot \left(\frac{1}{2}\right)^{\log_{2}\left(\frac{2k}{k-\tilde{k}}\right)} \leq 1- \frac{\tilde{k}}{k} , \] 
			 it follows that $ t_{\mathsf{mix}}^{(\tilde{k},k)}\leq 2\ell t \leq 2 \tmix \left \lceil \log_{2}\left(2k/(k-\tilde{k})\right)\right \rceil$.

		For Item \eqref{itm:maxb}, by the Markov property for any non-negative integers $\ell$ and $t$, we have 
		\[\max_{x\in V, S\subseteq V, \pi(S)\geq 1/4}\Pru{x}{\tau_S>\ell t}\leq \left(\max_{x\in V, S\subseteq V, \pi(S)\geq 1/4}\Pru{x}{\tau_S>t}\right)^\ell. \]
		By Markov's inequality we have  $\max_{x\in V, S\subseteq V, \pi(S)\geq 1/4}\Pru{x}{\tau_S>2t_{\mathsf{H}}}\leq 1/2$ and by Theorem \ref{thm:mixhit} there exists some $C$ such that $t_{\mathsf{H}}\leq C\tmix$. Thus, if we take $T = 2C\tmix \cdot \left \lceil \log_2\left(k/(k-\tilde{k})\right)\right \rceil$, then \[\max_{x\in V, S\subseteq V, \pi(S)\geq 1/4}\Pru{x}{\tau_S>T}\leq  \left(\max_{x\in V, S\subseteq V, \pi(S)\geq 1/4}\Pru{x}{\tau_S>2t_{\mathsf{H} }}\right)^{\log_2\left(k/(k-\tilde{k})\right)} \leq 1 -\frac{\tilde{k}}{k}.  \] It follows that $t_{\mathsf{large-hit}}^{(\tilde{k},k)} \leq T = C'\tmix \left\lceil\log\left(k/(k-\tilde{k})\right)\right\rceil$ for some $C'<\infty$.
	\end{proof}
	
	\begin{proof}[Proof of Lemma \ref{lem:mixlowerbound}]For Item \eqref{itm:mix1}, if we let $\ell= \tmix^{(\tilde{k},k)}+1$ then the separation distances satisfies $s(\ell)\leq 1-\tilde{k}/k$. Thus, by the definition of separation distance, for any pair of vertices $x,y\in V$ we have $P_{x,y}^\ell\geq (\tilde{k}/k) \cdot \pi(y)$. Thus for any $x\in V$ and set $S\subseteq V$ satisfying $\pi(S)\geq 1/4$ we have 
		\begin{equation}\label{eq:hitprob}
		\Pru{x}{\tau_S \leq \ell }\geq \sum_{y\in S} P_{x,y}^\ell \geq  \sum_{y\in S} \frac{\tilde{k}}{k} \cdot \pi(y) = \frac{\tilde{k}}{k}\sum_{y\in S}  \pi(y) = \frac{\tilde{k}}{k}\pi(S)\geq \frac{\tilde{k}}{4k}. 
		\end{equation}
		Since \eqref{eq:hitprob} holds for all $x\in V$ and $S\subseteq V$ where $\pi(S)\geq 1/4$, we have $ t_{\mathsf{H}} \leq \ell \cdot (4k/\tilde{k})$, as $\tau_S$ is stochastically dominated by $\ell$ times the number of phases of length $\ell$ before $S$ is hit. By Theorem \ref{thm:mixhit} there exists a universal constant $ C<\infty$ such that for any graph $\tmix \leq   C\cdot t_{\mathsf{H}}$, thus  
		\[\tmix \leq C\cdot t_{\mathsf{H}} \leq C\cdot \ell \cdot (4k/\tilde{k}) =\frac{4C k}{\tilde{k}} \left( \tmix^{(\tilde{k},k)}+1\right) \leq C'\cdot \frac{k}{\tilde{k}} \cdot \tmix^{(\tilde{k},k)},   \] for some universal constant $C'<\infty$ as $ \tmix^{(\tilde{k},k)}\geq 1$.
		
		For Item \eqref{itm:mix2}, observe that if instead we set $ \ell = t_{\mathsf{large-hit}}^{(\tilde{k},k)}$ then \eqref{eq:hitprob} still holds (in fact the stronger bound $\Pru{x}{\tau_S \leq \ell } \geq \tilde{k}/k$  holds) and the rest of the proof goes through unchanged. \end{proof}

	\begin{proof}[Proof of Lemma \ref{lem:mixingbdd}] If we let $\ell= \tmix^{(4\tilde{k},k)}+1$ then $s(\ell)\leq 1-4\tilde{k}/k$. So, $P_{x,y}^\ell\geq (4\tilde{k}/k) \cdot \pi(y)$ for any pair of vertices $x,y\in V$. Thus for any $x\in V$ and set $S\subseteq V$ satisfying $\pi(S)\geq 1/4$ we have 
		\begin{equation*}
		\Pru{x}{\tau_S \leq \ell }\geq \sum_{y\in S} P_{x,y}^\ell \geq  \sum_{y\in S} \frac{4\tilde{k}}{k} \cdot \pi(y) = \frac{4\tilde{k}}{k}\sum_{y\in S}  \pi(y) = \frac{4\tilde{k}}{k}\pi(S)\geq \frac{\tilde{k}}{k}. 
		\end{equation*}Consequently, we have $ t_{\mathsf{large-hit}}^{(\tilde{k},k)}\leq \ell = \tmix^{(4\tilde{k},k)} +1 \leq 2\tmix^{(4\tilde{k},k)},$ as claimed since $\tmix^{(4\tilde{k},k)}\geq 1$. 
	\end{proof}

	\subsection{Proofs of Upper and Lower Bounds for Covering via Partial Mixing}
	
	\begin{proof}[Proof of Theorem \ref{thm:characterupper}]
		Fix any $1 \leq \tilde{k} < k$.
		It suffices to prove that with $k$ walks starting from arbitrary positions running for 
		\[
		t= t_{\mathsf{mix}}^{(\tilde{k},k)}  + 2 \tcov^{(\tilde{k})}(\pi) \leq 3 \cdot \max \left( t_{\mathsf{mix}}^{(\tilde{k},k)} , \tcov^{(\tilde{k})}(\pi) \right)
		\]
		steps, we cover $G$ with probability at least $1/4$. Consider a single walk $X_1$ on $G$ starting from $v$. From the definition of $s_v(t)$ we have that at time $T= t_{\mix}^{(\tilde k, k)}$ there exists a probability measure $\nu_v$ on $V$ such that, 
		$$P_{v,w}^T = (1-s_v(T))\pi(w)+s_v(T)\nu_v(w).$$
		Now, note that \eqref{def:partialMixingSep} yields $(1-s_v(T))\geq \tilde k/k$, therefore, we can generate $X_1(T)$ as follows: with probability $1-s_v(T)\geq \tilde k/k$ we sample from $\pi$, otherwise we sample from $\nu_v$. If we now consider $k$ independent walks, the number of walks that are sampled at time $T$ from $\pi$ has a binomial distribution $\text{Bin}(k,\tilde{k}/k)$ with $k$ trials and probability $\tilde k/k$, whose expectation is $\tilde k$. Since the expectation $\tilde{k}$ is an integer it is equal to the median. Thus, with probability at least $1/2$, at least $\tilde k$ walks are sampled from the stationary distribution. 
		Now, consider only the $\tilde k$ independent walks starting from $\pi$. After  $2t_{\mathsf{cov}}^{(\tilde k)}(\pi)$ steps, these walks will cover $G$ with probability at least $1/2$, due to Markov's inequality.
		
		We conclude that in $t$ time steps, from any starting configuration of the $k$ walks, the probability we cover the graph is at least $1/4$. Hence in expectation, after (at most) 4 periods of length $t$ we cover the graph. 
	\end{proof}
	
	\begin{proof}[Proof of Theorem \ref{thm:characterlower}] By the definition of $ t_{\mathsf{large-hit}}^{(\tilde{k},k)} $, there exists a vertex $u$ and $S\subseteq V$ such that $\Pru{u}{\tau_{S}\leq t_{\mathsf{large-hit}}^{(\tilde{k},k)}-1} < \tilde k /k$. For such a vertex $u$, we consider $k$ walks, all started from $u$, which run for
		$t_{\mathsf{large-hit}}^{(\tilde{k},k)}-1$ steps. It follows that the number of walks that hit $S$ before time $t_{\mathsf{large-hit}}^{(\tilde{k},k)}$ is dominated by a binomial distribution with parameters $k$ and $p = \tilde k/k$, whose expected value and median is $\tilde k$. We conclude that with probability at least $1/2$, at most $\tilde k$ walks hit $S$.
		Note that $|S| \geq \pi(S) /\pi_{\max} \geq 1/(4 \pi_{\max})$. Hence even if all $\tilde{k}$ walks that reached $S$ before time $t_{\mathsf{large-hit}}^{(\tilde{k},k)}$ were allowed to run for exactly $1/(8 \tilde{k} \pi_{\max})-1$ steps (if a walk exits $S$, we can completely ignore the steps until it returns to $S$), then the total number of covered vertices in $S$ would be at most
		\[ \tilde{k} \cdot \left(\frac{1}{8 \tilde{k} \pi_{\max}} \right) < 1/(4 \pi_{\max}) \leq |S|,
		\]
		which concludes the proof of the first bound.
		
		For the second bound, we follow the first part of the proof before and consider again at most $ \tilde{k}$ walks that reach the set $S$.
		Let us denote by $\kappa$ the induced distribution over $S$ upon hitting $S$ from $u$ for the first time. Now, each of the $\tilde{k}$ walks continues for another $\ell=\lfloor \epsilon (n/\tilde{k}) \log(n) \rfloor$ steps, where $0<\epsilon=\epsilon(\gamma) <1$ is a sufficiently small constant fixed later. We now define, for any $v \in S$, the probability $
		p_{v} = \Pru{\kappa}{ \tau_v \leq \epsilon (n/\tilde{k}) \log n }.$
		Observe that since a walk of length $\ell$ can cover at most $\ell$ vertices, we have $\sum_{v \in S} p_v \leq \ell$. Further, define the set $\widetilde{S}=\left\{ v \in S \colon p_v < 8 \ell/n \right\}$. 
		
		Since for every $v$ with $p_v\geq 8 \ell/n$ we have that $np_v/(8 \ell)\geq 1$ and
		\begin{align*}
		|S \setminus \tilde{S}| = \left|\left\{v \in S: p_v \geq 8\ell/n\right\} \right| \leq \sum_{v \in S} \frac{np_v}{8\ell} \leq \frac{n}{8} \leq \frac{|S|}{2},
		\end{align*}
		where the last inequality holds since $G$ is regular. Thus $\tilde{S} \geq |S|/2$. Now, let $Z$ be the number of unvisited vertices in $\widetilde{S}$ after we run $\tilde{k}$ random walks starting from $\kappa$. If $p_*=\max_{v\in \tilde{S}}p_v<8\ell/n$ then, by definition and recalling that $\ell=\lfloor \epsilon (n/\tilde{k}) \log(n) \rfloor$ where $0<\eps<1$, we have
		\begin{align*}
		\Ex{Z} \geq |\widetilde{S}| \cdot \left(1 - p_* \right)^{\tilde{k}} \geq (n/2) \cdot e^{-\tilde{k}\cdot p_* }  \cdot \left(1 - \tilde{k} p_{*}^2 \right) \geq 
		n/2 \cdot e^{- 8 \epsilon \log(n) }
		\cdot \left(1 - \frac{(8 \log n)^2}{\tilde{k}}  \right) \geq
		n^{1-9 \eps },
		\end{align*}
		where the second inequality holds by \eqref{eq:cheatsheet} and the last since $\tilde{k}\geq n^{\gamma}  $ where $\gamma>0$ fixed.
		
		Finally, since each of these $\tilde{k}$ random walks can change $Z$ by at most $\ell$ vertices, by the method of bounded differences \cite[Theorem 5.3]{DPmobd},
		\[\Pr{Z<\Ex{Z}/2}\leq \exp\left(-2\frac{(Z/2 )^2 }{\tilde{k}\ell^2 }  \right) \leq \exp\left(- \frac{ n^{2-18\epsilon} \tilde{k} }{8 \eps^2 n^2 \log^2 n }  \right)  < 1/2 ,  \] provided we have $\epsilon <  \gamma/18$ as $\tilde{k} \geq n^{\gamma}$. This implies that $\Pr{Z \geq 1} \geq 1/2$. 
	\end{proof}
	
	\begin{proof}[Proof of Theorem \ref{thm:cyclelower}]Since we are bounding the worst-case cover time from below we can assume that all walks start from a single vertex $u$. First, consider the $k$ walks running for $ t_{\mathsf{large-hit}}^{(\tilde{k},k)}  -1$ steps. Then, by the definition of $ t_{\mathsf{large-hit}}^{(\tilde{k},k)} $, there exists a vertex $u$ and $S\subseteq V$ such that $\Pru{u}{\tau_{S}\leq t} < \tilde k /k$, therefore, the number of walks that, starting from $u$, hit $S$ before time $t$ is dominated by a binomial distribution with parameters $k$ and $p = \tilde k/k$, whose expected value and median is $\tilde k$. We conclude that, if $\mathcal{E}$ is the event that at most $\tilde k$ walks hit $S$ by time $ t_{\mathsf{large-hit}}^{(\tilde{k},k)}-1$, then $\Pr{\mathcal{E}}\geq 1/2$. 
		
		Although, conditional on $\mathcal{E}$, we know at most $\tilde{k}$ walks hit $S$ by time $ t_{\mathsf{large-hit}}^{(\tilde{k},k)} -1$, we do not know when they arrived or which vertices of $S$ they hit first. For a lower bound we assume these $\tilde{k}$ walks arrived at time $0$ and then take the minimum over all sets $S$ such that $\pi(S)\geq 1/4$ and starting distributions $\mu$ on $\partial S$, the vertex boundary of $S$ (note all particles started from $u$, so they have the same distribution when they enter $S$ for first time). It follows that, conditional on $\mathcal{E}$, the expected time for $\tilde{k}$ walks which hit $S$ before time $ t_{\mathsf{large-hit}}^{(\tilde{k},k)} $ to cover $S$ is at least $t_{\mathsf{large-cov}}^{(\tilde k)}$ and so the result follows. 
	\end{proof}

	\subsection{Proofs of Bounds on the Large-Hit and Large-Cover Times}
	\begin{proof}[Proof of Lemma \ref{conductancebd}]
		Let $S \subseteq V$ be a set such that $\Phi(G)=\Phi(S)$ and $\pi(V \setminus S) \geq 1/2$ (such a set exists by symmetry of $\Phi(S)=\Phi(V \setminus S)$). Let $\pi_{S}$ be the stationary distribution restricted to $S$, that is $\pi_{S}(s)= \pi(s)/\pi(S)$ for $s\in S$ and $\pi_{S}(x)=0$ for $x\notin S$. As shown in \cite[Proposition~8]{ST12}, the probability that a (single) random walk $X_t$ remains in a set $S$ when starting from a vertex in $S$ sampled from $\pi_{S}$, within $t$ steps is at least $ (1-\Phi(S)/2)^{t} \geq 1 - \Phi(S) t/2 $. Let $t<\frac{2\tilde k}{k\Phi(S)}$, then $\Pru{\pi_S}{\tau_{S^c} \leq t} < \tilde k/k$, and thus by taking $u$ as the vertex that minimises the escape probability from $S$, we conclude that  $t_{\mathsf{large-hit}}^{(\tilde{k},k)}  \geq \frac{\tilde{k}}{k} \cdot \frac{2}{\Phi(S)}$. 
	\end{proof}

	Before proving Lemma~\ref{lem:largehit} we first establish an elementary result.
	
	\begin{lemma}[{cf.~\cite[Theorem 13.4]{ProbNetworks}}]\label{lem:carne}
		Let $G=(V,E)$ be a $d$-dimensional torus ($d \geq 2$) or cycle ($d=1$). Then for any $ D\geq 0$ and  $t \geq 1$,
		\[
		\Pr{ \max_{1 \leq s \leq t} \operatorname{dist}(X_0,X_s) \geq  D } \leq 2d \cdot  \exp\left( -\frac{D^2}{2td^2} \right).
		\]
	\end{lemma}
	\begin{proof}[Proof of Lemma \ref{lem:carne}] 
		Consider a random walk for $t$ steps on a $d$-dimensional torus (or cycle), where $d\geq 1$.
		Let $Z_1,Z_2,\ldots,Z_{t} \in \{-1,0,+1\}$ be the transitions along the first dimension, and let $S_i = Z_1+\ldots+ Z_i$. Note that $S_i$ is a zero-mean martingale with respect to the $Z_i$. Define $\tau=\min\{i: |S_i|\geq D/d\}\wedge t$, which is a bounded stopping time, and thus $S_{\tau \wedge i}$ is another martingale with increments bounded by 1. Then, by Azuma's inequality \cite[Theorem 13.4]{PandC},
		\begin{align*}
		\Pr{\tau\leq t} = \Pr{|S_{\tau \wedge t}| \geq D/d} \leq 2\exp\left( - \frac{(D/d)^2 }{ 2t } \right)
		\end{align*}
		
		Now, for the random walk, in order to overcome a distance $D$ during $t$ steps, the above event must occur for at least one of the $d$ dimensions, so by the union bound
		\begin{align*}
		\Pr{ \max_{1\leq s \leq t}\dist(X_0,X_s) \geq D } \leq 2d \cdot \exp \left( - \frac{D^2}{2t d^2 } \right),
		\end{align*}as claimed.
	\end{proof}

	\begin{proof}[Proof of Lemma~\ref{lem:largehit}]
	Recall from the definition of $t_{\mathsf{large-hit}}^{(\tilde{k},k)}$ that there must exist a vertex $u$ and set $S$ with $\pi(S)\geq 1/4$ such that the probability a random walk of length $t_{\mathsf{large-hit}}^{(\tilde{k},k)}$ started from $u$ has hit $S$ is at most $\tilde k/ k $. Notice that since the $d$-dimensional torus is regular any set $S$ of size at least $n/4$ satisfies $\pi(S)\geq 1/4$. In order for a random walk to hit the set $S$ within $t$ steps, it must reach a distance $D=\dist(u,S)$ from $u$ at least once during $t$ steps. Let $t$ be given by
		$$t = \left\lfloor \frac{D^2}{2d^2\log\left(2dk/\tilde k\right)}\right\rfloor.$$ Then by Lemma~\ref{lem:carne},
		\begin{align*}
		\Pru{u}{ \tau_{S} \leq t } &\leq \Pr{ \max_{1 \leq s \leq t} \dist(u,X_t) \geq D }\leq 2d \cdot \exp\left( -\frac{D^2}{2td^2} \right) \leq \frac{\tilde{k}}{k}, 
		\end{align*}
		and the result follows. 
	\end{proof}
	
	\begin{proof}[Proof of Lemma \ref{lem:obliviouslowerbound}]
		In the first part of the proof we will work with random walks whose lengths are independent samples from $\geo{1/t}$. That is, we consider walks $(X_0,X_1,\ldots, X_{L-1})$ where $L\geq 1$ is a geometric random variable with mean $t$, which is independent of the trajectory $(X_s)_{s\geq 0 }$. Let $\tilde P_{w,u}^s = \Pru{w}{X_s = u, s<  L}$. We call the above a geometric random walk of expected length $t$. We consider a collection of $k$ independent geometric random walks of length $t$, in particular the lengths of these walks are also independent of each other.
		
		Let us lower bound the expected number of unvisited vertices in $S$ by $k$ independent geometric random walks of expected length $t$. Define a subset $S' \subseteq S$ as
		\[
		S' = \left\{ u \in S \colon ~\sum_{s=0}^{\infty} \sum_{w \in V} \mu(w) \tilde{P}_{w,u}^s \leq 8 t \cdot \pi(u) \right\}.
		\]
		Since the geometric random walk has an expected length of $t$, it visits at most $t$ vertices in expectation and thus
		\[
		\sum_{u \in S} \sum_{s=0}^{\infty} \sum_{w \in V} \mu(w) \tilde{P}_{w,u}^s \leq t.
		\]
		It follows by definition of $S'$ that 
		\[
		t \geq \sum_{u \in V \setminus S'} \sum_{s=0}^{\infty} \sum_{w \in V} \mu(w) \tilde{P}_{w,u}^s \geq \sum_{u \in V \setminus S'} 8 t \cdot \pi(u),
		\]
		and thus
		$\sum_{u \in V \setminus S'} \pi(u) \leq 1/8$. Hence $\sum_{u \in S'} \pi(u) \geq \pi(S)-1/8 \geq 1/4-1/8 = 1/8$.

		Let $Z_t=Z_t(u)$ denote the number of visits to $u$ by a geometric random walk of expected length $t$. 
		The probability a single walk starting from $\mu$ visits a vertex $u \in S'$ before being killed is at most
		\begin{align*}
		\Pru{\mu}{ Z_t \geq 1} &= \frac{ \Exu{\mu}{Z_t}}{\Exu{\mu}{Z_t \, \mid \, Z_t \geq 1}} 
		= \frac{ \sum_{s=0}^{\infty} \sum_{w \in V} \mu(w) \tilde{P}_{w,u}^s   }{ \sum_{s=0}^{\infty} \tilde{P}_{u,u}^s  },  \end{align*}since conditional on the walk having reached a vertex $u$, the expected remaining returns before getting killed is equal to $\sum_{s=0}^{\infty} \tilde{P}_{u,u}^s$. Now observe that $\tilde{P}^{s}_{u,u} \geq \frac{1}{4} \cdot P_{u,u}^{s}$, which follows since $\Pr{ \geo{1/t} > s} =  \left(1 - \frac{1}{t} \right)^{s} \geq \left(1 - \frac{1}{t} \right)^{t}\geq 1/4$, for any $t\geq 2$ and any $s\leq t$. Thus we have  \begin{align*}
			\Pru{\mu}{ Z_t \geq 1} &\leq \frac{ \sum_{s=0}^{\infty} \sum_{w \in V} \mu(w) \tilde{P}_{w,u}^s   }{ \frac{1}{4}\cdot \sum_{s=0}^{t}  P_{u,u}^s  }
	   \leq \frac{8 t \cdot \pi(u)}{\frac{1}{4} \cdot( 32 \cdot t\cdot \pi(u) \cdot k)}  =\frac{1}{k},
		\end{align*}by hypothesis. Let $Y$ the stationary mass of the unvisited vertices in $S'$, then
		\[
		\Ex{ Y } \geq \sum_{u \in S'} \pi(u) \cdot \left(1 - \frac{1}{k} \right)^{k} \geq \frac{1}{4} \cdot \pi(S').  
		\]
		Hence with probability at least $1/4$, at least one vertex in $S'$ remains unvisited by the $k$ random walks whose length is sampled from $\geo{1/t}$. Finally, since each walk is independent the number of walks which run for more than $t$ steps is binomially distributed with parameters $k$ and $p= (1 - 1/t)^t \geq 1/4$. Thus  by a Chernoff bound the probability that less than $1/8$ of the $k$ random walks run for more than $t$ steps is at most $\exp\left(-\frac{(k/8)^2}{(2k/4)} \right)= e^{-k/32} $. Hence, by coupling, $k/8$ random walks of length $t$ do not visit all vertices in $S'$ with probability at least $1/4-e^{-k/32} \geq 1/5$, provided $k\geq 100$.
	\end{proof}

	\section{Applications to Standard Graphs}\label{fundamental}

	In this section we apply the results of the previous sections to determine the stationary and worst-case multiple walk cover times. Firstly, we determine the stationary cover times for many fundamental networks using results from Sections \ref{ReturnSec} and \ref{Sec:partialmixing}. Secondly, using our results for the stationary cover times, we then apply them to the  $\min$-$\max$ (and $\max$-$\min$) characterisations from Section~\ref{Sec:partialmixing}. 
	Along the way, we also have to derive bounds for the partial mixing time and the time to hit a large set. Due to this section being large, the proofs of all results are located in the same subsections as the statements (unlike Sections \ref{ReturnSec} and \ref{Sec:partialmixing}). For a quick reference and comparison of the results of this section the reader is encouraged to consult Table \ref{tbl:results}. 
	
	\subsection{The Cycle} \label{sec:cycle}
	
	Our first result determines the stationary cover time of the cycle up to constants. This result comes from Theorem \ref{nonregbdd} and Lemma \ref{lemma:generalLowerboundpi} along with some additional results and arguments. 
	
	\begin{theorem}\label{cyclethm}	For the $n$-vertex cycle $C_n$, and any integer $ k \geq 2$, we have 
		\[\tcov^{(k)}(\pi)= \BT{\left(\frac{n}{k}\right)^2\log^2 k}.\]
	\end{theorem}
	The lower bound for $\tcov^{(k)}(\pi)$ provided $k\geq n^{1/20}$ was already known  \cite[Lemma 18]{KKPS}. Here we prove Theorem \ref{cyclethm}, which holds for any $k\geq 2$, by extending the applicable range of $k$ for the lower bound and supplying a new upper bound. We also demonstrate how to fully recover the worst-case cover time below using our new methods from Section \ref{Sec:partialmixing}.

	\begin{theorem}[{\cite[Theorem 3.4]{Multi}}]\label{cycleworst}
		For the $n$-vertex cycle $C_n$,  and any $2 \leq k \leq n$, we have 
		\[
		t_{\mathsf{cov}}^{(k)} = \Theta \left( \frac{n^2}{\log k} \right).
		\]
	\end{theorem}   
	
	We begin with the proof for the stationary case. 
	
	\begin{proof}[Proof of Theorem \ref{cyclethm}] We split the analysis for $k\geq 2$ into two cases.

	\noindent\textbf{Case (i)} [$2< k\leq n^{1/20}$]\textbf{:} The lower bound is covered by \cite[Lemma 18]{KKPS} so we just prove the upper bound. For a single walk $\tcov(\pi) = \mathcal{O}(n^2)$ \cite[Proposition 6.7]{aldousfill}  and so since $k$ walks take at most as long to cover as a single walk, we can assume that $k\geq 10000$ when proving the upper bound. 
	
To begin, divide the cycle as evenly as possible into $k$ disjoint intervals $\mathcal{I}_1,\dots, \mathcal{I}_k$ of consecutive vertices, each of size $\lfloor n/k\rfloor$ or $\lceil n/k\rceil$. For $2 \leq  c \leq k/\log k $ let $t^*(c)=t^* = \lceil(cn/k)^2\log^2 k\rceil$. Now, for each interval $\mathcal{I}_i$ we let $\mathcal{J}_i(c)$ be an interval of length $\ell=\lfloor \sqrt{t^*(c)} \rfloor  \leq n$ centred around $\mathcal{I}_i$. Note that since $c\geq 2$ and $n$ is large we have $\mathcal{I}_i\subset \mathcal{J}_i $ for each $1\leq i\leq k$. 
		\begin{claim}\label{clm} For any $2 \leq  c \leq k/\log k $, a walk of length $t^*(c)$ starting at any vertex in the interval $\mathcal{J}_i(c)$ will visit all vertices of $\mathcal{I}_i$ with probability at least $1/250$ when $n$ is suitably large.
		\end{claim}
			\begin{pocd}{\ref{clm}}	Let $\mathcal{N}(0,1)$ be the standard normal distribution, then by \cite{Mills}  for any $x>0$: 
				\begin{equation}\label{eq:mills}
					\Pr{\mathcal{N}(0,1)> x} \geq  \frac{x}{x^2+1} \cdot \frac{1}{\sqrt{2\pi}}e^{-x^2/2}     .
				\end{equation}
	Note that all vertices of $\mathcal I_i$ will have been covered if the walk has traveled from one endpoint of $\mathcal J_i$ to the other via a path contained within $\mathcal J_i$. 	Let $S_j$ be the distance of a random walk at time $j$ from its start point. Then, for large $n$, by the Central Limit Theorem and \eqref{eq:mills} we have
			\begin{equation}\label{eq:prob}\Pr{\frac{S_{ t^*/2}}{\sqrt{t^*/2}} > \frac{\ell}{\sqrt{t^*/2}} } \geq  \left(1-o(1)\right) \Pr{\mathcal{N}(0,1)> \sqrt{2}}\geq \left(1-\lo{1}\right)\frac{\sqrt{2}}{3} \cdot \frac{ e^{-1}}{\sqrt{2\pi }} > \frac{1}{15}.    \end{equation}  Now, by symmetry and \eqref{eq:prob}, regardless of its start point within the interval $\mathcal J_i$, with probability at least $1/15$ the walk will have hit the `left' end of $\mathcal J_i$ within at most $t^*/2$ steps.
			Once at the `left' end of $\mathcal J_i$, then again by \eqref{eq:prob} with probability at least $1/15$, it will have reached the right end via a path though $\mathcal J_i$ within at most $t^*/2$ additional steps.
			Thus, for suitably large $n$, with probability at least $(1/15)^2 >1/250$ a walk of length $t^*$ starting in $\mathcal J_i$ will cover $\mathcal I_i$.
		\end{pocd}Let $w_i$ be the number of walks which start in $\mathcal{J}_i$. By Chernoff's bound \cite[Theorem 4.5]{PandC}: 
	\begin{equation}\label{intervalprob}\Pru{\pi^k}{w_i< \frac{3}{4}\cdot c\log k  }=\Pr{\bin{k}{\frac{c\log k }{k} }< \frac{3}{4}\cdot c\log k   } \leq e^{-\frac{(1/4)^2}{2}\cdot  c\log k }= k^{-c/32}.\end{equation}  
		By Claim \ref{clm}, conditional on $w_i$, none of the walks in $\mathcal{J}_i$ cover $\mathcal{I}_i$ w.p.\ at most $(\frac{249}{250})^{w_i}$ when $n$ is large. 
		Hence, by \eqref{intervalprob} a fixed interval $\mathcal{I}_i$ is not covered w.p.\ at most $(249/250)^{(3/4)\cdot c\log k} + k^{-c/32} $. 
		As $(3/4)\ln(249/250)<-3/1000$ and $k\geq 10000$, for any $1000\leq  c \leq k/\log k $, we have  \begin{equation}\label{eq:covt*}\Pru{\pi^k}{\tau_{\mathsf{cov}}^{(k)}>t^*(c)}\leq k\left(k^{-3c/1000} + k^{-c/32}\right) \leq 2k^{1-3c/1000}\leq k^{-c/1000}, \end{equation} by the union bound. Now, as $t^*(k/\log k) = n^2$, observe that we have
		 \begin{equation}\label{eq:splitup}	    	\Exu{\pi^k}{\tau_{\mathsf{cov}}^{(k)} }  \leq t^*(1000)  +\sum_{i=1000}^{\lfloor \frac{k}{\log k }\rfloor}  [t^*(i+1)-t^*(i)]\cdot \Pru{\pi^k}{\tau_{\mathsf{cov}}^{(k)}>t^*(i)}  + \sum_{t=n^2}^\infty \Pru{\pi^k}{\tau_{\mathsf{cov}}^{(k)}>t}. 
		\end{equation} Using $\Pru{\pi^k}{\tau_{\mathsf{cov}}^{(k)}>t^*(i)} \leq k^{-i/1000} $ for $1000\leq  i \leq k/\log k $ by \eqref{eq:covt*}, gives
		\begin{equation}\label{eq:sumup1}\begin{aligned}\sum_{i=1000}^{\frac{k}{\log k }}  [t^*(i+1)-t^*(i)]\cdot \Pru{\pi^k}{\tau_{\mathsf{cov}}^{(k)}>t^*(i)} &\leq\left(\left(\frac{n}{k}\right)^2\log^2 k + 1 \right) \sum_{i=0}^{\infty } (2i+1) \cdot k^{-i/1000} \\&= \BO{\left(\frac{n}{k}\right)^2\log^2 k}.\end{aligned}\end{equation}
Recall that $\Pru{\pi^k}{\tau_{\mathsf{cov}}^{(k)}>t^*(k/\log k) } \leq k^{- \frac{1}{1000}\cdot k/\log k}=e^{-k/1000}$ by \eqref{eq:covt*}. If $\mathbf{v}$ is the worst-case start position vector for $k$ walks to cover a cycle then for any $t= i\cdot n^2$ and $i\geq 1$, 
\[ \Pru{\pi^k}{\tau_{\mathsf{cov}}^{(k)}>t }\leq\Pru{\pi^k}{\tau_{\mathsf{cov}}^{(k)}>t^*(k/\log k) }\cdot \Pru{\mathbf{v}}{\tau_{\mathsf{cov}}^{(k)}>t- n^2 } \leq e^{-k/1000}\cdot  \Pru{\mathbf{v}}{\tau_{\mathsf{cov}}^{(k)}>(i-1)n^2 },\]by the Markov property. Note that $\Pru{\mathbf{v}}{\tau_{\mathsf{cov}}^{(k)}>(i-1)n^2 }\leq (1/2)^{i-1}$ by Markov's inequality since $\tcov = n(n-1)/2$. Thus the second sum on the RHS of \eqref{eq:splitup} satisfies 	
\begin{equation}\label{eq:sumup2}\sum_{t=n^2}^\infty \Pru{\pi^k}{\tau_{\mathsf{cov}}^{(k)}>t} \leq n^2\cdot  \sum_{i=1}^\infty \Pru{\pi^k}{\tau_{\mathsf{cov}}^{(k)}>i\cdot n^2} \leq n^2\cdot  \sum_{i=1}^\infty e^{-k/1000}2^{-i+1}  = 2e^{-k/1000}n^2 .  \end{equation}
Case ($i$) then follows by inserting \eqref{eq:sumup1} and \eqref{eq:sumup2} into \eqref{eq:splitup}.

	\noindent\textbf{Case (ii)} [$k> n^{1/20}$]\textbf{:} The upper bound follows from Theorem \ref{nonregbdd} since in this case $\log k=\Omega(\log n )$. By Lemma \ref{torireturns} we have $\sum_{i=0}^{\lfloor t/2\rfloor}P_{v,v}^{2i} =\BOhm{\sqrt{t}}$ so, by Lemma \ref{covfromvisits} \eqref{returns},\begin{equation}\label{eq:pbound} \Pru{\pi}{\tau_v\leq t}=\BO{\sqrt{t}/n}.\end{equation} 
 
We shall now apply Lemma \ref{lemma:generalLowerboundpi} with $S=V$ and so if we set $t= (cn/k)^2\log^2 k$ for a suitably small constant $c>0$ then $p=\max_{v\in S}\Pru{\pi}{\tau_v\leq t} \leq  \frac{\log n}{100k},$ by \eqref{eq:pbound}. Observe that $|S|\min_{v\in S}\pi(v) = 1$ and also  $2p^2k < 1$ since $k\geq n^{1/20}$ and $n$ is large. Thus Lemma \ref{lemma:generalLowerboundpi} gives \[\Pru{\pi^k}{ \tau_{\mathsf{cov}}^{(k)} \leq  t} \leq \frac{4kp^2e^{2kp}}{|S|\min_{v\in S}\pi(v)} \leq  4k\left(\frac{\log n}{100k}\right)^2\cdot e^{2 k \cdot \frac{\log n}{100k}  } = \frac{4(\log n)^2 }{100^2 k }\cdot n^{1/50} = o\left(1\right) , \] as $k\geq n^{1/20}$, which completes Case ($ii$) and finishes the proof. \end{proof}

The final element we need for our analysis is to identify the partial mixing time of the cycle. For such, we provide bounds for the partial mixing time for all $d$-dimensional torus $\mathds{T}_d$ (which are going to be used later), and we recall the cycle is the $1$-dimensional torus $\mathds{T}_1$. 
      
 	 	\begin{lemma}\label{lem:cyclemixbdd} For any integer $d\geq 1$ there exists a constant $C_d<\infty$ such that for any $1\leq \tilde{k} \leq k/2 = \BO{n}$ the partial mixing time of the $n$-vertex torus $\mathds{T}_d$ satisfies $t_{\mathsf{mix}}^{(\tilde k,k)}  \leq C_d \cdot  n^{2/d}/\log(k/\tilde{k}).$
	 \end{lemma}
	 \begin{proof}
	 	Let $Q$ and $P$ be the transition matrices of the lazy walk on the (infinite) $d$-dimensional integer lattice $\mathbb{Z}^d$ and the (finite) $d$-dimensional $n$-vertex torus $\mathds{T}_d$, respectively. By \cite[Theorem 5.1 (15)]{HSC}, for each $d\geq 1$ there exist constants $C,C',C''>0$ such that for any $t\geq 1$ and $u,v \in \mathbb{Z}^d$ satisfying $||u -v||_{2} \leq t/C''$ we have  \begin{equation}\label{LCLT1} Q_{u,v}^t \geq  \left(\frac{C}{t}\right)^{d/2}\exp\left(-C'\cdot\frac{ ||u-v||_2^2}{t} \right) .  \end{equation}
	For any $t\geq 0$ and $u,v\in V(\mathds{T}_d) $ we have $ P_{u,v}^t \geq  Q_{u,v}^t$ and $\| u - v \|_{\infty} \leq n^{1/d}/2$. Therefore, 
	 	\begin{equation}\label{eq:bddondist}\max_{u,v \in V(\mathds{T}_d)} || u-v||_2^2 \leq d\cdot(n^{1/d}/2)^2  = dn^{2/d}/4, \end{equation} and so if we set  \begin{equation}\label{eq:defoft}
	 		t=  \left\lceil \frac{  dC'  \cdot  n^{2/d}}{4\log(k/\tilde{k})}\right\rceil 
	 		\leq \frac{dC'\cdot  n^{2/d}}{2 \log(k/\tilde{k})}.\end{equation}
	 then, for large $n$ and any $u,v \in V(\mathds{T}_d)$, we have $||u -v||_{2} \leq t/C''$ and thus, by \eqref{LCLT1} and \eqref{eq:bddondist}, 
	 	\begin{equation*}P_{u,v}^t \geq    \left(\frac{C}{t}\right)^{d/2}\cdot\exp\left(-C'\cdot  \frac{dn^{2/d}/4}{t} \right)  \geq  \frac{1}{n}  \left(\frac{C2\log(k/\tilde{k})}{dC'}\right)^{d/2}\cdot \exp\left(- \log\left(\frac{k}{\tilde{k}}\right) \right)\geq c'\cdot \frac{\tilde{k}}{n k}  \end{equation*} for some $c':= c'(d)$, as $ \log( k/\tilde{k}) \geq \log 2 $, which holds by hypothesis. It follows from the definition of separation distance that $s(t)\leq 1 - c'\cdot \frac{\tilde{k}}{n k}$ for $t$ given by \eqref{eq:defoft}. Note that if $c'\geq 1$ then the statement of the lemma follows by \eqref{def:partialMixingSep}, the definition of $t_{\mathsf{mix}}^{(\tilde k,k)}$. Otherwise, assuming $c' <1$, if we take $t' = \lceil 2/c'\rceil \cdot t$ then as separation distance is sub-multiplicative \cite[Ex.\ 6.4]{levin2009markov} we have
	 	\[s(t') \leq s(t)^{\lceil 2/c'\rceil }\leq \left(1-    \frac{c'\tilde{k}}{n k}\right)^{\lceil 2/c'\rceil }\leq \frac{1}{1 +\lceil 2/c'\rceil\cdot \frac{c'\tilde{k}}{n k} }  = 1 -  \frac{ \lceil 2/c'\rceil\cdot \frac{c'\tilde{k}}{n k} }{1 +\lceil 2/c'\rceil\cdot \frac{c'\tilde{k}}{n k} }\leq 1- \frac{\tilde{k}}{n k},  \] for suitably large $n$ where in the second to last inequality we have used the fact that $(1+x)^r \leq \frac{1}{1-rx}$ for any $r\geq 0$ and $x\in [-1, 1/r)$. \end{proof} 
	 	 Next we apply our new methodology to the cycle to recover the worst-case cover time from the stationary case.

	\begin{proof}[Proof of Theorem \ref{cycleworst}] We can assume that $k\geq C$ for a large fixed constant $C$ (in particular one satisfying $\log C \geq 1$), as otherwise the result holds since $\tcov = \Theta(n^2)$, and by \cite[Theorem 4.2]{ER09} the speed-up of the cover time on any graph is $\BO{k^2}$.
	
	For $k\geq C$, define $\tilde k = \lfloor \log(k) \rfloor\geq 1$. Recall that $t_{\mathsf{mix}}^{(\tilde k,k)}  = \BO{n^2/\log (k/\log k) }  = \BO{n^2/\log k } $ by Lemma~\ref{lem:cyclemixbdd}, and $	t_{\mathsf{cov}}^{\tilde k}(\pi) =\BO{\frac{n^2}{(\log k)^2} (\log \log k)^2}$ by Theorem~\ref{cyclethm}. Thus, Theorem~\ref{thm:characterupper} yields $\tcov^{(k)}=  \BO{ \frac{n^2}{\log k}}.$

	 We will use Theorem~\ref{thm:cyclelower} to prove the lower bound, and for such, we need lower bounds for $t_{\mathsf{large-hit}}^{(\tilde{k},k)}$, and  	$t_{\mathsf{large-cov}}^{(\tilde{k})}$ for an appropriate choice of $\tilde k$.	We will indeed prove that if we choose $\tilde k$ as a constant, then  $t_{\mathsf{large-hit}}^{(\tilde{k},k)} = \Omega( n^2/\log k )$ and $t_{\mathsf{large-cov}}^{(\tilde{k})} = \Omega(n^2)$, leading to the desired result as	$t_{\mathsf{cov}}^{(k)} = \Omega(\min(t_{\mathsf{large-hit}}^{(\tilde{k},k)},t_{\mathsf{large-cov}}^{(\tilde{k})}))$.
		
		For $t_{\mathsf{large-hit}}^{(\tilde{k},k)}$ we note that  for any vertex $u$ of the cycle we can find a set of vertices of size at least $n/2$ with minimum distance at least $ \lfloor n/4\rfloor $ from $ u$. Thus by Lemma~\ref{lem:largehit} for any $\tilde{k} \leq k/2$ we have \begin{equation*} 		t_{\mathsf{large-hit}}^{(\tilde{k},k)}  = \Omega\left( \frac{n^2}{\log( k/\tilde{k}) } \right).
		\end{equation*}
		To find a lower bound for $t_{\mathsf{large-cov}}^{(\tilde{k})}$, by Lemma \ref{torireturns} there exists some constant $c>0$ such that for any $1\leq t\leq n^2$ and $u\in V$ the return probabilities in a cycle satisfy  $\sum_{s=0}^t P_{u,u}^{s} \geq  c\sqrt{t}$. Thus if we take $\tilde{k} \geq \min(100,c/8)$ and let $ t= \lfloor (cn/(256\tilde{k}) )^2 \rfloor$ then $c\sqrt{t}\geq 32\cdot t \cdot \pi(u) \cdot (8\tilde{k})$ is satisfied. Thus by Lemma \ref{lem:obliviouslowerbound} (with $k=8\tilde{k}$) we have 
		$t_{\mathsf{large-cov}}^{(\tilde{k})} \geq t/5  = \Omega(n^2)$.
	\end{proof}

 	\subsection{Complete Binary Tree and 2-Dimensional Torus}
 	
 	In this section we derive the multiple cover times for the Complete Binary Tree and 2-Dimensional Torus. We treat them together as their proofs have several common elements.  Some standard estimates such as return probabilities and other elementary results on trees can be found in Section \ref{sec:treeret}.

	\begin{theorem}\label{thm:tree} Let $G$ be the two-dimensional torus $\mathds{T}_2$ or the complete binary tree $\mathcal{T}_n$. Then there exists a constant $c>0$ (independent of $n$) such that for any $1\leq k\leq cn\log n$, 
		\[\tcov^{(k)}(\pi)  = \BT{\frac{n\log n }{k}\log\left( \frac{n\log n }{k}\right) }.\]  
	\end{theorem}
	
	For worst-case cover time of the binary tree the best previously known bounds differ by multiplicative $\mathsf{poly}(\log n)$ factors \cite{POTC}.  Using our new $\min$-$\max$ and $\max$-$\min$ characterisations, and some additional calculations, we can now determine $\tcov^{(k)}$ up to constants for any $1 \leq k \leq n$. 
	\begin{theorem}\label{treeworst}
		For the complete binary tree $\mathcal{T}_{n}$:
		\begin{equation*}
		\tcov^{(k)}=
 \begin{cases}
		\BT{\dfrac{n}{k}\log^ 2 n}  & \mbox{if $1 \leq k \leq \log^2 n$}, \vspace{.2cm}\\ 
		 	\BT{\dfrac{n}{\sqrt{k}}\log n}   & \mbox{if $\log^2 n \leq k \leq n$.}
		\end{cases}
		\end{equation*}
	\end{theorem}

	The worst-case cover time of the $2$d-torus was shown in \cite{IKPS17}:
	\begin{theorem}[\cite{IKPS17}]\label{2dtreeworst}
		For the two-dimensional torus $\mathds{T}_2$:
		\begin{equation*}
		\tcov^{(k)}=
		\begin{cases}
		\displaystyle{\BT{ \frac{n}{k} \log^ 2 n}} \vspace{.2cm} & \mbox{if $1 \leq k \leq \log^2 n$}, \\
		\displaystyle{\BT{  \frac{n }{\log (k/\log^2 n)} }}   & \mbox{if $\log^2 n \leq k \leq n$.}
		\end{cases}
		\end{equation*}
	\end{theorem}

	Using the tools introduced in Section \ref{Sec:partialmixing} we can recover the upper bounds in Theorem \ref{2dtreeworst} fairly efficiently. However, for the lower bounds in Theorem \ref{2dtreeworst} we did not find a way to apply our (or any other) general techniques to give a tight bound easily for all $k$. The methods presented in this work give a lower bound tight up to a $\log n$ factor however we do not give the details as recovering loose bounds on known quantities is not the goal of this work.

	\subsubsection{Stationary Cover time of the Binary Tree \& \texorpdfstring{$2$d}{2d}-Torus} 
	
	In this section we prove Theorem \ref{thm:tree}. The upper bounds are established by applying Lemma \ref{harmonicreturns} to both graphs.  A matching lower bound is proved by showing that both graphs have a set which is particularly hard to cover.

		\begin{proof}[Proof of Theorem \ref{thm:tree}]

		For the upper bound, in either graph we have, $ \sum_{i=0}^t P_{v,v}^i= \BO{1+\log t}$ for any $v\in V$ and $t\leq \trel$ by Lemmas \ref{treereturnslowerbound}, \ref{treeReturnsUpper} \& \ref{torireturns}, and that  $\tmix =\BO{n}$ for both graphs by \cite[Eq.\ 5.59]{aldousfill} and \cite[Eq.\ 5.6]{levin2009markov}. Thus the upper bound follows directly from Lemma \ref{harmonicreturns}.  
		
		We split the lower bound into three cases depending on the value of $k$. First, set \begin{equation}\label{t*} t^*=\frac{n\log n}{k} \cdot \log\left(\frac{n\log n}{k}\right),\end{equation}  and observe that we aim to prove $\tcov^{(k)}(\pi) =\Omega(t^*)$. We now show this in each case.

	 \noindent\textbf{Case (i)} [$1 \leq k \leq (\log n)^{5/3}$]\textbf{:} First recall the following bound by \cite[Theorem 4.8]{ER09}: \begin{equation}\label{eq:reinbdd}	\tcov \leq k t_{\mathsf{cov}}^{(k)}(\pi) + \BO{k\tmix \log k} +  \BO{k\sqrt{ t_{\mathsf{cov}}^{(k)}(\pi) \tmix } }. \end{equation}For both $\mathds{T}_2$ and $\mathcal{T}_n$ we have $\tmix =\BT{n}$, $\thit =\BT{n\log n}$ and $\tcov =\BT{n\log^2 n}$ by \cite[Section 11.3.2]{levin2009markov} and \cite[Theorem 6.27]{aldousfill}, respectively. Recall also that $\tcov^{(k)}(\pi)\leq \BO{\frac{\thit\log n}{k}} =\BO{\frac{n\log^2 n}{k}}$ by Theorem \ref{thm:statmatthews}, and so plugging these bounds into \eqref{eq:reinbdd} gives
		\begin{equation*}
		\tcov  \leq k t_{\mathsf{cov}}^{(k)}(\pi) + \BO{kn \log\log n} +  \BO{\sqrt{k}n \log n }.
		\end{equation*}Thus for either graph, if $1 \leq k \leq (\log n)^{5/3}$, we have \[t_{\mathsf{cov}}^{(k)}(\pi)\geq \tcov/k -\BO{n\log\log n} - \BO{\frac{n}{\sqrt{k}}\log n} =  \Omega\left( \frac{n\log^2 n}{k}\right) = \Omega(t^*).\]

		\noindent\textbf{Case (ii)} [$(\log n )^{5/3}\leq k\leq n^{1/2}$]\textbf{:} Let $t^*$ be as \eqref{t*} and observe that in this case \begin{equation*}\frac{n(\log n)^2}{2k}\leq \frac{n\log n}{k} \cdot \log\left(n^{1/2} \log n\right)\leq t^*\leq \frac{n\log n}{k} \cdot \log\left(\frac{n}{(\log n)^{2/3}}\right)\leq \frac{n(\log n)^2}{k}. \end{equation*} Let $\widehat {G}:=\widehat {G}(x)$ be the  geometric reset graph from Definition \ref{georesetgraph} where $x=k/(n(\log n)^2)$. We use $\Pru{u,H}{\cdot}$ to denote the law of the (non-lazy) random walk on $H$ started from $u\in V(H)$. For ease of reading, we prove the next claim on Page \pageref{clmproof} after concluding the current proof. 
	
		\begin{claim}\label{CL1}
		Let $G=\mathcal{T}_n$ or $\mathds{T}_2$, $(\log n )^{5/3}\leq k\leq n^{1/2} $,  and $\widehat {G}:=\widehat {G}(x)$, where $x=\kappa\cdot k/(n(\log n)^2)$ for a fixed constant $\kappa>0$. Then, there exists a subset $S\subseteq V$ such that $\log |S| \geq (\log n)/100$ and a constant $\kappa'>0$ (independent of $\kappa$) such that $\Exu{u,\widehat {G}}{\tau_v} \geq  \kappa '\cdot n \cdot \log\left(\frac{n\log n}{k}\right)$ for all $u,v \in S$.
	\end{claim}
	
	In light of Claim \ref{CL1} it follows from \cite[Theorem 1.4]{KKLV} that, for any $x=\kappa\cdot k/(n(\log n)^2)$,  \begin{equation}\label{eq:lowercovhatg}\Exu{\widehat{\pi},\widehat {G}(x)}{\tau_{\cov}}\geq \frac{\log|S|}{2}\cdot \min_{u,v \in S} \Exu{u,\hat{G}}{\tau_v}=\frac{1}{2}\cdot  \frac{\log n}{100}\cdot  \kappa ' n\log\left(\frac{n\log n}{k}\right) \geq  \frac{\kappa ' n (\log n)^2}{400},\end{equation} where we note that although \cite[Theorem 1.4]{KKLV} is stated only for the simple random walk, it holds for all reversible Markov chains, see \cite[Page 4]{CooperSiro}. Thus, by Lemma \ref{lowerbddreversible} and \eqref{eq:lowercovhatg}, there exists some constant $c>0$, independent of $\kappa$, such that for all suitably large $n$, \begin{equation}\label{eq:coupleghat}\Pru{\widehat{\pi},\widehat{G}(x)}{\tau_{\mathsf{cov}} >cn(\log n)^2} >1/3.\end{equation} We now aim to apply Lemma \ref{resetcouple} where to begin with we choose the values $T=cn\log^2 n$ and $C=100$. Observe that $T\geq 5Ck$ for large enough $n$ since $k\leq n^{1/2}$. Finally, since $c$ is independent of $\kappa$, we can set $\kappa = 100/c$ such that $x=\kappa\cdot k/(n(\log n)^2)= Ck/T$. Then, by  Lemma \ref{resetcouple} and \eqref{eq:coupleghat}, 
	\[\Pru{\pi^k,G}{\tau_{\mathsf{cov}}^{(k)} >\frac{cn(\log n)^2}{1000k}}> \Pru{\widehat {\pi},\widehat{G}(x)}{\tau_{\mathsf{cov}} >cn(\log n)^2} -\exp\left(-\frac{100k}{50}\right)\geq \frac{1}{3} - e^{-2}\geq \frac{1}{10}.\]  
	
			\noindent\textbf{Case (iii)} [$ n^{1/2}\leq k \leq cn\log n$]\textbf{:} Let $\delta\in (0,1)$ be a small constant to be chosen later, and let $c = (\delta/2)^2$, and $t^*$ be as given by \eqref{t*}. Let $u$ be a leaf of the tree $\mathcal{T}_n$, or any vertex of $\mathds{T}_2$, then,  $\sum_{i=0}^{\delta t^*} P_{u,u}^i = \BOhm{\log(\delta t^*)}$ by Lemmas \ref{treereturnslowerbound} and \ref{torireturns} respectively. Then, by an application of Lemma \ref{covfromvisits} \eqref{returns} there exists a non-negative constant $C$, such that	\begin{align}
		\Pru{\pi}{ \tau_u \leq\delta\cdot t^* } &\leq \frac{2\delta \cdot (t^*+1) \cdot \pi(u) }{\sum_{s=0}^{\delta t^*/2} P_{u,u}^s} \notag 
		\\ &\leq \frac{C\cdot \delta\log n \cdot \log( (n/k)\cdot \log n)  }{k \cdot \log\left((\delta/2) \cdot (n/k)\log n \cdot \log\left((n/k)\cdot \log n\right)\right) }\nonumber\\
		&=  \frac{C\cdot \delta\log n \cdot \log( (n/k)\cdot \log n))  }{k \cdot \left(\log(\delta/2)+\log((n/k)\cdot  \log n) + \log( \log((n/k)\cdot \log n)) \right)}\nonumber\\
		&\leq \frac{2C\delta \log n\cdot \log((n/k)\cdot \log n)}{k\cdot \log((n/k)\cdot \log n)}\notag 
		\\ &= \frac{2C\delta \log n}{k},\label{eq:treepiprob}
		\end{align}
		where in the last inequality we use that $k\leq c(n\log n) = (\delta/2)^2 (n\log n)$, and thus we have that $\log ((n/k)\cdot \log n)\geq 2\log(\delta/2)$.
		 
		We will now apply Lemma \ref{lemma:generalLowerboundpi}, where for the binary tree $\mathcal{T}_n$ we choose $S$ as the set of all leaves, and for $\mathds{T}_2$ choose $S = V(\mathds{T}_2)$. Thus in either case we have $|S|\min_{v\in S}\pi(v) \geq 1/3$, $p\leq \frac{2C\delta \log n}{k}$ by \eqref{eq:treepiprob} and $2p^2k<1$ since $k \geq n^{1/2}$. 
		Thus by Lemma \ref{lemma:generalLowerboundpi} 
	\[\Pru{\pi^k}{ \tau_{\mathsf{cov}}^{(k)} \leq  t} \leq \frac{4kp^2e^{2kp}}{|S|\min_{v\in S}\pi(v)}  \leq  12k\left(\frac{2C\delta \log n}{k}\right)^2\cdot e^{2 k \cdot \frac{2C\delta \log n}{k} } = \frac{48C^2\delta^2(\log n)^2 }{  k }\cdot n^{4C\delta } = o(1) , \] 
where the last equality follows by taking $\delta = \frac{1}{12C} $ since $k \geq n^{1/2}$.\end{proof}
	It remains to prove Claim \ref{CL1}.

	\begin{pocd}{\ref{CL1}}\label{clmproof}
		
		For $\mathcal{T}_n$ we let $S$ be the set of leaves at pairwise distance at least $(\log n)/2$ and in $\mathds{T}_2$ we take an (almost) evenly spaced sub-lattice where the distance between points next to each other is $\Theta(n^{1/4})$. It is easy to see that one can find such an $S$ of polynomial size, in particular we can take $\log |S|\geq (\log n)/100$.  
	
		We first consider a walk $\hat{W}_t$ in $\hat{G}(x)$ from $\hat{\pi}$ (rather than $u$). Let $N_v(T)$ be the number of visits to $v\in S$ in the interval $[0,T)$, where, for some $\delta >0$, \begin{equation}\label{eq:T} T= \left\lfloor \delta \cdot  n \cdot \log\left(\frac{n\log n}{k}\right) \right\rfloor.\end{equation} Let the random variable $Y$ be the first time that the walk in $\widehat {G}(x)$ started from a vertex in $V(G)$ leaves $V(G)$ to visit $z$. 
		Observe that  $Y \sim \geo{x}$ regardless of the start vertex, and thus $\Pr{Y\geq 1/x} = (1-x)^{1/x}  \geq e^{-1}(1-x^2)\geq e^{-2}$ by \eqref{eq:cheatsheet}. Conditional on $\{Y\geq 1/x\}\cap\{\hat{W}_0\neq z\}$ the walk has the same law as a walk $P$ on $G$ until time $1/x$, as edges not leading to $z$ all have the same weight in $\hat{G}$. 
		Thus, if $Q$ is the transition matrix of the walk on $\hat{G}(x)$ then for any vertex $v\neq z$, we have 
		\[  Q_{v,v}^i \geq e^{-2} \cdot P_{v,v}^i.\] Now, as $\min(T/2-1, 1/x)=1/x$, by Lemma \ref{covfromvisits} \eqref{returns},  \[\Exu{\hat{\pi},\hat{G}}{N_v(T) \mid N_v(T)\geq 1} \geq \frac{1}{2}\cdot \sum_{i=0}^{\min(T/2-1, 1/x)} Q_{v,v}^i \geq \frac{e^{-2}}{2}\cdot   \sum_{i=0}^{1/x} P_{v,v}^i\geq   C\cdot \log(1/x) , \] for some $C>0$ fixed by Lemmas \ref{treereturnslowerbound} \& \ref{torireturns}. 
		Now, for any fixed $\kappa>0$, by Lemma \ref{covfromvisits} \eqref{returns}, \[\Pru{\widehat{\pi},\widehat {G}}{N_v(T)\geq 1} = \frac{\Exu{\hat{\pi},\hat{G}}{N_v(T)}}{\Exu{\hat{\pi},\hat{G}}{N_v(T) \mid N_v(T)\geq 1}} \leq \frac{T/n}{C\cdot \log(1/x) } \leq \frac{\delta \log\left(\frac{n\log n}{k}\right) }{C\log\left(\frac{n(\log n)^2}{\kappa \cdot k}\right)}  \leq \frac{\delta}{C},\] since $\Exu{\hat{\pi},\hat{G}}{N_v(T)}= \hat{\pi}(v)T\leq T/n$. 
		Thus, by taking $\delta= C/2>0$,  we have \begin{equation}\label{eq:hatpihit}\Exu{\widehat \pi,\widehat {G}}{\tau_{v }}\geq \left(1-\Pru{\widehat{\pi},\widehat {G}}{N_v(T)\geq 1} \right)\cdot T \geq \frac{1}{2}\cdot T .\end{equation}

		Recall $Y\sim \geo{x}$ and observe that $1/x\leq n(\log n)^{1/3}/\kappa = o(n\sqrt{n}) $ since $k\geq (\log n)^{5/3}$ and $\kappa>0$ is  fixed. Thus  $\Pru{u,\widehat {G}}{Y> n\sqrt{\log n}  } = (1-x)^{n\sqrt{\log n}} = \lo{1} $. Hence, for any $u,v\in S$,
		\begin{align*} \Pru{u,\widehat {G}}{ \tau_v <Y} &= \sum_{i=0}^{\infty} \Pru{u,\widehat {G}}{\tau_v<i, Y=i } =\lo{1} + \sum_{i=0}^{n\sqrt{\log n }} \Pru{u,\widehat {G}}{\tau_v<i\mid Y=i }\Pru{u,\widehat {G}}{ Y=i }.\end{align*}
		Now, if we condition on the walk in $\widehat {G}$ not taking an edge to $z$ up to time $i$ then, since the weights on edges not going to $z$ are all the same, this has the same law as a  trajectory in $G$ of length $i$. It follows that $\Pru{u,\widehat {G}}{\tau_v<i\mid Y=i }  = \Pru{u,G}{\tau_v<i }$, and so we have 
		\begin{equation}\label{eq:uandhat} \Pru{u,\widehat {G}}{ \tau_v <Y} =\lo{1} + \sum_{i=0}^{n\sqrt{\log n }} \Pru{u,G}{\tau_v<i}\Pru{u,\widehat {G}}{ Y=i }\leq \Pru{u,G}{\tau_v<n\sqrt{\log n }}+\lo{1} .\end{equation}		 
		
		For either graph $G=\mathcal{T}_n,\mathds{T}_2$ we have $\thit(G)=\BO{n\log n} $ and $\Exu{u,G}{\tau_v}= \BOhm{n\log  n}$, where the latter bound is by symmetry as the effective resistance between any two vertices in $S$ is $\BOhm{\log n}$. Thus we can apply Lemma \ref{lowconcfromupconc} to give $\Pru{u,G}{ \tau_v>n\sqrt{\log n} } > 1/3$. Therefore $\Pru{u,\widehat {G}}{ \tau_v>Y}>1/3 -\lo{1} \geq 1/4$ by \eqref{eq:uandhat}. We now extend the distribution $\pi$ on $G$ to $\hat{G}$ by setting $\pi(z)=0$, and observe that $\hat{\pi}(y)\leq \pi(y)$ for all $y\neq z$. A walk on $\hat{G}$ at $z$ moves to a vertex of $ V(\hat{G})$ distributed according to $\pi$ in the next step, thus 
		\[\Exu{\hat{\pi},\widehat {G}}{ \tau_v} = \sum_{y\in V(\hat{G})} \hat{\pi}(y)\Exu{y,\widehat {G}}{ \tau_v}  \leq \sum_{y\neq z} \pi(y)\Exu{y,\widehat {G}}{ \tau_v}  + \hat{\pi}(z) \left(\Exu{\pi,\widehat {G}}{ \tau_v}+1\right)\leq 3 \Exu{\pi,\widehat {G}}{ \tau_v}.\] For any $i\geq 0$,  $\Exu{u,\widehat {G}}{ \tau_v \,\big|\, \tau_v>Y, Y=i} = i+2 +\Exu{\pi,\widehat {G}}{ \tau_v}\geq \Exu{\hat{\pi},\widehat {G}}{ \tau_v}/3$. To see the first equality note that at time $i$ the walk has not yet hit $v$ and takes a step to $z$, thus two steps later the walk is at a vertex sampled from $\pi$. Hence, when $n$ is large, for any $u,v\in S$ by \eqref{eq:hatpihit}
		\[\Exu{u,\widehat {G}}{ \tau_v} \geq \Exu{u,\widehat {G}}{ \tau_v \,\big|\, \tau_v>Y}\cdot \Pru{u,\widehat {G}}{ \tau_v>Y} \geq  \frac{\Exu{\widehat{\pi},\widehat {G}}{ \tau_v}}{3}\cdot \frac{1}{4} \geq \frac{T}{24} . \] The claim follows from  \eqref{eq:T} since $\delta=C/2>0$ and $C>0$ is fixed and independent of $\kappa$.    
	\end{pocd}

	\subsubsection{Worst-Case Binary Tree} 

		The following result, needed for Theorem \ref{treeworst}, gives a bound on the partial mixing time.  
\begin{lemma}\label{lem:treemixing}
		For the complete binary tree $\mathcal{T}_n$ and any $1 \leq \tilde{k} \leq k/100$,  \[t_{\mathsf{mix}}^{(\tilde{k},k)} =\BO{\frac{\tilde{k}}{k} \cdot n + \log n}.\]
	\end{lemma}
	
	\begin{proof} Let $r$ be the root and $s_r(t)$ be the separation distance from $r$. We begin with two claims, which will be verified once we finish the current proof.
		\begin{claim}\label{clm:treesep}There exists some $C<\infty$ such that $s_r(t)\leq n^{-10}$ for any $t\geq t_0=C\log n$.
		\end{claim}
				\begin{claim}\label{claim:worsthyper2}
			 For any $u\in V$, $\frac{50k\log n}{n} \leq \tilde{k}\leq \frac{k}{100}$ and $t_1=\frac{50\tilde{k} n}{k}$ we have $\Pru{u}{\tau_r \leq t_1} \geq 2\tilde k/k$. 
		\end{claim}
	
		From these two claims, we conclude that, for any  pair of vertices $v,u$, and $\frac{50k\log n}{n} \leq \tilde{k}\leq \frac{k}{100}$,
		\begin{align*}
		P_{u,v}^{t_1+t_0} &\geq \sum_{s=0}^{t_1} \Pru{u}{\tau_r =s}\Pru{r}{X_{t_1+t_0-s}=v} \geq \sum_{s=0}^{t_1} \Pru{v}{\tau_r =s} \pi(v)(1-n^{-10}) \nonumber\geq \frac{2\tilde k} {k}\cdot \frac{\pi(v)}{2} = \frac{\tilde k}{k}\pi(v),
		\end{align*}
		thus for $\tilde{k}$ as above the result follows from definition \eqref{def:partialMixingSepnew} of $\tmix(\tilde{k},k)$. For $\tilde{k}\leq \frac{50k\log n}{n}$ the result follows since $\tmix(\tilde{k},k)$ is increasing in $\tilde{k}$.  \end{proof}
	
	\begin{pocd}{\ref{clm:treesep}} 
	Let $Y_t$ be the distance from $X_t$ to the root $r$, where $X_t$ is a lazy random walk starting from $r$. Then $Y_t$ is a biased random walk (towards root (left) w.p. $1/6$, towards leaves (right) w.p. $2/6$ and stay put w.p. $1/2$) on the path ${0,\ldots, h-1}$ with reflecting barriers, and $Y_0 = 0$. Consider $Y'_t$ to be an independent random walk with the same transition matrix as $Y_t$ but starting from distribution $\mu$ that denotes the stationary distribution of the biased walk on the path. To couple the walks we assume that $Y'$ starts at a vertex $i>0$ (or else it has already met $Y$), then both walks move independently unless $Y'$ is next to $Y$ (thus to the right of it). In this case we sample $Y'$ first then if $Y'$ moves left then $Y$ stays put (and they meet) otherwise, $Y$ moves either left, right or stays with probabilities $1/5$, $2/5$ and $2/5$ respectively.  
			
			Now, notice that $Y$ and $Y'$ must have met by the time that $Y$ reaches $h-1$. We can upper bound $\Pr{Y_t\leq h-1}$ by $\Pr{\sum_{i=1}^t Z_i\leq h-1}$, where $Z_i$ are i.i.d. random variable that take value $1$ w.p $2/6$, value $-1$ w.p $1/6$, and value $0$ w.p $1/2$. Since $h = \log_2 n$ and by choosing $t_0 = C\log_2n$ with $C$ large enough, by a simple application of Chernoff's bound, we have that $\Pr{\sum_{i=1}^t Z_i\leq h-1} \leq n^{-12}$. We conclude that the probability that $Y_t$ and $Y'_t$ do not meet in $\BO{\log n}$ steps is $n^{-12}$. By the standard coupling characterisation of the total variation distance (\cite[Proposition 4.7]{levin2009markov}), we have $||\Pru{0}{Y_t=\cdot}-\mu(\cdot)||_{\mathrm{TV}}\leq n^{-12}$, and thus for all $i$, $\Pru{0}{Y_t=i}-\mu(i) \geq -n^{-12}$. By symmetry, for any vertex at height $i$ in the binary tree we have $\Pru{r}{X_t = v} -\mu(i)/
			2^i \geq - n^{-12}$, since $\mu(i)/2^i = \pi(v)$. We conclude that $\Pru{r}{X_t=v}\geq \pi(v)\cdot\left(1-{n^{-10}}\right)$ for any $v\in V(\mathcal{T}_n)$ as claimed.
		\end{pocd}

		\begin{pocd}{\ref{claim:worsthyper2}}We first bound the probability started from $\pi$. By Lemma \ref{rootreturns} we have $\sum_{s=0}^{t_1} P_{r,r}^s \leq 2+8t_1/n$, where $r$ is the root. Therefore, by Lemma~\ref{covfromvisits} \eqref{returns}, we have \[\Pru{\pi}{\tau_r\leq t_1} \geq \frac{t_1}{2n + 8t_1}.\]
			
			Note that the worst-case for our claim is when $v\in \mathcal{L}$ is a leaf, so we assume this. Denote by $\tau_\mathcal{L}$ the first time the random walk hits a leaf, then $\Pru{\pi}{\tau_r<\tau_{\mathcal{L}}} \leq \frac{8\log n}{n} $ by Lemma \ref{rootreturns}, thus
			\begin{align}\label{eq:treeroothitbound}
			\Pru{\pi}{\tau_r\leq t_1} &\leq  \Pru{\pi}{\tau_r\leq t_1, \tau_r>\tau_{\mathcal{L}}}+ \Pru{\pi}{\tau_r<\tau_{\mathcal{L}}}\nonumber\\
			&\leq  \Pru{\pi}{\tau_r\leq t_1, \tau_r>\tau_{\mathcal{L}}, \tau_{\mathcal{L}}\leq c\log n}+\Pru{\pi}{\tau_{\mathcal{L}}>c\log n}+ \frac{8\log n}{n}  \nonumber\\
			&\leq \sum_{s=0}^{\lfloor c\log n \rfloor} \Pru{v}{\tau_r\leq t_1-s}\Pru{\pi}{\tau_{\mathcal{L}} =s}+\frac{9\log n}{n} 
			\end{align}
			the last inequality holds as $\Pru{\pi}{\tau_{\mathcal{L}}>c\log n}\leq (\log n )/n$, for $c$ large by using the Chernoff's bound and the fact the height of a walk on the tree is a biased walk on a path with reflective barriers (as was used in the proof of Claim \ref{clm:treesep}). Thus, by \eqref{eq:treeroothitbound} we have
			\begin{align*}
			\Pru{\pi}{\tau_r\leq t_1} &\leq \Pru{u}{\tau_r\leq t_1} \Pru{\pi}{\tau_{\mathcal{L}} \leq c\log n}+\frac{9\log n}{n}\leq  \Pru{u}{\tau_r\leq t_1}+\frac{9\log n}{n}.
			\end{align*}
			
			Since $(50\log n)/n\leq \tilde{k}/k\leq 1/100  $ we conclude that \[\Pru{u}{\tau_r\leq t_1} \geq		\frac{t_1}{2n + 8t_1} - \frac{9\log n}{n} \geq \frac{1}{3}\cdot  \frac{t_1}{2n + 8t_1}  = \frac{50\tilde{k}/k}{6\left(1 + 4\cdot 50\tilde{k}/k \right) }\geq \frac{2\tilde{k}}{k},\]  as desired.
		\end{pocd}

	We now prove a lemma which may be regarded as a large-hitting time -- since the random walk starts at a leaf in the left-subtree, and the goal is to hit any vertex in the right-subtree. 
	
 	\begin{lemma}\label{lem:treelargehitting}
 		Let  $\mathcal{T}_{n}$ be a complete binary tree rooted at $r$, then for any leaf $\ell$ and any $t \geq 1$
 		\[
 		\Pru{\ell}{ \tau_{r} \leq t } \leq 6t/n.
 		\]
 	\end{lemma}
	
 	\begin{proof} 
 		We first would like to prove that $P_{\ell,r}^t \leq 6/n$ for any $t \geq 1$. However, this follows immediately by reversibility since $P_{r,\ell}^t \leq 3/n$, as a random walk from $r$ will have a uniform probability over all $2^{h}\geq n/3$ leaves by symmetry. Hence, for any $t\leq n$, we have $	\Pru{\ell}{ \tau_{r} \leq t } \leq  \sum_{s=0}^t P_{\ell,r}^s \leq  6t/n$ by Markov's inequality.	\end{proof}
 
	Finally we are ready to prove the worst-case cover time for the binary tree.
	\begin{proof}[Proof of Theorem \ref{treeworst}] The results holds immediately for $k=1$ by known results of cover times of binary trees. We proceed by a case analysis. 
	
		 \noindent\textbf{Case (i)} [$2 \leq k \leq (\log n)^2$]\textbf{:} Choose $\tilde k = \lfloor k/2 \rfloor$, then by Lemma \ref{lem:mixupper} the time for $\tilde k$ walks to mix is bounded by a constant times the single walk mixing time, which is $\Theta(n)$ for the case of the binary tree. Also, we have  $\tcov^{(k)}= \BO{ t_{\mathsf{cov}}^{(\tilde k)}(\pi)}= \BO{ (n/\tilde k) (\log n)^2}$ by Theorem~\ref{thm:tree}. Therefore by Theorem \ref{thm:characterupper} the upper bound follows. For the lower bound we have $\tcov^{(k)}\geq t_{\mathsf{cov}}^{(k)}(\pi)$, thus the first part of the formula has been shown.
		
		 \noindent\textbf{Case (ii)} [$(\log n)^2\leq k \leq n$]\textbf{:} Again, as shown in Theorem~\ref{thm:tree}, for any $1 \leq \tilde{k} \leq k$
		\[
		\tcov^{(\tilde{k})}(\pi) = \BO{\frac{n\log n }{\tilde{k}}\log\left( \frac{n\log n }{\tilde{k}}\right) }.
		\]
		Also by Lemma~\ref{lem:treemixing}, for any $1\leq \tilde{k}\leq k/100$, we have 
		\[
		t_{\mathsf{mix}}^{(\tilde{k},k)}  = \mathcal{O}\!\left( \frac{\tilde{k}}{k} \cdot n + \log n \right).
		\]
		To balance the last two upper bounds, we choose $\tilde{k} = \lfloor(\sqrt{k} \cdot \log n)/100\rfloor$ so that  
		\[
		\max\left\{ \tcov^{(\tilde{k})}(\pi), t_{\mathsf{mix}}^{(\tilde{k},k)}  \right\} = \BO{ \frac{n\log n}{\sqrt{k} }},
		\]
		and the upper bound follows from Theorem~\ref{thm:characterupper}.
		
		To prove a matching lower bound, let $\tilde{k}=\lfloor \sqrt{k} \cdot \log n \rfloor $. Assume that all $k$ random walks start from an arbitrary but fixed leaf from the left subtree of $r$. Let $t =   (n/6) \cdot \tilde{k}/k$. The number of walks that reach the root by time $t$ has binomial distribution $\bin{k}{p}$ with parameters $k$ and $p$ where $p= 		\Pru{l}{ \tau_{r} \leq t } \leq 6t/n = \tilde{k}/k$ by Lemma~\ref{lem:treelargehitting}. Thus the expected number of walks to reach to root within $t$ steps is upper bounded by the integer $\tilde k$ and so is the median. Therefore with probability at least $1/2$, at most $\tilde{k}$ out of the $k$ walks reach the root vertex $r$ by time $t$. Once we have $\tilde{k}$ walks at the root $r$, we consider the problem of covering the right sub-tree of $r$ with root $r_1$ (which has $2^{d-1}-1 = \Theta(n)$ vertices), assuming that $\tilde{k}$ walks start at the root $r_1$ (at step $0$). Since we are looking for a lower bound we can assume that no walks leave the sub-tree and so the problem reduces to  compute a lower bound of the cover time of a binary tree with $\tilde k$ walks from the root of the tree. 
		
		Since the $\tilde k$ walks start from the root, the time it takes to cover the whole set of vertices of the binary tree is lower-bounded by the time it takes to cover the set of leaves (starting from the root). We claim that the previous quantity is, again, lower-bounded by starting the walks from the stationary distribution. To see this, for each walk, independently sample a height $H$ with probability proportional to the sum of the degrees in such a height, then stop the walk when it reaches height $H$ for first time. A simple analysis shows that the distribution of the vertex where the walk stops is the stationary distribution. Note that before stopping the walk cannot have reached a leaf. Hence, we can ignore the time it takes to stop the walks, start all the random walks from the stationary distribution, and, for a lower bound, only consider the expected time to cover the leaves. We can apply the same argument as in the proof of Theorem~\ref{thm:tree} (an application of Lemma \ref{lemma:generalLowerboundpi} with $S$ taken to be a well spaced subset of the leaves) to lower bound the time taken to cover the leaves. Therefore, we conclude that there exists a constant $c>0$ such that the expected time it takes to cover the leaves with $\tilde k$ walks starting from $\pi$ is bounded from below by\[
		c\cdot  \frac{n\log n }{\tilde{k}}\log\left( \frac{n\log n }{\tilde{k}}\right). 
		\]
	Recall that with probability at least $1/2$, at most $\tilde{k}$ walks reach the root by time $(n/6)\cdot \tilde{k}/k $, thus,  \[
		\tcov^{(k)} \geq \frac{1}{2}\cdot \min \left( \frac{n}{6} \cdot \frac{\tilde{k}}{k} ,\; c\cdot\frac{n\log n }{\tilde{k}}\log\left( \frac{n\log n }{\tilde{k}}\right) \right).
		\]
		Since we set $\tilde{k} = \lfloor \sqrt{k} \cdot \log n \rfloor$ earlier, we obtain $\tcov^{(k)} = \Omega\left(\frac{n \log n }{\sqrt{k}}\right)$.
	\end{proof}

	\subsubsection{Worst-Case \texorpdfstring{$2$d}{2d}-Torus} 
	
	Recall that Lemma \ref{lem:cyclemixbdd} in Section \ref{sec:cycle} bounds the partial mixing time of the $d$-dim torus. We now use this and our stationary cover time bounds to prove the upper bound in Theorem \ref{2dtreeworst}.

	\begin{proof}[Proof of the Upper Bound in Theorem \ref{2dtreeworst}] Observe that the case $k=1$ is immediate by known results on the cover time of the torus. Now, by Lemma~\ref{lem:cyclemixbdd}, for any $1\leq \tilde{k}\leq k/2$ we have
		\begin{equation}\label{tmixpagain}
		t_{\mathsf{mix}}^{(\tilde{k},k)} = \BO{ \frac{n}{ \log( k / \tilde{k} )} }.
		\end{equation}
		Further, by Theorem \ref{thm:tree}, we have
		\[
		\tcov^{(\tilde{k})}(\pi) = \BO{ \frac{n \log n}{\tilde{k}} \cdot \log\left( \frac{n \log n}{\tilde{k}} \right) }.
		\]
				 \noindent\textbf{Case (i)} [$2 \leq k \leq 2(\log n)^2$]\textbf{:} Choose $\tilde{k}= \lfloor k/2 \rfloor$, and the bound on $t_{\mathsf{cov}}^{(k)}(\pi)$ dominates, and we obtain $\tcov^{(k)}= \BO{ (n/k) \log^2 n }$,
		by Theorem \ref{thm:characterupper}. The lower bound follows by $\tcov^{(k)}\geq t_{\mathsf{cov}}^{(k)}(\pi)$.

		 \noindent\textbf{Case (ii)} [$  2(\log n)^2\leq n$, upper bound only]\textbf{:} 	We now prove (only) the upper bound in the remaining case $2(\log n)^2 \leq k \leq n$. We choose $
		\tilde{k} = \lfloor (\log n)^2 \cdot \log (\log k/(\log n)^2)\rfloor $ and obtain
		\[
		\tcov^{(\tilde{k})}(\pi) = \BO{ \frac{n}{\log n \cdot \log (k / (\log n)^2)} \cdot \log n } = \BO{  \frac{n}{ \log (k/(\log n)^2)} },
		\]
		which is of the same order as the upper bound of $t_{\mathsf{mix}}^{(\tilde{k},k)}$ in \eqref{tmixpagain}, since $\log( (a/b) \log (a/b) )= \Theta( \log (a/b))$, thus the result follows from Theorem \ref{thm:characterupper}. \end{proof}

	\subsection{Expanders and Preferential Attachment}

	Formally, an expander is a (sequence of) graphs $(G_n)_{n \geq 1}$ such that for all $n \geq 1$: $(i)$ $G_n$ is connected, $(ii)$ $G_n$ has $n$ vertices, and $(iii)$ $\trel(G_n) = 1/(1-\lambda_2)\leq C$ for some constant $C>0$ independent of $n$, where $\lambda_2$ is the second largest eigenvalue of the transition matrix. Equivalently, due to Cheeger's inequality, a graph is an expander if $\inf_n \Phi(G_n) > 0$.
	
	All previous works~\cite{Multi,ER09,ElsSau} on multiple random walks required expanders to be regular (or regular up to constants). Here, we allow a broader class of expanders --- our methods can treat any graph with bounded relaxation time provided it satisfies $\pi_{\min}=\Omega(1/n)$. This class includes some graphs with heavy-tailed degree distributions as long as they have a constant average degree. Such non-regular expanders are quite common, as they include graph models for the internet such as preferential attachment graphs \cite{MPS06}.
	
	\begin{theorem}\label{pro:expander}
		For any expander with $\pi_{\min}=\Omega(1/n)$, for any $1 \leq k \leq n$,
		\[
		\tcov^{(k)}=\BT{\tcov^{k}(\pi)}= \BT{\frac{n}{k} \log n }.
		\]
	\end{theorem}
	\begin{proof}
	Note that the case $k = 1$ follows immediately from known results about cover times in expanders. For $k\geq 2$, consider $\tilde{k} = \lfloor k/2 \rfloor$, and recall that  $t_{\mathsf{mix}}^{(\tilde{k},k)} = \BO{t_{\mathsf{mix}}} = \BO{\log n}$ by Lemma \ref{lem:mixupper}. By hypothesis $m/\dmin =\mathcal{O}(n)$ and since the graph is an expander $\trel=\BO{1}$, hence by Corollary~\ref{relbdd} we have $t_{\mathsf{cov}}^{(\tilde k)}(\pi)=\BO{(n/\tilde k) \log n}$.  Thus by Theorem \ref{thm:characterupper} we have $\tcov^{(k)}= \BO{(n/k) \log n }$, proving the upper bound.
	
	The lower bound follows by Theorem~\ref{generallower} since $\tcov^{(k)}\geq t_{\mathsf{cov}}^{(k)}(\pi) =\Omega((n/k) \log n)$.
	\end{proof}

	\subsection{The Hypercube}

	The hypercube is not covered by the results in the previous section, since it is not an expander. However, we will show that the same bound on stationary cover times holds nevertheless:
	
	\begin{theorem}\label{thm:hypercubecover}
		Let $G$ be the hypercube with $n$ vertices,  then for any $k\geq 1$, \[t_{\mathsf{cov}}^{(k)}(\pi) = \Theta\left(\frac{n}{k}\log n\right).\]
	\end{theorem}

	\begin{proof}We wish to apply Lemma \ref{thm:constreturn}. This is applicable since the hypercube is regular and by \cite{cooper2014note}, we have that for any vertex $v$,  $ \sum_{t=0}^{\trel} P_{vv}^t\leq  2 + \lo{1}$ and $\trel = \BO{\log n}$.
	\end{proof}
	
	We will also derive the result below in a more systematic way than the original proof \cite{ElsSau} using our new characterisations involving partial mixing time and hitting times of large sets.
	\begin{theorem}[{\cite[Theorem 5.4]{ElsSau}}]\label{thm:hypercubeworstcase}
		For the hypercube with $n$ vertices,  \begin{equation*}
		\tcov^{(k)}=
		\begin{cases}
		\displaystyle{\BT{ \frac{n}{k} \log n} }\vspace{.05cm} & \mbox{if $1 \leq k \leq n/\log \log n$}, \\
		\Theta( \log n \log \log n )   & \mbox{if $n / \log \log n \leq k \leq n$.}
		\end{cases}
		\end{equation*}
	\end{theorem}

	In order to prove this theorem, we need to bound $t_{\mathsf{large-hit}}^{\left(\alpha,n\right)}$ from below which is done with the following lemma.  
	\begin{lemma}\label{lem:hypercubelargehit}
		For the hypercube with $n$ vertices and any $1\leq \tilde{k}\leq k$ satisfying $\tilde{k} \geq k\cdot e^{-\sqrt{\log  n}}\geq 1$ we have  \[t_{\mathsf{large-hit}}^{(\tilde{k},k)}\geq \frac{1}{100}\cdot  (\log n) \log \log n .\]
	\end{lemma}

	\begin{proof}[Proof of Lemma \ref{lem:hypercubelargehit}]
		Let $d$ denotes the dimension of the hypercube with $n=2^d$ vertices.
		Fix the vertex $u=0^{d}$ and consider the set $S_u=\{v \in V : d_H(u,v)\leq d/2 \}$, so $|S_u| \leq 3n/4$. Here $d_H$ is the Hamming distance.
		We will estimate the probability of a random walk leaving $S_u$ in $\ell= (1/100) d \log d$ steps. Recall that a lazy random walk on the hypercube can be considered as performing the following two-step process in each round: 1.) Choose one of the $d$ bits uniformly at random, 2.) Independently, set the bit to $\{0,1\}$ uniformly at random.
		
		Let us denote by $C_{t}$ the set of chosen coordinates, and $U_{t}$ the set of unchosen coordinates in any of the first $t\leq \ell$ steps in the process above. Note that the unchosen coordinates are zero, while the chosen coordinates 
		are in $\{0,1\}$ independently and uniformly. By linearity of expectations and since $|U_t|$ is non-increasing in $t$, we have
		\[
		\Ex{ |U_{t}|  }\geq \Ex{ |U_{\ell}|  } = d \cdot  \left( 1- \frac{1}{d} \right) ^{\ell} \geq d\cdot \left(e^{-1}\left( 1 - \frac{1}{d}\right)\right)^{(\log d)/100}\geq  d\cdot \left(e^{-1}/2\right)^{(\log d)/100} \geq d^{0.9},
		\] where the first inequality is by \eqref{eq:cheatsheet}. Using the Method of Bounded Differences~\cite[Theorem 5.3]{DPmobd}, we conclude that for any  $t \leq \ell$,
		\begin{equation}\label{eq:hyp1}
		\Pr{ |U_{t}| \leq d^{0.8} } \leq   \exp\left( - \frac{ 2 (d^{0.9}-d^{0.8})^2 }{t \cdot 1^2} \right) \leq  \exp(- d^{0.7}).
		\end{equation}
		Next consider the sum of the values at the chosen coordinates $C_{t}$ at time $t\leq \ell $, which is given by 
		\[
		Z_t= \sum_{i \in C_{t}} Y_i,
		\]
		where the $Y_i \in \{0,1\}$ are independent and uniform variables, representing the coordinates of the random walk. Note that $\Ex{Z_t\mid |C_{t}|}=|C_{t}|/2$, and so Hoeffding's bound implies
		\begin{equation}\label{eq:hyp2}
		\Pr{  | Z_t - |C_{t}|/2 | \geq d^{0.75} \; \Big| \; |C_{t}| } \leq  2\exp \left(- 2d^{1.5}/d \right) = 2\exp \left(- 2d^{0.5} \right).
		\end{equation}
		Conditional on the events $|U_{t}| \geq d^{0.8}$ and  $Z_t \leq \Ex{Z_t} +2 d^{0.75}$ we have $Z_t\leq (d-d^{0.8})/2+2d^{0.75}<d/2$, and so the random walk is still in the set $S_u$ at step $t\leq \ell$. By the Union bound over the all steps $t=1,\ldots,\ell$, and \eqref{eq:hyp1} and \eqref{eq:hyp2}, for large $d$, these events hold with probability at least 
		\[1-   \ell \cdot \left(   \exp(- d^{0.7})+ 2\exp \left(- 2d^{0.5} \right) \right) \geq 1 - \exp\left(-\sqrt{d} \right)\geq 1 - \exp\left(-\sqrt{\log n} \right) . \] Thus a random walk of length $\ell$ escapes the set $S_u$ with probability at most $e^{-\sqrt{\log  n}}$. Since $S_u$ must be escaped to hit a worst-case set with stationary mass at least $1/4$, it follows that from the definition \eqref{eq:largehit} of $t_{\mathsf{large-hit}}^{(\tilde{k},k)}$ and monotonicity that for any $\tilde{k} \geq k\cdot e^{-\sqrt{\log  n}}\geq 1$ we have $t_{\mathsf{large-hit}}^{(\tilde{k},k)}\geq \ell$. 
	\end{proof}
	We can now apply our characterisation to find the worst-case cover time of the hypercube.  
	\begin{proof}[Proof of Theorem \ref{thm:hypercubeworstcase}] We can assume that $k\geq 2$ by known results for the cover time of the hypercube. Now, observe that Lemma~\ref{lem:mixupper} and \cite[(6.15)]{levin2009markov} yield
		\begin{equation}\label{eq:hypmix}t_{\mathsf{mix}}(\lfloor k/2 \rfloor,k)=\BO{ t_{\mathsf{mix}}} = \BO{\log n \cdot \log \log n}.
		\end{equation}
	 \noindent\textbf{Case (i)} [$2 \leq k \leq n/\log \log n$]\textbf{:} By Theorem~\ref{thm:hypercubecover}, for any $1 \leq \tilde{k} \leq k$ we have 
		\[
		\tcov^{(\tilde{k})}(\pi) = \Theta\left( \frac{n}{\tilde{k}}\cdot \log n\right). 
		\]
		Let $\tilde{k}=k/2$ and then by \eqref{eq:hypmix},  $t_{\mathsf{mix}}(\tilde k,k)$ is always at most $\BO{\tcov^{(\tilde{k})}(\pi)}$. Hence Theorem~\ref{thm:characterupper} implies $\tcov^{(k)}= \BO{(n/k) \log n}$. The lower bound follows by Theorem~\ref{generallower} since \[\tcov^{(k)}\geq t_{\mathsf{cov}}^{(k)}(\pi) =\Omega\left(\frac{n}{k}\cdot \log n\right).\]

		\noindent\textbf{Case (ii)} [$n/ \log \log n \leq k \leq n$]\textbf{:} If we choose $\tilde{k}=\lfloor \frac{n}{2\log \log n} \rfloor$, then by monotonicity
		\[
		\tcov^{(k)}\leq t_{\mathsf{cov}}^{\left(\tilde{k}\right)} = \BO{ \log n \cdot \log \log n}.
		\]
		Also by monotonicity and \eqref{eq:hypmix} we have $t_{\mathsf{mix}}(\tilde{k},k)\leq t_{\mathsf{mix}}(\lfloor k/2 \rfloor ,k) = \BO{\log n \cdot \log \log n} $  thus the results follows from Theorem~\ref{thm:characterupper}. 
		 
		To prove a matching lower bound, recall that Lemma~\ref{lem:hypercubelargehit} states 
		\[
		t_{\mathsf{large-hit}}^{( n \exp(-\sqrt{\log n}), \,n) }\geq \frac{1}{100}\cdot(\log n )\cdot \log \log n.
		\]
		Again, by monotonicity, we can assume $k=n$ and choose $\tilde{k} = \lfloor n \exp( - \sqrt{ \log n})\rfloor \leq k$, giving \[\tcov^{(k)}\geq  \min\left(t_{\mathsf{large-hit}}^{(\tilde{k}, \,n) }, \frac{1}{\tilde{k}\pi_{\min}}\right)   \geq \min \left( \frac{\log n}{100} \cdot \log \log n, \;\exp\left( \sqrt{\log n}\right) \right)\] by an application of the first bound in Theorem~\ref{thm:characterlower}.
	\end{proof}

	\subsection{Higher Dimensional Tori}
	The proof of the stationary cover time of higher dimensional tori is similar to the hypercube.
	
	\begin{theorem}\label{thm:3+toristat}
		For $d$-dimensional torus $\mathds{T}_d$, where $d\geq 3$, and any $1\leq k \leq n$ we have \[t_{\mathsf{cov}}^{(k)}(\pi) = \Theta\left(\frac{n}{k}\log n\right).\]
	\end{theorem}
	\begin{proof} For the $d$-dimensional torus, where $d\geq 3$ we have $ \sum_{t=0}^{\trel} P_{vv}^t= \BO{1}$  by Lemma \ref{torireturns}. Also $\trel = \BO{n^{2/d}}=\lo{n}$,  see \cite[Section 5.2]{aldousfill}. Thus, we can apply Lemma \ref{thm:constreturn}.
	\end{proof}

	Using our machinery, we can recover the following result in full quite easily.
	\begin{theorem}[\cite{IKPS17}]\label{3+toriworst}
		For the $d$-dimensional torus, where $d \geq 3$ is constant:
		\begin{equation*}
		\tcov^{(k)}=
		\begin{cases}
		\displaystyle{\BT{ \frac{n}{k}\cdot  \log n} }\vspace{.1cm}  & \mbox{if $1 \leq k \leq 2n^{1-2/d} \log n$}, \\
		\displaystyle{\Theta\left(  n^{2/d} \cdot \frac{1}{\log (k/(n^{1-2/d} \log n))}  \right) }  & \mbox{if $2n^{1-2/d} \log n < k \leq n$.}
		\end{cases}
		\end{equation*}
	\end{theorem}

	\begin{proof}[Proof of Theorem \ref{3+toriworst}] As before, we can assume $k\geq 2$ by known results for the (single walk) cover time of the $d$-dim torus. 
		By Theorem~\ref{thm:3+toristat} and Lemma \ref{lem:cyclemixbdd}, respectively, we have 
		\begin{equation}\label{eq:mixcov}
		\tcov^{(\tilde{k})}(\pi) = \BT{ (n/\tilde{k}) \log n}, \quad \text{and} \quad t_{\mathsf{mix}}^{(\tilde{k},k)}  = \BO{ n^{2/d} / \log ( k / \tilde{k} )}\quad \text{if $1\leq \tilde{k}\leq k/2$}.
		\end{equation}

	 \noindent\textbf{Case (i)} [$2 \leq k \leq 2 n^{1-2/d} \log n$]\textbf{:}	For the upper bound, we can choose $\tilde{k} = \lfloor k/2 \rfloor$, then by Theorem \ref{thm:characterupper}, the expected stationary cover time by $k$ walks is $O(n/k \cdot \log n)$. To obtain a matching lower bound we can simply use $\tcov^{(k)}\geq t_{\mathsf{cov}}^{(k)}(\pi)$.

	 \noindent\textbf{Case (ii)} [$2n^{1-2/d} \log n < k \leq n$]\textbf{:} Beginning with the upper bound, set
		\[
		\tilde{k} = \left\lfloor n^{1-2/d}\cdot   \log \left(\frac{k}{n^{1-2/d} \log n} \right)\cdot \log  n \right\rfloor \leq \frac{k}{2} .
		\]
 
	Then inserting this value for $\tilde k$ into the bounds from \eqref{eq:mixcov} gives 
		\[
		\tcov^{(\tilde{k})}(\pi) = \BO{ \frac{n^{2/d}}{\log \frac{k}{n^{1-2/d} \log n} } },
		\quad\text{and}\quad
		t_{\mathsf{mix}}^{(\tilde{k},k)}  = 
		\BO{
			\frac{ n^{2/d}}{ \log( \frac{k}{n^{1-2/d} \log n} \cdot \frac{1}{\log (k /(n^{1-2/d} \log n))   })  }
		}.
		\]
		These bounds are both of the same order and so the upper bound follows from  Theorem \ref{thm:characterupper}.

		For the lower bound set $\tilde k = n^{1-2/d}\log n \leq k/2 $. Then since $n^{1/3}< \tilde{k} \leq k$, by Theorem~\ref{thm:characterlower},
		\begin{equation}\label{eq:lowerd3}
		\tcov^{(k)}\geq C\cdot \min( t_{\mathsf{large-hit}}^{(\tilde{k},k)}, (n/\tilde{k}) \log (n) )\geq  \min( t_{\mathsf{large-hit}}^{(\tilde{k},k)}, n^{2/d}), 
		\end{equation}for some constant $C>0$. For $t_{\mathsf{large-hit}}^{(\tilde{k},k)} $, it follows from Lemma~\ref{lem:largehit} where we fix $u$ to be any vertex and $S$ to be the complement of the ball of radius $n^{1/d}/10$ around $u$ that
		\[
		t_{\mathsf{large-hit}}^{(\tilde{k},k)} =  \Omega\left( \frac{n^{2/d}}{ \log( k/ \tilde k ) }  \right)= \Omega\left( \frac{n^{2/d}}{ \log \frac{ k}{ n^{1-2/d}\log n} }  \right), 
		\]
		hence by \eqref{eq:lowerd3} we obtain a matching lower bound.  \end{proof}

	\section{Conclusion \& Open Problems} 
	
	In this work, we derived several new bounds on multiple stationary and worst-case cover times. We also introduced a new quantity called \emph{partial mixing time}, which extends the definition of mixing time from single random walks to multiple random walks. By means of a $\min$-$\max$ characterisation, we proved that the partial mixing time connects the stationary and worst-case cover times, leading to tight lower and upper bounds for many graph classes.
	
	In terms of worst-case bounds, Theorem \ref{nonregbdd} implies that for any regular graph $G$ and any $k\geq 1$, $t_{\mathsf{cov}}^{(k)}(\pi) = \BO{\left(\frac{n}{k}\right)^2\log^2 n }. $ This bound is tight for the cycle when $k$ is polynomial in $n$ but not for smaller $k$. We suspect that for any $k\geq 1$ the cycle is (asymptotically) the worst-case for $t_{\mathsf{cov}}^{(k)}(\pi)$ amongst regular graphs, which suggests $
	t_{\mathsf{cov}}^{(k)}(\pi) = \BO{\left(\frac{n}{k}\right)^2\log^2 k }.
	$
	
	Some of our results have been only proven for the independent stationary case, but it seems plausible they extend to the case where the $k$ random walks start from the \emph{same} vertex. For example, extending the bound 
	$t_{\mathsf{cov}}^{(k)}(\pi) =\Omega( (n/k) \log n)$ to this case would be very interesting.

	Although our $\min$-$\max$ characterisations involving partial mixing time yields tight bounds for many natural graph classes, it would be interesting to establish a general approximation guarantee (or find graph classes that serve as counter-examples). For the former, we believe techniques such as Gaussian Processes and Majorising Measures used in the seminal work of Ding, Lee and Peres \cite{DLP12} could be very useful.

	\section*{Acknowledgements} 
	All three authors were supported by the ERC Starting Grant 679660 (DYNAMIC MARCH). Nicol\'as Rivera was supported by ANID FONDECYT grant number 3210805. John Sylvester was also supported by ESPRC grant number EP/T004878/1 while at the University of Glasgow.
	
We thank Jonathan Hermon for some interesting and useful discussions, and Przemys\l{}aw Gordinowicz for his feedback on an earlier version of this paper.

	\appendix
	
	\section{Appendix: Elementary Results}

	\begin{lemma}\label{lowconcfromupconc}Let $X$ be a non-negative integer random variable such that $\Ex{X}\geq b$ and there exists $c\geq 0$ such that $\Pr{X>\ell c}\leq \Pr{X>c}^\ell$ for all integers $\ell\geq 0$. Then for any $a<c $ \[\Pr{X>a}\geq \frac{b-a}{b+2c}.\]
	\end{lemma}
	\begin{proof}Let $p= \Pr{X> a}\geq \Pr{X> c}$ as $\Pr{X>x}$ is non-increasing in $x$. Now we have 
		\[b\leq a+ \sum_{i=a+1}^{c-1}\Pr{x>i} +c\sum_{\ell=1}^{\infty}\Pr{X>c}^\ell \leq a + p(c-a) + cp/(1-p).    \] This implies $(1-p)b\leq a + 2pc$, rearranging gives the result.
	\end{proof}

	\subsection{Returns in the Torus}

	The following result is well-known, however we state it for completeness.

 	\begin{lemma}\label{torireturns}For the d-dimensional torus $\mathds{T}_d$ on $n$ vertices and any $1\leq t\leq \trel$ we have  
		\[\sum_{i=0}^t P^i_{v,v} = \begin{cases} \Theta(\sqrt{t}) &\text{ if } d=1\\
		\Theta(1+ \log t) &\text{ if } d=2\\
		\Theta(1) &\text{ if } d \geq 3\\
		\end{cases}.\]
 	\end{lemma}
	\begin{proof}We begin with the lower bounds. Let $Q$ and $P$ be the transition matrices of the lazy walk on the $d$-dimensional integer lattice $\mathbb{Z}^d$ and the $d$-dimensional torus $\mathds{T}_d$, respectively. By \cite[Theorem 5.1 (15)]{HSC}, for each $d\geq 1$ there exist constants $C>0$ such that for any $t\geq 1$ and $v \in \mathbb{Z}^d$ we have  $ Q_{v,v}^t \geq  \left(C/t\right)^{d/2}$. The lower bounds for the torus then follow by summation since for any $t\geq 0$ and $v\in V(\mathds{T}_d) $ we have $ P_{v,v}^t \geq  Q_{v,v}^t$. 
	
	We now prove the three cases for the upper bounds separately.

 \noindent\textbf{Case (i)} [$d=1$, cycle]\textbf{:} For a lazy random walk in the cycle it holds for $t\geq 1$ that $P_{v,v}^t \leq \frac{1}{n}+ \frac{c}{\sqrt t}$, for some constant $c>0$ \cite[Theorem 4.9]{LOreturns}. Now, for any $C\geq 1$, and any $1\leq t\leq Cn^2$, it holds for some $c'>0$ that $$\sum_{i=0}^tP_{v,v} \leq c' \sqrt{t}+ C(t+1)/n^2 = O(\sqrt{t}).$$
 
 \noindent\textbf{Case (ii)} [$d=2$]\textbf{:} A random walk of length $t$ in two dimension can be generated as follows; first sample a random integer $x$ according to $B(t) \sim \bin{t}{1/2}$, where $x$ is the number of lazy random walk steps the walk takes in the first dimension (so $t-x$ is the number of lazy random walk steps the walk takes in the second dimension). Hence, if $Q$ denotes the law of a lazy random walk on a cycle with $\sqrt{n}$ vertices, we have $P ^t_{v,v} = \Ex{Q_{v,v}^{B(t)} Q_{v,v}^{t-B(t)}}. $
Observe that for real functions $f$ and $g$, non-increasing and non-decreasing respectively, random variables $B$ and $B'$, where $B'$ is an independent copy of $B$, we have 
\[2\left(\Ex{f(B)g(B)}-\Ex{f(B)}\Ex{g(B)}\right) = \Ex{(f(B)-f(B'))(g(B)-g(B')}\leq 0.\] 
Note that $s\to Q_{v,v}^s$ is non-increasing \cite[Exercise 12.5]{levin2009markov}, so $s\to Q_{v,v}^{t-s}$ is non-decreasing, hence   \[P ^t_{v,v} = \Ex{Q_{v,v}^{B(t)} Q_{v,v}^{t-B(t)}}\leq \Ex{Q_{v,v}^{B(t)}}\cdot \Ex{Q_{v,v}^{t-B(t)}} \leq \Ex{Q_{v,v}^{B(t)}}^2.\]

As $s\to Q_{v,v}^s$ is non-increasing, we have $Q_{v,v}^s \leq Q_{v,v}^{\lfloor t/4 \rfloor} $ for any $s\geq t/4$. Then, for $t\geq 1$ but $t = \mathcal{O}(n)$, it holds $Q_{v,v}^{\lfloor t/4 \rfloor} \leq c_1/\sqrt{t}$ for some constant $c_1>0$ by \cite[Theorem 4.9]{LOreturns}. Thus, by Hoeffding's bound, \[ \Ex{Q_{v,v}^{B(t)}}^2 \leq \frac{c_1^2}{t} +
 \Pr{ \bin{t}{1/2} \leq t/4 }
\leq \frac{ c_1^2}{t} + \exp(-t/2) \leq \frac{  c_2}{t},\] for some constant $c_2>0$. The statement of Case (ii) then follows by summation.  

\comment{ OLD VERSION OF CASE ii:  A random walk of length $t$ in two dimension can be generated as follows; first sample a random integer $x$ according to $\bin{t}{1/2}$, where $x$ is the number of steps the walk takes in the first dimension (so $t-x$ is the number of steps the walk takes in the second dimension). Hence with $Q$ denoting the law of a random walk on a cycle with $\sqrt{n}$ vertices,\begin{align*}
 P_{v,v}^t &= \sum_{x=0}^{t} \Pr{ \bin{t}{1/2} = x} \cdot Q_{v,v}^x \cdot Q_{v,v}^{t-x} \\
 &\leq \sum_{x=t/4}^{3t/4} \Pr{ \bin{t}{1/2} = x} \cdot Q_{v,v}^x \cdot Q_{v,v}^{t-x} 
 + 2 \sum_{x=0}^{t/4-1} \Pr{ \bin{t}{1/2} = x} .\end{align*}
each cycle has length $\sqrt{n}$, and since $t \leq n$, we know from the $d=1$ case (and monotonicity of $Q_{v,v}^x$ in $x$ provided $x \in [t/4,3t/4]$) that $Q_{v,v}^x \cdot Q_{v,v}^{t-x} \leq Q_{v,v}^{t/4} \cdot Q_{v,v}^{t/4} \leq ( c_1 / \sqrt{t/4})^2$ for some constant $c_1 > 0$, so
 \[ P_{v,v}^t \leq \frac{c_1^2}{t/4} +
 \Pr{ |\bin{t}{1/2} - t/2| \geq t/4 }
\leq \frac{4 \cdot c_1^2}{t} + 2 \cdot \exp(-\Omega(t) ) \leq \frac{5 \cdot c_1^2}{t},\]
  where the second inequality follows from the Hoeffding bound. The result for $d=2$ then follows from summation.}
 
 \noindent\textbf{Case (iii)} [$d\geq 3$]\textbf{:} By \cite[Proposition 10.13]{levin2009markov} there is a constant $C>0$ such that $\Exu{u}{\tau_v} \leq C \cdot n$ for all $u,v \in V$. Hence by Markov's inequality, 
\[
 \Pru{u}{ \tau_v \geq 2C \cdot n } \leq 1/2.
\]
Therefore, for $t = 2 C \cdot n$, and by averaging over the start vertex,
\[
 1/2 \leq \Pru{\pi}{ \tau_v \leq t } = \frac{\Exu{\pi}{ N_{v}^t} }{ \Exu{\pi}{ N_{v}^t \, \mid \, N_{v}^t \geq 1} }
 \leq \frac{ t \cdot 1/n}{1/2 \cdot \sum_{i=0}^{t/2} P_{v,v}^i,
}
 \]
and rearranging yields
\[
 \sum_{i=0}^{t} P_{v,v}^i \leq 2 \sum_{i=0}^{t/2} P_{v,v}^i \leq 16 C,
\]
where the first inequality holds by monotonicity of $P_{v,v}^i$ in $i \geq 0$. \end{proof}

	\subsection{Returns in the Binary Tree}\label{sec:treeret}

	To prove the results for the binary tree we must control the return probabilities of single random walks, we gather the results required for this task here. 
	
Recall that if $R(x,y)$ is the effective resistance between $x$ and $y$, (see \cite[Section 9.4]{levin2009markov}), then for any $x,y\in V$ by \cite[Proposition 9.5]{levin2009markov}:
	\begin{equation}\label{resistance}
	\Pru{x}{\tau_y< \tau_x^+} = \frac{1}{\deg(x)R(x,y)}. 
	\end{equation}

		\begin{lemma}\label{rootreturns}
		Let $r$ be the root of a binary tree of height $h\geq 4$, and $\mathcal{L}$ be the set of leaves. Then, for any $T\leq n$, $\sum_{t=0}^T P_{r,r}^t\leq 2 + 6 T/n $. Additionally $\Pru{\pi}{\tau_{r}\leq \tau_{\mathcal{L}}  } \leq (8\log n)/n$.
	\end{lemma}
	\begin{proof}
	    Identify $\mathcal{L}$ as a single vertex and observe that
	$R(r,\mathcal{L}) = \sum_{i=1}^h(1/2)^i = 1- 1/2^{h+1}  $, thus $\Pru{r}{\tau_{\mathcal{L}}< \tau_r^+}\geq 1/2$. Once the walk hits $\mathcal{L}$ equation \eqref{resistance} yields $\Pru{\mathcal{L}}{\tau_r< \tau_\mathcal{L}^+} = \left(2^h(1-1/2^{h+1})\right)^{-1}\leq 3/n$. Let $X_t$ be a random walk on the binary tree starting from vertex $r$, i.e. $X_0 = r$, and let $Z_t = \sum_{s=0}^t \ind{\{X_s=r\}}$ denotes the number of times $X_t$ hits the root up to time $t$. Define the stopping time $L_i$ as the $i$-th time the random walk hits some leaf, i.e.:  $L_1= \min\{t\geq 0: X_t \in \mathcal L\}$, and for $i\geq 2$, $L_i=\min\{t>L_{i-1}: X_t \in \mathcal L\}$. For $i\geq 1$ let $C_{i} = \sum_{t=L_i
	+1}^{L_{i+1}} \ind{\{X_t = r\}}$ denote the number of visits to the root between times $L_i$ and $L_{i+1}$, and also let $C_0 = \sum_{t=0}^{L_1} \ind{\{X_t=r\}}$. Since $L_i-L_{i-1}\geq 1$, then  have that\[
	   \sum_{t=0}^T P_{r,r}^t= \E(Z_T)\leq \E(C_0)+\sum_{i=1}^T \E(C_i).\]
	Now, $\Ex{C_0}\leq 2$ since $\Pru{r}{\tau_{\mathcal{L}}< \tau_r^+}\geq 1/2$ and $\Ex{C_i|C_i\geq 1}= \Ex{C_0}$ as in the interval $C_0$ the walk starts from the root. Also for $i\geq 1$, $\Pr{C_i\geq 1}=\Pru{\mathcal{L}}{\tau_r< \tau_\mathcal{L}^+}\leq 3/n$, and thus $$\Ex{C_i} \leq (3/n)\Ex{C_i|C_i\geq 1}= (3/n)\Ex{C_0} \leq 6/n,$$
	concluding that $\sum_{t=0}^T P_{r,r}^t\leq 2 + 6T/n$.
	
	For the second result, let $v_i$ be a vertex at distance $0\leq i\leq h$ from the leaves. Since the walk moves up the tree with probability $1/3$ and down the tree w.p. $2/3$ we have $\Pru{v_i}{\tau_r< \tau_\mathcal{L}} = \frac{2^i-1}{2^h-1} $, since this is the classical (biased) Gambler's ruin problem \cite[Section 17.3.1]{levin2009markov}. It follows that 
	\[\Pru{\pi}{\tau_r< \tau_\mathcal{L}} =\sum_{i=1}^h \Pru{v_i}{\tau_r< \tau_\mathcal{L}}\cdot 2^{h-i}\pi(v_i)\leq   \sum_{i=1}^h \frac{2^i}{2^h-1}\cdot 2^{h-i}\cdot \frac{3}{2(n-1)} = \frac{3h}{n-1}, \]where the last equality holds as $2^{h+1}-1=n$ (and so $2^{h}=(n-1)/2$). Since $h\geq 4$, the number of vertices is at least 15, and then  $h = \log_2(n+1) -1\leq 2\log n$. Finally, since $n-1\geq (6/8)n$, for $n\geq 8$, it holds that $\Pru{\pi}{\tau_r< \tau_\mathcal{L}} \leq \frac{3h}{(n-1)} \leq 8  \log n/n $. 	\end{proof}

	\begin{lemma}[{\cite[eq 8.21]{JHFrogs2018}}]\label{treereturnslowerbound}
		Let $u$ be any leaf in the binary tree. Then for any $t\geq 1$, $\sum_{i=0}^tP_{u,u}^t = \BT{ 1+\log (t)+t/n}$.
	\end{lemma}

	\begin{lemma}\label{treeReturnsUpper}
		Let $\ell \in V(\mathcal{T}_n)$ be a leaf. Then for any vertex $u\in V(\mathcal{T}_n)$ and $t\geq 0$  we have $\sum_{i=0}^{t}P_{u,u}^i \leq 6\cdot  \sum_{i=0}^{t}P_{\ell,\ell}^i$.
	\end{lemma}
	\begin{proof}
		Clearly, if $u$ is a leaf itself, then there is nothing to prove. Hence assume that $u$ is an internal node. 
		Note that $\sum_{s=0}^{t} P_{u,u}^s$ is the expected number of visits to $u$ of a random walk of length $t$ starting from $u$. Divide the random walk of length $t$ into two epochs, where the second epochs starts as soon as a leaf in the subtree rooted at $u$ is visited. We claim that in the first epoch the expected number of visits to $u$ is constant. To show this we identify all the leaves of the subtree rooted at $u$ as a single vertex $\mathcal L$, then by Equation~\eqref{resistance} we have
		\begin{align*}
		\Pru{u}{\tau_{\mathcal L}<\tau_u^+} = \frac{1}{\deg(u)R(u,{\mathcal L})} \geq \frac{1}{3} 
		\end{align*}
		since $R(u,\mathcal L)=\sum_{i=1}^h(1/2)^i\leq 1$. Therefore, in expectation we need at most $3$ excursions to reach $\mathcal L$, therefore the expected number of visits to $u$ in the first epoch is at most $3$. Then, for any leaf $\ell\in \mathcal{L}$, the expected number of visits to $u$ in the second epoch satisfies
		\[
		\sum_{s=1}^{t} P_{\ell,u}^{s} 
		\leq 3 \sum_{s=1}^{t} P_{u,\ell}^{s} 
		\leq 3 \sum_{s=0}^{t} P_{\ell,\ell}^{s},
		\]
		where the first inequality holds by reversibility and the second inequality 
		holds since the expected number of visits to a vertex is maximised if a random walk starts from that vertex. Adding up the expected number of visits from the two epochs yields the claim.
	\end{proof}

\end{document}